%% file: main.tex
\documentclass[11pt,letter,hidelinks]{article}

\usepackage[margin=1in]{geometry}

\usepackage{settings}
\usepackage{macros}
\usepackage{natbib}

\usepackage{authblk}
\usepackage[misc]{ifsym}

\title{Integer Programming with GCD Constraints}

\author[1,2]{R\'emy Defossez}
\author[3]{Christoph Haase}
\author[1\ \href{mailto:alessio.mansutti@imdea.org}{\Letter}]{Alessio Mansutti}
\author[4]{Guillermo A.~Perez}
\affil[1]{IMDEA Software Institute, Spain}
\affil[2]{\'Ecole Normale sup\'erieure, France}
\affil[3]{University of Oxford, UK}
\affil[4]{University of Antwerp, Belgium}

\date{}

\begin{document}

\maketitle

\begin{abstract}
  \input{abstract}

\end{abstract}

\clearpage

\section{Background and overview of main results}\label{sec:introduction}
\input{sec-introduction}

\input{section-CRT}
\input{sec-divisibility-algo}

\input{sec-ip-gcd}

\clearpage
\appendix 
\input{appendix-sieve}
\input{appendix-crt}
\input{appendix-module-basis}
\input{appendix-bound-pzero}
\input{appendix-local-to-global}

\clearpage
\bibliographystyle{plainnat}
\bibliography{bibliography}

\end{document}

%% file: abstract.tex
\noindent
We study the non-linear extension of integer programming with greatest
common divisor constraints of the form $\gcd(f,g) \sim d$, where $f$
and $g$ are linear polynomials, $d$ is a positive integer, and $\sim$
is a relation among $\leq, =, \neq$ and $\geq$.
We show that the feasibility problem for these systems is in \np,
and that an optimal solution minimizing a linear objective function,
if it exists, has polynomial bit length.
To show these
results, we identify an expressive fragment of the existential theory
of the integers with addition and divisibility that admits solutions
of polynomial bit length. It was shown by Lipshitz [\emph{Trans.\ Am.\
Math.\ Soc.}, 235, pp.~271--283, 1978] that this theory adheres to a
local-to-global principle in the following sense: a formula $\Phi$ is
equi-satisfiable with a formula $\Psi$ in this theory such that $\Psi$
has a solution if and only if $\Psi$ has a solution modulo every prime
$p$. We show that in our fragment, only a polynomial number of primes
of polynomial bit length need to be considered, and that the solutions
modulo prime numbers can be combined to yield a solution to $\Phi$ of
polynomial bit length. As a technical by-product, we establish a
Chinese-remainder-type theorem for systems of congruences and
non-congruences showing that solution sizes do not depend on
the magnitude of the moduli of non-congruences.

%% file: sec-introduction.tex
Integer programming, the problem of finding an (optimal)
solution over the integers to a systems of linear
inequalities $A \cdot \vec x \le \vec b$, is a central
problem computer science and operations research.
Feasibility of its 0-1 variant constituted one of Karp's 21
seminal \np-complete problems~\cite{Karp72}. In the 1970s,
membership of the unrestricted problem in \np was
established independently by Borosh and Treybig~\cite{BT76},
and von zur Gathen and Sieveking~\cite{vzGS78}. To show
membership in \np, both groups of authors established a
small witness property: if an instance of integer
programming is feasible then there is a solution whose bit
length is polynomially bounded in the size of the instance.
Reductions to integer programming have become a standard
tool to show membership of numerous problems in \np. In this
paper, we study a non-linear generalization of integer
programming which additionally allows to constrain the
numerical value of the greatest common divisor (GCD) of two
linear terms.

Throughout this paper, denote by $\RR$ the set of real
numbers, $\ZZ$ the set of integers, $\NN$ the set of
non-negative integers including zero, and $\PP$ the set of
all prime numbers. For $R\subseteq \RR$, denote by $R_+
\coloneqq \{ r \in R : r > 0 \}$. Formally, an instance of
integer programming with GCD constraints (IP-GCD) is a
mathematical program of the following form:
\begin{align*}
  \text{minimize}~~~ & \vec c^\intercal \vec x\\
  \text{subject to}~~~ & A \cdot \vec x \le \vec b\\
                       & \gcd(f_i(\vec{x}),g_i(\vec{x})) \sim_i d_i, & 1 \le i \le k \,,
\end{align*}
where $\vec{c} \in \ZZ^n$, $A\in \ZZ^{m\times n}$, $\vec
b\in \ZZ^m$, $d_i \in \pZZ$, $\vec x=(x_1,\ldots,x_n)$ is a
vector of unknowns, the $f_i$ and $g_i$ are linear
polynomials with integer coefficients, and ${\sim_i} \in
\{{\leq},{=},{\neq},$ ${\geq}\}$.
We call $\vec a \in \ZZ^n$ a solution if setting $\vec x =
\vec a$ respects all constraints; $\vec a$ is an optimal
solution if the value of $\vec c^\intercal \vec a$ is
minimal among all solutions. We will first and foremost
focus on the feasibility problem of IP-GCD and discuss
finding optimal solutions later on in this paper. The main
result of this paper is to establish a small witness
property for IP-GCD and consequently membership of the
problem in \np.%
\begin{theorem}\label{thm:small-model} If an instance of
  IP-GCD is feasible then it has a solution (and an optimal
  solution, if one exists) of polynomial bit length. Hence,
  IP-GCD feasibility is \np-complete.
\end{theorem}
We remark that IP-GCD feasibility is \np-hard even for a
single variable, in contrast to classical integer
programming, which is polynomial-time decidable for any
fixed number of variables~\cite{Lenstra83}. It is shown
in~\cite[Theorem~5.5.7]{BachShallit96} that deciding a
univariate system of non-congruences ${x \not\equiv a_i
\pmod {m_i}}$, $1\le i\le k$, is an \np-hard problem.
Hardness of IP-GCD then follows from observing that a
non-congruence $x \not\equiv a \pmod m$ is equivalent to
$\gcd(x-a,m) \neq m$.

\subsection{The NP upper bound at a
glance}\label{ssec:np-at-a-glance}

Even decidability of the IP-GCD feasibility problem is far
from obvious, but can be approached by observing that
deciding a GCD constraint is a \emph{``Diophantine problem
{`in disguise'}''}~\cite{Koen14}. It follows from B\'ezout's
identity that $\gcd(x,y)=d$ if and only if there are
$a,b,u,v\in \ZZ$ such that ${u\cdot d=x}$, ${v\cdot d=y}$,
and $d=a\cdot x+b\cdot y$. While arbitrary systems of
quadratic Diophantine equations are
undecidable~\cite{Mat70}, we see that the unknowns $a,b,u,v$
are only used to express divisibility properties. Hence,
those equations can equivalently be expressed in the
existential fragment of the first-order theory of the
structure $\Ldiv=(\ZZ,0,1,+,\le,\divides)$, where $m \div n$
holds whenever there exists a unique\footnote{This
definition implies that $0 \div n$ does not hold for any
$n\in \ZZ$, $0$ included. Throughout this paper, we assume
wlog.\ that $f\neq 0$ for any divisibility $f\div g$. For
GCD, we instead use the standard interpretation where
$\gcd(0,n) = n$ for any $n \in \mathbb{N}$; this mismatch
between the interpretation of divisibility and GCD is for
technical convenience only.} integer $q$ such that $n = q
\cdot m$:
\[
  u \cdot d=x \wedge v\cdot d=y \wedge d=a\cdot x+b\cdot y \iff
  \exists s\, \exists t\colon
  d \divides x ~\land~
  d \divides y ~\land~
  x \divides s ~\land~
  y \divides t ~\land~
  d = s + t\,.
\]

The full first-order theory of $\Ldiv$ is easily seen
to be undecidable~\cite{Robinson49}. However, decidability of its
existential fragment was independently shown by
Lipshitz~\cite{Lipshitz78,Lip81} and
Bel'tyukov~\cite{Beltyukov80}, and later also studied by van
den Dries and Wilkie~\cite{DW03}, Lechner et
al.~\cite{LechnerOW15}, and
Starchak~\cite{Starchak21B,Starchak21C}. The precise
complexity of the existential fragment is a long-standing
open problem. It is known to be \np-complete for a fixed
number of variables~\cite{Lip81,LechnerOW15}, and membership
in \nexptime has only more recently been
established~\cite{LechnerOW15}. In particular, the bit
length of smallest solutions can be
exponential~\cite{LechnerOW15}, as demonstrated by the
family of formulae $\Phi_n \coloneqq {x_n>1} \land
\bigwedge_{i=0}^{n-1} {x_i > 1} \land {x_i\divides x_{i+1}}
\land {x_i+1\divides x_{i+1}}$, for which any solution
satisfies $x_n \geq 2^{2^{n}}$. From those results, it is
possible to derive that IP-GCD feasibility is decidable in
\nexptime. However, IP-GCD does not require the full
expressive power of $\Ldiv$. In fact, the first-order theory
of $\Ldiv$ can be seen to be equivalent to the theory of
$(\ZZ,0,1,+,\le,\gcd)$ in which the divisibility predicate
is replaced by a full ternary relation $\gcd(x,y) = z$.  In
contrast, IP-GCD only requires countably many binary
predicates $(\gcd(\cdot, \cdot)=d)_{d\in \pZZ}$ and
$(\gcd(\cdot, \cdot)\ge d)_{d\in \pZZ}$ with the obvious
interpretation. Several expressiveness results concerning
(fragments of) the existential theory of the structure
$(\ZZ,0,1,+,\le,(\gcd(\cdot,\cdot)=d)_{d \in \pZZ})$ have
recently been provided by Starchak~\cite{Starchak21A}. The
question of whether this theory admits solutions of
polynomial bit length is explicitly stated as open
in~\cite{Starchak21A}. \Cref{thm:small-model} answers this
question positively.

Our starting point for establishing \Cref{thm:small-model}
is Lipshitz'~\cite{Lipshitz78,Lip81} decision procedure for
the existential theory of $\Ldiv$ that was later refined by
Lechner et al.~\cite{LechnerOW15}. Given a system of
divisibility constraints $\Phi(\vec x) \coloneqq \bigwedge_{i=1}^m
f_i(\vec x) \div g_i(\vec x)$ for linear polynomials~$f_i$
and $g_i$, Lipshitz' algorithm first computes from $\Phi$ an
equi-satisfiable formula $\Psi$ in so-called
\emph{increasing form}. Informally speaking, $\Psi$ is in
increasing form whenever $\Psi$ is a system of
divisibility constraints augmented with constraints imposing a total
(semantic) ordering on the values of the variables in
$\Psi$, and whenever the largest variable with respect to
that ordering occurring in any non-trivial divisibility~$f
\div g$ implied by $\Psi$ only appears in the right-hand
side~$g$. For instance, the system $x < y \land x + 1 \div y
- 2$ is in increasing form, but adding $x + 1 \div x + y$
results in a non-increasing system, since $x + 1 \div y - 2
\,\land\, x + 1 \div x + y$ implies $x+1 \div x + y - (y -
2)$, i.e., $x+1 \div x + 2$. Such implied divisibilities are
captured in~\cite{LechnerOW15} by the notion of a
\emph{divisibility module} that we later formalize
in~\Cref{sec:intro-local-global}. One conceptual
contribution of this paper is to identify a weaker notion of
formulae in increasing form that is syntactic in nature, as
it does not explicitly enforce a particular ordering among
the variables. Informally speaking, a system of
divisibility constraints $\Psi$ is \emph{$r$-increasing} whenever
there exists a partial order $\incord$ over the free
variables of $\Psi$ whose longest chain is of length at most
$r-1$, and for any non-trivial divisibility $f \div g$
implied by $\Psi$, the set of variables occurring in $f \div
g$ has a $\incord$-maximal variable that only appears in the
right-hand side $g$. Referring to the previous example, we
observe that $x + 1 \div y - 2$ is $2$-increasing, witnessed
by the (total) order $x \incord y$. This concept is
fundamental for establishing~\Cref{thm:small-model}, since,
as we discuss below, for fixed $r$, any satisfiable
$r$-increasing formula $\Psi$ of $\Ldiv$ has a smallest
solution of polynomial bit length, and \Ldiv formulae
resulting from IP-GCD instances are $3$-increasing.

Returning to Lipshitz' approach, the key property of
existential $\Ldiv$ formulae in increasing form is that they
enable appealing to a local-to-global property: Lipshitz
shows that any $\Phi$ in increasing form has a solution over
$\ZZ$ if and only if $\Phi$ has a solution in the $p$-adic
integers $\ZZ_p$ for every prime $p$ belonging to a finite
set of difficult primes $\pzero(\Phi)$, the other primes
being ``easy'' in the sense that a $p$-adic solution for
them always exists and that they do not influence the bit
length of the minimal solution of $\Phi$.  In order to
combine the $p$-adic solutions to an integer solution of
$\Phi$, Lipshitz invokes (a generalized version of) the
Chinese Remainder Theorem (CRT):
\begin{theorem}[CRT]
  \label{thm:crt}
  \label{proposition:standard-crt}
  Let $M=\{m_1,\ldots,m_k\},~b_1,\ldots,b_k\in \ZZ$ be such
  that $m_i$ and $m_j$ are coprime for all $1\le i\neq j\le
  k$. The system of simultaneous congruences $x \equiv b_i
  \bmod m_k$, $1\le i\le k$, has a solution, and all
  solutions lie on the shifted lattice $a+\ZZ \cdot \setprod
  M$ for some $a\in \ZZ$.
\end{theorem}       
\noindent
Here and below, for a finite set $M \subseteq \ZZ$, we
denote by $\setprod M$ the product of all elements in $M$.
It follows that the smallest non-negative solution of the
system of congruences is of polynomial bit length. 
As a key
technical contribution of this paper, required to establish
\Cref{thm:small-model}, we develop the following
Chinese-remainder-style theorem that includes additional
non-congruences and yields a bound for the smallest solution
that is, in certain settings, substantially better than the
one that can be achieved by the CRT. For a finite set $S$,
we write $\card{S}$ for its cardinality.

\begin{restatable}{theorem}{ThmMixedCRT}\label{thm:mixed-crt}
  Let $d\in \pZZ$, $M \subseteq \pZZ$ finite, and
  $Q\subseteq \PP$ be a non-empty finite set of primes such
  that the elements of $M \cup Q$ are pairwise coprime, $M
  \cap Q = \emptyset$, and $\min(Q) > d$.
  Consider the univariate system of simultaneous congruences
  and non-congruences $\mcS$ defined by
  \begin{align*}
    x \equiv b_m & \pmod m & m \in M \phantom{\,.}\\
    x \not\equiv c_{q,i} & \pmod q & q \in Q,~1 \le i \le d\,.
  \end{align*}
  Then, for every $k\in \ZZ$, $\mcS$ has a solution in the
  interval $\left\{k,\ldots, k+ \setprod M \cdot \ecrtf(Q,d)
  \right\}$, where
  \begin{center}
    \vspace{-3pt}
  $
    \ecrtf(Q,d) \coloneqq \big((d+1) \cdot \card{Q}\big)^{4 (d+1)^2 (3 + \ln \ln (\card{Q}+1))}.
  $
  \end{center}
\end{restatable}
\noindent
The strength of \Cref{thm:mixed-crt} can be seen as follows.
While it is possible to deduce from the classical CRT that
the solutions of $\mcS$ are periodic with period $\setprod
Q\cdot \setprod M$, we have ${\setprod Q \gg \ecrtf(Q,d)}$
as the magnitude of the primes in $Q$ grows, as in
particular $\ecrtf(Q,d)$ only depends on $\card Q$ and $d$.
We further discuss some results used to establish
\Cref{thm:mixed-crt} in \Cref{sec:intro-crt} below.

Another key technical contribution towards establishing
\Cref{thm:small-model} is to propose a refinement of the set
of difficult primes~$\pzero(\Phi)$. The definition of this
set was changed from~\cite{Lipshitz78} to~\cite{LechnerOW15}
to decrease its bit length from doubly to singly
exponential. We refine the definition once more, and show
that we obtain a set of polynomially many primes of
polynomial bit length.
This result is achieved by an in-depth analysis of how the
integer solution for $\Phi$ is constructed starting from the
\mbox{$p$-adic} solutions. The bound on $\pzero(\Phi)$ also
enables us to derive an \np algorithm for increasing
formulae. It is shown in~\cite{GuepinHW19} that, for every
prime $p\in \PP$, the existential theory of the $p$-adic
integers with linear $p$-adic valuation constraints is
decidable in \np. Deciding an increasing $\Phi$ thus reduces
to a polynomial number of independent queries to an \np
algorithm and is hence in \np. It is worth mentioning that
the family of formulae $\Phi_n$ above is increasing only for
the ordering $x_1 \incord x_2 \incord \dots \incord x_n$
(i.e., it is $n$-increasing but not $(n-1)$-increasing).
Hence, even though the smallest solution of $\Phi_n$ has
exponential bit length, our bound on $\pzero(\Phi)$ enables
us to witness the \emph{existence} of a solution in~\np.

Moreover, this bound leads to a further main result of this
paper, showing that we can construct an integer solution for
$\Phi$ from the relevant $p$-adic solutions that is
asymptotically smaller when compared to the existing
local-to-global approaches~\cite{Lipshitz78,LechnerOW15}.
These improved bounds also crucially rely
on~\Cref{thm:mixed-crt}. To formally state this result, we
require some further definitions. Given $\vec v\in \ZZ^d$,
denote by $\norm{\vec v}$ the maximum absolute value of the
components of $\vec v$, and by $\bitlength{\cdot}$ the bit
length encoding an object under some reasonable standard
encoding in which numbers are encoded in binary.
Furthermore, for a system of divisibility constraints
$\Phi\coloneqq\bigwedge_{i=1}^m f_i \div g_i$, denote by
$\pdiff(\Phi)$ the set of all primes that are less or equal
than $m$ or that divide some number occurring in $\Phi$. 
For $p\in \PP$ and $a\in \ZZ \setminus \{ 0\}$, we write
$v_p(a)$ for the largest $k\in \NN$ such that $a = p^k b$
for some $b\in \ZZ$, and $v_p(0)\coloneqq \infty$. We say
that $\Phi$ has a solution modulo $p$ if 
there is some $\vec b_p \in \ZZ^d$ such that $f_i(\vec
b_p)\neq 0$ and $v_p(f_i(\vec b_p)) \le v_p(g_i(\vec b_p))$
for all $1\le i\le m$. Note that every integer solution is a
solution modulo $p$ for all $p \in \PP$, and therefore if
$\Phi$ does not have a solution modulo some prime~$p$, then
$\Phi$ is unsatisfiable over~$\ZZ$.
The following theorem now gives
bounds on the bit length of an integer solution of $\Phi$ in
terms of solutions modulo $p$ for primes in $\pdiff(\Phi)$.
\begin{restatable}{theorem}{TheoremLocalToGlobal}\label{theorem:local-to-global}
  Let $\Phi(\vec x)$ be an $r$-increasing system of
  divisibility constraints 
  such that $\Phi$ has a solution $\vec b_p \in \ZZ^d$
  modulo~$p$ for every prime $p \in \pdiff(\Phi)$. Then
  $\Phi$ has infinitely many solutions, and a solution $\vec
  a \in \NN^d$ such that $\maxbl{\vec a}\le
  (\bitlength{\Phi} + \max\{ \maxbl{\vec b_p} : p \in
  \pdiff(\Phi) \})^{O(r)}$.
\end{restatable}
\noindent
The bound achieved in \Cref{theorem:local-to-global}
primarily improves upon existing upper bounds by being
exponential only in $r$, as opposed to exponential in
$\poly{d}$ as established in~\cite{LechnerOW15}, where $d$
is the number of variables of $\Phi$. In particular, for $r$
fixed, as is the case for systems of divisibility constraints resulting from IP-GCD systems,
\Cref{theorem:local-to-global} yields small solutions of
polynomial bit length. Observe
that~\Cref{theorem:local-to-global} does not explicitly
invoke the set of difficult primes $\pzero(\Phi)$, but
rather the set $\pdiff(\Phi)$. The latter is the subset of
those primes $p$ in $\pzero(\Phi)$ for which solutions
modulo $p$ might not exist, and one of the initial steps in
the proof~\Cref{theorem:local-to-global} is to compute
solutions modulo $q$ for every prime $q \in \pzero(\Phi)
\setminus \pdiff(\Phi)$. We give further details on the
proof of~\Cref{theorem:local-to-global} in
\Cref{sec:intro-local-global} and then outline in
\Cref{sec:intro-gcd-ip} how it can be used to obtain the NP
upper bound for \Cref{thm:small-model}. But first, we
continue with the promised discussion on some details on
\Cref{thm:mixed-crt}.

\subsection{Small solutions to systems of congruences and
non-congruences}\label{sec:intro-crt}

Let us introduce some notation. Given
$a,b \in \ZZ$, we define $[a,b] \coloneqq
\{a,a+1,\dots,b\}$. We write $\divisors(a) \subseteq \NN$
for the (positive) divisors of $a$ and $\PP(a)$ for $\PP \cap \divisors(a)$.
A function $m \colon \pZZ \to \RR_+$ is
\emph{multiplicative} if $m(a \cdot b) = m(a) \cdot
m(b)$ for all $a,b \in \NN$ coprime (so, $m(1) = 1$). 

The proof of~\Cref{thm:mixed-crt} is based on an abstract
version of Brun's pure sieve~\cite{Brun1915}. Similarly to
other results in sieve theory, Brun's pure sieve considers a
finite set $A \subseteq \ZZ$ and a finite set of primes~$Q$,
and (subject to some conditions) derives bounds on the
cardinality of the set $A \setminus \bigcup_{q \in Q} A_q$,
where $A_q$ is the subset of the elements in $A$ that are
divisible by $q$. In other words, the sieve studies the
number of $x \in A$ satisfying $x \not\equiv 0 \pmod q$ for
every $q \in Q$. In comparison, \Cref{thm:mixed-crt}
requires $x$ to be non-congruent modulo~$q$ to multiple
integers, instead of non-congruent to just~$0$. The key
insight in overcoming this difference is to notice that
Brun's result can be established for arbitrary sets $A_q$,
as long as a simple \emph{independence} property holds
together with Brun's \emph{density} property (a formal
statement is given below). 
%
A second technical issue concerns the bounds obtained from
Brun's sieve. In its standard formulation (see
e.g.~\cite[Ch.~6]{CojocaruM05}), given an arbitrary $u \in
\pZZ$, the sieve gives an estimate on the cardinality of the
set $A \setminus \bigcup_{q \in Q \cap [2,u]} A_q$ that
depends on $u$; and to estimate $\card{\big(A \setminus
\bigcup_{q \in Q} A_q \big)}$ one sets $u$ as the largest
prime in $Q$. The resulting bound is, however, inapplicable
in our setting as we seek to be independent of the bit
length of the primes in~$Q$. This issue is overcome by
revisiting the analysis of Brun's pure sieve
from~\cite{CojocaruM05}, and by requiring an additional
hypothesis: the multiplicative function $m \colon \pZZ \to
\RR_+$ used to express Brun's \emph{density} property must
satisfy $m(q) \leq q - 1$ for~all~$q \in Q$. Those insights
and requirements lead us to the following sieve.

\begin{restatable}{lemma}{LemExtendedBrun}
  \label{lem:extended-brun}
  Let $A \subseteq \ZZ$ and $Q \subseteq \PP$ be non-empty
  finite sets, and let $n \coloneqq \setprod Q$ and $d \in
  \pZZ$. Consider a multiplicative function $m \colon \pZZ
  \to \RR_+$ satisfying $m(q) \leq q-1$ on all $q \in Q$,
  and an (error) function $\sigma \colon \NN \to \RR$. Let
  $(A_r)_{r \in \divisors(n)}$ be a family of subsets of $A$
  satisfying the following two properties:
  \begin{description}
    \item[independence:] $A_{r \cdot s} = A_r \cap A_{s}$,
    for every $r,s \in \divisors(n)$ coprime, and $A_1 = A$; 
    \item[density:] $\card{A_r} = \card{A} \cdot
    \frac{m(r)}{r} + \sigma(r)$, for every $r \in
    \divisors(n)$.
  \end{description}
  Assume $\abs{\sigma(r)} \leq m(r)$,  and $m(q) \leq d$,
  for every $r \in \divisors(n)$ and $q \in Q$. Then,
  \[
    \frac{1}{2} \cdot \card{A} \cdot W_m(Q) - \brunf(Q,d) \, \leq \, \card{\Big(A
    \setminus \bigcup\nolimits_{q \in Q} A_q\Big)} \, \leq \, \frac{3}{2} \cdot
    \card{A} \cdot W_m(Q) + \brunf(Q,d),
  \]
  where $W_m(Q) \coloneqq \prod_{q \in
    Q}\left(1- \frac{m(q)}{q}
  \right)$ and $\brunf(Q,d)
  \coloneqq ( d \cdot \card{Q})^{4  (d+1)^2  (2 + \ln
  \ln(\card Q+1))+2}$.
\end{restatable}

Note that setting $A_r = \{ a \in A : r \div a \}$ for every
${r \in \divisors(n)}$, as usually done in sieve theory,  
results in a family of subsets of $A$ satisfying the
\emph{independence} property. We defer the proof
of~\Cref{lem:extended-brun} and only
sketch
here how to establish~\Cref{thm:mixed-crt}.
Both proofs are given in full details in~\Cref{section:CRT}.

\paragraph*{Proof sketch of~\Cref{thm:mixed-crt}.}
  Below, the set of primes $Q$ and $d \in \pZZ$ 
  defined in the
  statement of~\Cref{thm:mixed-crt} coincide with their
  homonyms in~\Cref{lem:extended-brun}.
  %
  Let $n \coloneqq \setprod Q$. By the CRT, the system of
  congruences ${\forall m \in M}$, ${x \equiv b_m \pmod m}$
  has a solution set $S_M$ that is a shifted lattice with
  period $\setprod M$. Fix some $k\in \ZZ$. 
  We consider the parametric set~${B(z) \coloneqq [k,k+z] \cap S_M}$, and
  find a small value for $z \in \NN$ ensuring that~$B(z)$
  contains at least one solution to~$\mcS$. To do so we rely
  on~\Cref{lem:extended-brun}: we set $A \coloneqq B(z)$,
  and for every~${q \in Q}$, define ${A_q \coloneqq \{ a \in A :
  \text{there is } {i \in [1,d]} \text{ s.t.}~a \equiv
  c_{q,i} \pmod  q \}}$. By definition, the sieved set $A \setminus
  \bigcup_{q \in Q} A_q$ corresponds to the set of solutions
  of $\mcS$ that belong in $[k,k+z]$. The definition of
  $A_q$ is extended to every $r \in \divisors(n)$ not prime as
  $A_r \coloneqq A \cap \bigcap_{q \in \PP(r)} A_q$. We
  establish that these sets satisfy the \mbox{\emph{independence}} and
  \emph{density} properties~of~\Cref{lem:extended-brun},
  subject to the following multiplicative function: ${m(r)
  \coloneqq \prod_{q \in \PP(r)} \#\{ c_{q,i} \mod q : i \in
  [1,d] \}}$, i.e., $m(r)$ is the product of the number of
  distinct values $(c_{q,i} \mod q)$, for every $q \in
  \PP(r)$. By hypothesis $\min(Q) > d$, hence $m(q) \leq d
  \leq q-1$ for every~${q \in Q}$. Furthermore, we show that
  $m$ and the error function $\sigma(r) \coloneqq \card{A_r}
  - \card{A} \cdot \frac{m(r)}{r}$ satisfy the assumption
  $\abs{\sigma(r)} \leq m(r)$, for all $r \in \divisors(n)$.
  Hence, by~\Cref{lem:extended-brun}, we obtain a lower
  bound on the sieved set $A \setminus \bigcup_{q \in Q}
  A_q$. Lastly, we show that taking $z = \ecrtf(Q,d)$ makes
  the lower bound strictly positive, concluding the proof.

\subsection{Small solutions to $r$-increasing systems of divisibility constraints}
\label{sec:intro-local-global}

We now provide an overview on the technical machinery underlying
\Cref{theorem:local-to-global}. Our main goal here is to formalize 
the notion of difficult primes $\pzero(\Phi)$ and 
to sketch the proof of \Cref{theorem:local-to-global}. 
The full proof is given in~\Cref{sec:divisibility-algo}.
We first need several key definitions and auxiliary notation. Subsequently,
$\ZZ[x_1,\dots,x_d]$ denotes the set of \emph{linear}
polynomials $f(x_1,\dots,x_d) = {a_1\cdot x_1 + \dots + a_d
\cdot x_d + c}$, often written as $f(\vec x) = \vec
a^\intercal \vec x + c$; when clear from the context, we
omit the vector of variables $\vec x$ and write $f$ instead
of $f(\vec x)$. The integers $a_1,\dots,a_d$ are the \emph{coefficients} of $f$, 
$c$ is its \emph{constant}.
A polynomial~$f$ is
\emph{primitive} if it is non-zero and $\gcd(f) = 1$,
where $\gcd(f) \coloneqq \gcd(a_1,\dots,a_d,c)$. For any
$b\in \ZZ$, we write $b \cdot f\coloneqq b\cdot \vec
a^\intercal \vec x + b\cdot c$, and $\ZZ f \coloneqq \{ b
\cdot f : b \in \ZZ \}$. The \emph{primitive part} of a
polynomial $g$ is the unique primitive polynomial $f$ such
that $g = \gcd(g) \cdot f$. Let~$\Phi(\vec x)
\coloneqq \bigwedge_{i = 1}^m f_i(\vec x) \div g_i(\vec x)$ be a system of \emph{divisibility constraints}. 
We let
${\terms(\Phi) \coloneqq \{f_i,g_i : 1 \le i \le m \}}$, and,
given a finite sequence $\{(n_i,x_i)\}_{i \in I}$ of
integer-variable pairs, write $\Phi\substitute{n_i}{x_i :
i \in I}$ for the system obtained from $\Phi$ by evaluating
$x_i$ as $n_i$, for all~$i \in I$.

\paragraph*{Divisibility modules and $r$-increasing form.}
As stated in~\Cref{ssec:np-at-a-glance}, when dealing with a
system of divisibility constraints~$\Phi(\vec x)$ one has to
consider all divisibility constraints that are implied by~$\Phi$. This
is done by relying on the notion of divisibility module. The
\emph{divisibility module} of a primitive polynomial $f$
with respect to~$\Phi$, denoted by $\module_f(\Phi)$, is the
smallest set such that
\begin{enumerate*}[(i)]
  \item\label{divmod:prop1}$f \in \module_f(\Phi)$;
  \item\label{divmod:prop2}$\module_f(\Phi)$ is a $\ZZ$-module, i.e., $\module_f(\Phi)$ is closed
    under integer linear combinations; and
  \item\label{divmod:prop3}if $g \div h$ is a divisibility constraint in $\Phi$ and $b \cdot g \in
    \module_f(\Phi)$ for some $b \in \ZZ$, then $b \cdot h \in
    \module_f(\Phi)$.
\end{enumerate*}
The following property holds: for every $g \in \module_f(\Phi)$ and
solution $\vec a$ to $\Phi$, the integer $f(\vec a)$ divides
$g(\vec a)$. The divisibility module
$\module_f(\Phi)$ is a vector subspace, hence it is spanned
by linear polynomials $h_1,\dots,h_\ell \in
\ZZ[x_1,\ldots,x_d]$, that is $\module_f(\Phi) = \ZZ h_1 +
\dots + \ZZ h_\ell$; where $+$ is the Minkowski sum. 

We can now formalize
the key concept of $r$-increasing formula. Let $\incord$ be a
syntactic order on variables $\vec x = (x_1,\ldots,x_d)$.
Given $f \in \ZZ[x_1,\ldots,x_d]$, we write
$\lv_{\incord}(f)$ for the \emph{leading variable} of $f$, that is the variable with non-zero
coefficient in $f$ that is maximal wrt.~$\incord$; if $f$ is constant then
$\lv_{\incord}(f)\coloneqq \bot$, and we postulate $\bot
\incord x_i$ for all $1\le i\le d$. We omit the
subscript $\incord$ when it is clear from the context. A
system of divisibility constraints $\Phi$ is in
\emph{increasing form (wrt.~$\incord$)} whenever
${\module_f(\Phi) \cap \ZZ[x_1,\dots,x_k] = \ZZ f}$ for every
primitive polynomial $f$ with $\lv(f) = x_k$, for every
$1\le k \le d$.  Given a partition $X_1,\dots,X_r$ of the
variables~$\vec x$, we write $(X_1 \incord \dots \incord
X_r)$ for the set of all orders~$\incord$ on~$\vec x$ with the property that for any two~$x,x'$, if $x\in X_i$ and $x'\in X_j$ for some $i<j$ then $x
\incord x'$.
\begin{restatable}{definition}{DefRIncreasing}
  A system of divisibility constraints $\Phi(\vec x)$ is
  \emph{$r$-increasing} if there exists a partition
  $X_1,\ldots, X_r$ of $\vec x$ such that $\Phi$ is in
  increasing form wrt.\ every ordering $\incord$ in $(X_1
  \incord \dots \incord X_r)$.
\end{restatable}    
\noindent
Observe that for any $\incord$ from $(X_1\incord \dots
\incord X_r)$, we have that for every primitive linear
polynomial~$f$ and~$g \in \module_f(\Phi)$, if $g \not\in
\ZZ f$ then $\lv_{\incord}(f) \in X_i$ and $\lv_{\incord}(g)
\in X_j$ for some $i < j$.


\paragraph*{The elimination property and $S$-terms.}
To handle systems in increasing form,
two more concepts are required in the context of the
local-to-global property. First, to compute the ``global'' integer
solution starting from the ``local'' solutions modulo
primes, the divisibility modules of all primitive parts of
polynomials in a system of divisibility
constraints~$\Phi$ need to be taken into account. One way
to do this, introduced in~\cite{LechnerOW15}, is to add
bases for these modules directly to~$\Phi$. This leads to
the notion of elimination property: $\Phi(\vec
x)$ has the \emph{elimination property} for the order ${x_1
\incord \dots \incord x_d}$ of the variables in~$\vec x$
whenever for every primitive part $f$ of a polynomial
appearing in the left-hand side of some divisibility in
$\Phi$, and for every $0\le k\le d$, $\{ g : {\lv(g)
\incordeq x_k} \text{ and } {f \div g} \text{ appears in }
\Phi \}$ is a set of linearly independent polynomials that
forms a basis for $\module_f(\Phi) \cap \ZZ[x_1,\dots,x_k]$,
where $x_0 \coloneqq \bot$. We show that closing a formula 
under the elimination property can be done in polynomial time.
\begin{restatable}{lemma}{LemmaAddEliminationPropertySummary}
  \label{lemma:add-elimination-property-summary}
  There is a polynomial-time algorithm that, given a system
  of divisibility constraints $\Phi(\vec x) \coloneqq
  \bigwedge_{i=1}^m f_i \div g_i$ and an order $x_1 \incord
  \dots \incord x_d$ for $\vec x$, computes $\Psi(\vec x)
  \coloneqq \bigwedge_{i=1}^n f_i' \div g_i'$ with the
  elimination property for~$\incord$ that is equivalent to
  $\Phi(\vec x)$, both over~$\ZZ$ and modulo each $p \in
  \PP$.
\end{restatable}

\noindent
In a nutshell, for every primitive part $f$ of a polynomial
appearing in the left-hand side of a divisibility in~$\Phi$,
the algorithm first computes a finite set $S$ spanning
$\module_f(\Phi)$. 
The algorithm then uses the 
Hermite normal form of a matrix, whose entries are the coefficients and
constant of the elements of $S$, 
to obtain linearly
independent polynomials $h_1,\dots,h_\ell$ with different
leading variables with respect to $\incord$. The system
$\Psi$ is then obtained by replacing divisibility constraints of the
form $f \div g$ appearing in~$\Phi$ with the divisibilities
$f \div h_1,\dots,f \div h_\ell$. Full details are given
in~\Cref{appendix:module-basis}.

The second concept is related to
how~\Cref{theorem:local-to-global} is proven. In a nutshell,
in the proof we iteratively assign values to the variables
in a way that guarantees the system of divisibility
constraints to stay in increasing form. To do that,
additional polynomials need to be considered. For an
example, consider the following system of divisibility
constraints~$\Phi$ in increasing form for the order $u
\incord v \incord x \incord y \incord z$, and with the
elimination property for that order: 
\begin{center}
    $\Phi \ \coloneqq \ v \div u + x + y \ \land \
    v \div x \ \land \
    y+2 \div z+1 \ \land \
    v \div z\,.$
\end{center}
From the first two divisibility constraints, we have $(u+y) \in
\module_v(\Phi)$; i.e., $(u-2)+(y+2) \in \module_v(\Phi)$.
Therefore, if $u$ were to be instantiated as $2$, the
resulting formula $\Phi'$ would satisfy $(y+2) \in
\module_v(\Phi')$ and hence $(z+1) \in \module_v(\Phi')$,
from the third divisibility constraint. Then, $1 \in \module_v(\Phi')$
would follow from the last divisibility, violating the
constraints of the increasing form. The reason why
increasingness is lost when setting $u = 2$ stems from the
fact that in $\Phi'$ we have an implied divisibility $v \div
y + 2$, where $y+2$ is a left-hand side that was not present
in $\module_v(\Phi)$. We can avoid this problem by
considering the polynomial $u - 2$ and forcing it to be
non-zero. The main issue is then to identify all such problematic
polynomials, which is done with the following notion of $S$-terms. Less refined versions of
this notion, as considered in~\cite{Lipshitz78,LechnerOW15}, result in
exponentially larger sets of polynomials.

Given polynomials $f(\vec x)$ and $g(\vec x)$ with
$\lv(f) = x_l$ and $\lv(g) = x_k$, we define their
\mbox{\emph{S-polynomial}} $S(f, g) \coloneqq b_{k} \cdot f - a_l
\cdot g$, where $a_l$ and $b_k$ are coefficients of $x_l$ in
$f$ and $x_k$ in $g$, respectively. For constant $f$
(resp.~$g$), i.e.,~$\lv(f) = \bot$, above $a_l \coloneqq f$
(resp.~$b_k \coloneqq g)$. Note that if $f$ and $g$ are
non-constant and $\lv(f) = \lv(g)$ then $\lv(S(f,g)) \incord
\lv(f)$. For any $X \subseteq \ZZ[x_1,\dots,x_n]$, we define
$S(X) \coloneqq X \cup \{S(f,g) : f,g \in X\}$.
Given a system of divisibility constraints $\Phi$ with the
elimination property for $\incord$ and a primitive
polynomial~$f$, we define the set of \emph{$S$-terms for $f$}, denoted as
$\sfterms{f}{\Phi}$, to be the smallest set such that
\begin{enumerate*}[(i)]
    \item $\terms(\Phi) \subseteq \sfterms{f}{\Phi}$, and
    \item if $f \mid g$ occurs in $\Phi$ and $h \in
    \sfterms{f}{\Phi}$ with $\lv(g) = \lv(h)$, then $S(g,h)
    \in \sfterms{f}{\Phi}$. 
\end{enumerate*}
We write $\sterms(\Phi)$ for the set of all $S$-terms for~$f$, 
where $f$ is any primitive part of a
polynomial in $\terms(\Phi)$.

\paragraph*{The set of difficult primes.}
We now turn towards 
identifying a 
small set of
difficult primes~$\pzero(\Phi)$ of polynomial bit length.
There are two categories of difficult primes: those
for which a solution to $\Phi$ modulo $p$ is not guaranteed
to exist, and those for which such a solution 
always exists, but which still influences the size of the
minimal integer solution for $\Phi$. The former is 
the set 
$\PP(\Phi)$ defined in
\Cref{ssec:np-at-a-glance}. The next lemma shows that $\Phi$
has a solution modulo any prime not in $\PP(\Phi)$.
\begin{restatable}{lemma}{LemmaSimplePrimes}
  \label{lemma:simple-primes} Let $\Phi(\vec x) \coloneqq
  \bigwedge_{i=1}^m f_i \div g_i$ and $p \in \PP \setminus
  \pdiff(\Phi)$. Then, $\Phi$ has a solution~$\vec b \in
  \NN^d$ modulo $p$ such that $v_p(f_i(\vec b)) = 0$ for
  every $1\le i \le m$, and $\norm{\vec b} \le p-1$.
\end{restatable}
\noindent
The proof of~\Cref{lemma:simple-primes} is given
in~\Cref{appendix:bound-pzero}. In a nutshell, $v_p(f_i(\vec
b)) = 0$ holds if and only if ${f_i(\vec b) \not\equiv 0
\pmod p}$, meaning that the solution $\vec b$ can be
computed by considering a system of at most $m$
non-congruences; one for each left-hand side of $\Phi$.
Consider an ordering $\incord$ of the variables in $\vec x$.
Since $p \not\in \pdiff(\Phi)$, $p$ does not divide any
coefficient or constant appearing in some~$f_i$. This means
that if $f_i(\vec x) = f_i' + a \cdot x$, with $x =
\lv_{\incord}(f_i)$, we can rewrite $f_i(\vec x) \not\equiv
0 \pmod p$ as $x \not\equiv -a^{-1} f_i' \pmod p$,
where $a^{-1}$ is the inverse of $a$ modulo $p$. Then,
since $p > m$, one can find $\vec b$ by
picking suitable residues in $\{0,\dots,p-1\}$; this can be
done inductively, starting from the $\incord$-minimal
variable.

Extending $\pdiff(\Phi)$ into $\pzero(\Phi)$, hence
capturing the second of the two categories above, is a
delicate matter. In fact, while $\pdiff(\Phi)$ is defined
for an arbitrary system of divisibility constraints, the set
$\pzero(\Phi)$ can only meaningfully be defined on systems
that have the elimination property for an order~$\incord$.
For systems without the elimination property, one must first
appeal to~\Cref{lemma:add-elimination-property-summary}.
Let $\Phi$ be a system of divisibility constraints with the
elimination property. The set of \emph{difficult
primes}~$\pzero(\Phi)$ is the set of primes~${p \in \PP}$
satisfying at least one the following conditions:
\begin{itemize}
  \item[\textlabel{(P1)}{pzero:1}] $p \leq
  \card{S(\sterms(\Phi))}$,
  \item[\textlabel{(P2)}{pzero:2}] $p$ divides any non-zero
  coefficient or constant of a polynomial in
  $S(\sterms(\Phi))$, or 
  \item[\textlabel{(P3)}{pzero:3}] $p$ divides the smallest
  (in absolute value) non-zero $\lambda \in \ZZ$ such that
  $\lambda \cdot g \in \module_f(\Phi)$ for some primitive
  polynomial $f$ occurring in $\Phi$ and $g \in
  \sfterms{f}{\Phi}$ (if such a $\lambda$ exists). 
\end{itemize}
Note that~\ref{pzero:1} and~\ref{pzero:2} imply
$\pdiff(\Phi) \subseteq \pzero(\Phi)$. 
The following lemma
establishes bounds on these two sets that are central to the
proof of \Cref{theorem:local-to-global}.

\begin{restatable}{lemma}{LemmaBoundOnPzero}
  \label{lemma:bound-on-pzero}
  Consider a system of divisibility constraints~$\Phi(\vec x)$
  in $d$ variables. Then, the set of primes $\pdiff(\Phi)$ satisfies $\log_2(\Pi \pdiff(\Phi)) \leq m^2
  (d+2) \cdot (\maxbl{\Phi}+2)$. Furthermore, if $\Phi$ has
  the elimination property for an order~$\incord$ on $\vec
  x$, then the set of primes $\pzero(\Phi)$ satisfies $\log_2(\Pi \pzero(\Phi)) \leq 64 \cdot m^5
  (d+2)^4 (\maxbl{\Phi}+2)$.
\end{restatable}

\noindent
The proof of~\Cref{lemma:bound-on-pzero} is given
in~\Cref{appendix:bound-pzero}. Note that $\bitlength{S} =
\bigO{\log_2(\Pi S)}$ for any finite set $S$ of positive
integers, and therefore the above lemma bounds
$\bitlength{\pdiff(\Phi)}$ and $\bitlength{\pzero(\Phi)}$
polynomially.

\paragraph*{Proof sketch of \Cref{theorem:local-to-global}.}
Recall that \Cref{theorem:local-to-global} establishes a
local-to-global property for \mbox{$r$-increasing} systems of divisibility constraints $\Phi(\vec x)$: if such a system
has a solution $\vec b_p \in \ZZ^d$ modulo~$p$ for every
prime $p \in \pdiff(\Phi)$, then it has infinitely many integer
solutions, and a solution $\vec a \in \NN^d$ such that
$\maxbl{\vec a}\le (\bitlength{\Phi} + \max\{ \maxbl{\vec
b_p} : p \in \pdiff(\Phi) \})^{O(r)}$.  We give a high-level
overview of the proof of this result,
focusing on the part of the statement that constructs a
solution over $\NN$. The full proof is given
in~\Cref{section:algo:local-to-global}. Fix an order
$\incord$ in $X_1 \incord \dots \incord X_r$. We compute a
map $\vec \nu \colon \big(\bigcup_{j=1}^{r} X_j\big) \to
\pZZ$ such that $\vec \nu(\vec x)$ is a solution for $\Phi$
by induction on $r$, populating $\vec \nu$ according the
order~$\incord$. 

If $r = 1$, the system $\Phi$ is of the
form ${\bigwedge_{i=1}^\ell c_i \div g_i(\vec x) \land
\bigwedge_{j = \ell+1}^m f_j(\vec x) \div a_j \cdot f_j(\vec
x)}$, with $c_i \in \ZZ \setminus \{0\}$ and $a_j \in \ZZ$,
and $\vec \nu$ can be computed using the CRT. Given ${p \in
\pdiff(\Phi)}$, one considers the natural number $\mu_p
\coloneqq  \max \big\{ v_p(f(\vec b_p)) : f(\vec x) \text{
left-hand side of a divisibility in } \Phi \}$, which
determines up to what power of~$p$ the integer solution
given by $\vec \nu$ has to agree with the solution~$\vec
b_p$. Then, the CRT instance to be solved is $x_k \equiv
b_{p,k} \pmod {p^{\mu_p+1}}$ for every $p \in \pdiff(\Phi)$
and $1 \leq k \leq d$, where $x_1 \incord \dots \incord x_d$
are the variables in $\Phi$ and $b_{p,1},\dots,b_{p,d}$ are
their related values in $\vec b_p$.

When $r \geq 2$, the construction is much more involved. The
goal is to define $\vec \nu$ for the variables in $X_1$ in
such a way that the formula $\Phi' \coloneqq
\Phi\substitute{\vec \nu(x)}{x : x \in X_1}$ is increasing
for $X_2 \incord \dots \incord X_r$, and has solutions
modulo $p$ for every $p \in \pdiff(\Phi')$. This allows us
to invoke~\Cref{theorem:local-to-global} inductively,
obtaining a solution $\vec \xi  \colon
\big(\bigcup_{j=2}^{r} X_j\big) \to \pZZ$ for $\Phi'$. An
integer solution for $\Phi$ is then given by the union $\vec
\nu \sqcup \vec \xi$ of $\vec \nu$ and $\vec \xi$, i.e., the
map defined as $\vec \nu(x)$ for $x \in X_1$ and as $\vec
\xi(y)$ for $y \in \bigcup_{j=2}^r X_j$. To construct~$\vec
\nu$ for $X_1$, we first close $\Phi$ under the elimination
property
following~\Cref{lemma:add-elimination-property-summary},
obtaining an equivalent system~$\Psi$, and extend the
solutions $\vec b_p$ to every $p \in \pzero(\Psi)$ thanks
to~\Cref{lemma:simple-primes}. We then populate $\vec \nu$
following the order~$\incord$, starting from the smallest
variable. In the proof, this is done with a second
induction. Values for the variables in $X_1$ are found 
using~\Cref{thm:mixed-crt}.
When a new value $a_k \in \pZZ$ for a variable $x_k \in X_1$
is found, new primes need to be taken into account, since
substituting $a_k$ for $x_k$ yields a complete evaluation of
the polynomials in $S(\sterms(\Phi))$ with leading variable
$x_k$, i.e.,~these polynomials become integers that may be
divisible by primes not belonging to~$\pzero(\Psi)$. For
subsequent variables in~$X_1$, we make sure to pick values
that keep the evaluated polynomials as ``coprime as
possible'' with respect to these new primes. This condition
is necessary to obtain the new solutions~$\vec b_p$ for the
formula $\Phi'$, modulo every $p \in \pdiff(\Phi')$. The
precise system of (non-)congruences considered when
computing $x_k$ is 
\[ 
  \begin{cases} x_{k} \equiv b_{p,k} 
    \hspace{1.07cm} {\pmod
      {p^{\mu_p+1}}} &p \in \pzero(\Psi)\\
      g(\vec \nu(\vec y),x_{k}) \not\equiv 0 {\pmod
      q} &q \in Q \setminus \pzero(\Psi),\ g(\vec y,
      x_{k}) \in S(\Delta(\Psi)) \text{ with }
      \lv_{\incord}(g) = x_{k} 
  \end{cases}
\]
where $Q$ is the set of new primes obtained when fixing the
variables $\vec y = (x_1,\dots,x_{k-1})$, and $\mu_p
\coloneqq  \max \big\{ v_p(f(\vec b_p)) : f(\vec x) \text{
left-hand side of a divisibility in } \Psi \}$. 
\Cref{thm:mixed-crt} can be applied on the system above because primes in $Q \setminus \pzero(\Psi)$ do not satisfy the properties~\ref{pzero:1} and~\ref{pzero:2}.

To 
show that~\Cref{theorem:local-to-global} can be applied inductively
on~$\Phi'$, 
we rely on~\ref{pzero:3} and the elimination
property of $\Psi$ to show
that~$\Phi'$ has solutions modulo every
$p \in \pdiff(\Phi')$,
and on properties of $S$-terms and again on the elimination
property of $\Psi$
to show that $\Phi'$ is increasing for~${X_2 \incord \dots \incord X_r}$.

\input{sec-ip-gcd-summary.tex}

\subsection{Conclusion and future work}
We have established a polynomial small witness property for
integer programming with additional
GCD~constraints over linear polynomials. 
Our work also sheds new light on the feasibility problem for
systems of divisibility constraints between linear polynomials over the
integers, and more broadly on the existential fragment of
the first-order theory of the structure
$\Ldiv=(\ZZ,0,1,+,\le,\divides)$, which is known to be \np-hard
and decidable in \nexptime~\cite{Lip81,LechnerOW15}. \Cref{prop:increasing-in-np}
shows that systems of divisibility constraints in increasing form
are decidable in \np. Thus, in order to improve the known
\nexptime upper bound of existential \Ldiv, it would suffice to
provide an algorithm that translates an arbitrary existential \Ldiv
formula in increasing form without the exponential blow-up that
existing algorithms incur~\cite{Lipshitz78,LechnerOW15}.

Our work may also enable obtaining improved complexity results for
other problems that reduce to the existential theory of \Ldiv.  For
instance, \cite{LinM21} Lin and Majumdar reduce deciding a
special class of word equations with length constraints and
regular constraints to existential \Ldiv, hence obtaining an
\nexptime for their problem. The formulas resulting from their
reduction are of a special shape, and showing them to be
$r$-increasing for some fixed $r$ would directly yield a 
\pspace decision procedure for the aforementioned class of word equations.

%% file: sec-ip-gcd-summary.tex
\subsection{Solving an instance of IP-GCD}
\label{sec:intro-gcd-ip}

We now briefly discuss the proof of~\Cref{thm:small-model},
full details are deferred to~\Cref{sec:ip-gcd-small-model}. In
a nutshell, this result is shown by giving an algorithm that
reduces an \emph{IP-GCD system}~$\Phi(\vec x)
\coloneqq {A \cdot \vec x \leq \vec b} \land
\bigwedge_{i=1}^k \gcd(f_i(\vec x), g_i(\vec x)) \sim_i c_i$
into an equi-satisfiable disjunction of
several \mbox{$3$-increasing} systems of divisibility
constraints with coefficients and constants of polynomial
bit length. We then study bounds on the solutions of each of these
systems modulo the primes required by the local-to-global
property, and conclude that IP-GCD has a small
witness property over the integers directly
from~\Cref{theorem:local-to-global}.

Our arguments heavily rely on syntactic properties of the
systems of divisibility constraints we obtain when
translating an IP-GCD system~$\Phi$. These syntactic properties are
captured in~\Cref{sec:ip-gcd-small-model} with the notion of
\emph{\gcdtodiv} triple. The formal definition is rather 
lengthy, for this overview it suffices to know that
a triple $(\Psi,\vec u,E)$ is a \gcdtodiv triple if $\Psi$
is a system of divisibility constraints in which all numbers appearing are positive,
and $\vec u$ and $E$ are a vector and a matrix
that act as a change of variables between the variables
in~$\Psi$ and the variables in $\Phi$. The following
proposition formalizes the role of \gcdtodiv triples.

\begin{proposition}
  \label{prop:to-three-increasing}
  Let $\Phi$ be an IP-GCD system in $d$ variables. There is
  a set $C$ of \gcdtodiv triples such that 
  the set of integer solutions to $\Phi$ is 
  $\{ \vec u +
    E \cdot \vec \lambda : (\Psi,\vec u, E) \in C \text{ and }
    \vec \lambda \in \NN^m \text{ solution to } \Psi \}$.
  Every $(\Psi,\vec u, E) \in C$ has bit length polynomial
  in $\bitlength{\Phi}$ and is such that $\Psi$ is in
  $3$-increasing form.
\end{proposition}

\noindent
Above, $m$ is the number of free variables in $\Psi$, which
is also the number of columns in $E$. The algorithm showing
this proposition, cf.~\Cref{lemma:gcd-to-div}
and~\Cref{lemma:from-IP-GCD-to-increasing-systems}
in~\Cref{sec:ip-gcd-small-model}, performs a series of
equivalence-preserving syntactic transformations of $\Phi$
that are mainly divided into two steps: we first compute
from $\Phi$ a set of \gcdtodiv triples $B$ satisfying $\{
\vec x \in \ZZ^d : \vec x \text{ solution to } \Phi \} = \{
\vec u + E \cdot \vec \lambda : (\Psi,\vec u, E) \in B
\text{ and } \vec \lambda \in \NN^m \text{ solution to }
\Psi \}$, and then obtains $C$ by manipulating every system of
divisibility constraints
in $B$ to make it $3$-increasing. Below we give a summary of
these two steps.

\paragraph*{Step I: from IP-GCD to divisibility constraints.}
This step is split into three sub-steps: 
\begin{enumerate}
\item Reduce the input IP-GCD system $\Phi$ into an
  equi-satisfiable
  disjunction of \mbox{IP-GCD} system having GCD
  of the form $\gcd(f(\vec x), g(\vec x)) = c$
  or~${\gcd(f(\vec x), g(\vec x)) \geq c}$, and a system of
  inequalities $A \cdot \vec x \leq \vec b$ fixing
  a sign for every polynomial $h(\vec
  x)$ appearing in a GCD constraint, i.e., $A \cdot \vec x
  \leq \vec b$ 
  has
  either $h(\vec x) \leq
  -1$ or $h(\vec x) \geq 1$ as a row.
  

  \item Let $G$ be the set of 
    systems computed at the
  previous step. The algorithm erases the system of
  inequalities $A \cdot \vec x \leq \vec b$ from every
  IP-GCD system~$\Psi \in G$ by performing a change of
  variables. In particular, relying on a well-known result
  by von zur Gathen and Sieveking~\cite{vzGS78}, the
  algorithm computes a finite set $\{(\vec u_i, E_i) : i \in
  I_{\Psi}\}$ such that $\{\vec x \in \ZZ^d :  A \cdot \vec
  x \leq \vec b \} = \{ \vec u_i + E_i \cdot \vec \lambda :
  {\vec \lambda \in \NN^m},\, i \in I_\Psi \}$. For every $i
  \in I_\Psi$, the algorithm constructs a system of GCD
  constraints~$\Psi_i$ by
  replacing $\vec x$ in all GCD constraints of $\Psi$ with $\vec u_i + E_i
  \cdot \vec y$, where $\vec y$ is a family of fresh
  variables.
  The latter transformation also ensures that all numbers in the $\Psi_i$ are positive.

  \item The algorithm translates every GCD
  constraint in every $\Psi_i$ into a divisibility. Each
  constraint $\gcd(f(\vec y), g(\vec y)) = c$ is replaced by
  $\exists z \in \NN :~ {c \div f} \, \land\, {c \div g} \,\land\, {f
  \div z} \,\land\, g \div z+c\,,$ following B\'ezout's
  identity, whereas $\gcd(f(\vec y), g(\vec y)) \geq c$
  becomes ${\exists z \in \NN :~ z + c \div f \,\land\, z + c
  \div g}$. 
  The
  triple $(\Psi_i, \vec u_i, E_i)$ obtained after these
  replacements is a \gcdtodiv triple. 
\end{enumerate}

\paragraph{Step II: enforcing increasingness.} The algorithm
considers each \gcdtodiv triple $(\Psi, \vec u, E)$ computed
in the previous step and further manipulates it, producing a
set of \gcdtodiv triples $D$ having only systems of
divisibility constraints in $3$-increasing form, and satisfying 
\begin{equation}
  \label{eq:to-increasing:key-equivalence-summary}
  \{ 
    \vec u + E \cdot \vec \lambda : \vec \lambda \in \NN^m \text{ solution for } \Psi
  \} 
  = 
  \{
    \vec u' + E' \cdot \vec \lambda : (\Psi',\vec u', E') \in D,\, \vec \lambda \in \NN^{m'} \text{ solution for } \Psi'
  \}.
\end{equation}
The set $D$ is computed as follows. If $\Psi$ is already
$3$-increasing, then $D \coloneqq \{(\Psi, \vec u, E)\}$.
Otherwise, properties of \gcdtodiv triples ensure that there
is a non-constant primitive polynomial $f$ with positive
coefficients and constant such that $\module_f(\Psi) \cap \ZZ
\neq \{0\}$. 
The algorithm computes the smallest positive integer $c$
belonging to $\module_f(\Psi)$. We have that 
$\Psi$ entails
$f \div c$.
Let
$\lambda_1,\dots,\lambda_j$ be all the variables
in $f$. Since the coefficients and constant of $f$ are all positive
and variables are now interpreted over the naturals, such a
divisibility constraint can only be satisfied by assigning to each
variable 
an 
integer in $[0,c]$. The algorithm
iterates over each assignment~$\vec \nu \colon
\{\lambda_1,\dots,\lambda_j\} \to [0,c]$ satisfying $f \div
c$, computing from $(\Psi, \vec u, E)$ the \gcdtodiv triple
$(\Psi_{\vec \nu}, \vec u_{\vec \nu}, E_{\vec \nu})$ where
$\Psi_{\vec \nu} \coloneqq \Psi\substitute{\vec
\nu(\lambda_i)}{\lambda_i : i \in [1,j]}$, and $\vec u_{\vec
\nu}$ and 
$E_{\vec \nu}$
are obtained from $\vec u$ and $E$ based on $\vec \nu$ too.
All such triples are added to $D$ to replace $(\Psi,\vec u, E)$.
However, some newly added system $\Psi_{\vec \nu}$ may not be
$3$-increasing. If that is the case, Step~II is iteratively performed on $(\Psi_{\vec \nu},
\vec u_{\vec \nu}, E_{\vec \nu})$. Termination is guaranteed
because $\Psi_{\vec \nu}$ has strictly fewer variables than~$\Psi$
and the set of computed \gcdtodiv triples 
is the set $C$
from~\Cref{prop:to-three-increasing}. 

\paragraph{Bounds on the solutions modulo primes and proof sketch of~\Cref{thm:small-model}.}
Following~\Cref{prop:to-three-increasing}, what is left to
apply \Cref{theorem:local-to-global} is to
compute the solutions modulo primes in
$\pdiff(\Psi)$, for all $(\Psi, \vec u, E) \in C$. In~\Cref{sec:ip-gcd:small-bound-p-adic-solutions} we rely on properties of \gcdtodiv triples to show
the result below.

\begin{restatable}{lemma}{LemmaIPGCDSmallSolutionsModulo}
  \label{lemma:ip-gcd:small-solutions-modulo}
  Let $(\Psi,\vec u, E)$ be a \gcdtodiv triple in which
  $\Psi$ has $d$ variables, and consider~${p \in
  \PP(\Psi)}$. If~$\Psi$ has a solution modulo~$p$, then it
  has a solution $\vec b_p \in \ZZ^d$ modulo $p$ with
  $\norminf{\vec b_p} \leq (d+1) \cdot 
  \norminf{\Psi}^3 p^2$.
\end{restatable}

\Cref{prop:to-three-increasing}, and \Cref{lemma:ip-gcd:small-solutions-modulo,lemma:bound-on-pzero} 
imply the part of~\Cref{thm:small-model} not concerning
optimization as a corollary
of~\Cref{theorem:local-to-global}.
For optimization, consider a linear objective $\vec
c^\intercal \vec x$ to be minimized (the argument is
analogous for maximization) subject to an IP-GCD system
$\Phi(\vec x)$, and let $C$ be the set of \gcdtodiv triples
computed from $\Phi$
following~\Cref{prop:to-three-increasing}. We show in
\Cref{sec:proof-of-small-model} the following
characterization that implies the optimization part
of~\Cref{thm:small-model}: an optimal solution exists if and
only if 
\begin{enumerate*}[(i)]
  \item there is $(\Psi, \vec u, E) \in C$ such that $\Psi$ satisfiable over
    $\NN$, and
  \item for every $(\Psi, \vec u, E) \in C$ with $\Psi$
satisfiable over~$\NN$, ${\vec c^\intercal
(\vec u + E \cdot \vec \lambda)}$ has no variable with
a strictly negative coefficient.
\end{enumerate*}
Moreover, if there is an
optimal solution, then there is one with polynomial bit
length with respect to~$\bitlength{\Phi}$ 
and~$\bitlength{\vec c}$.
Briefly, 
the double implication comes from the fact that 
the construction required to establish~\Cref{theorem:local-to-global}
also shows that for each variable in $\vec \lambda$ there are infinitely many values 
that yield a solution to $\Psi$, both in the positive and negative direction, and therefore 
the existence of a variable in~${\vec c^\intercal
(\vec u + E \cdot \vec \lambda)}$ having a negative coefficient
entails the non-existence of an optimum. 
For 
the bound,
one shows that~$\min\{ \vec c^\intercal \vec u : (\Psi, \vec u, E)
\in C \}$ is a lower bound to every solution of $\Phi$.
Then, the 
polynomial bound follows 
directly from~\Cref{prop:to-three-increasing}.

%% file: section-CRT.tex
\section{A Chinese remainder theorem with non-congruences}
\label{section:CRT}
In this section, we prove our Chinese-remainder-style
theorem for simultaneous congruences and non-congruences
(\Cref{thm:mixed-crt}) as well as the abstract version of
Brun's pure sieve (\Cref{lem:extended-brun}). Throughout
this paper, $e$ is reserved for Euler's number, and $\exp(x)
\coloneqq e^x$.

We start by providing the proof of~\Cref{lem:extended-brun},
which following the original proof by Brun is established by
analyzing a truncated inclusion-exclusion principle.

\LemExtendedBrun*

\begin{proof}
  We define $S(A, Q) \coloneqq \card{\big(A \backslash
  \bigcup_{q \in Q}A_q\big)}$. 
  By definition of $S(A,Q)$
  we have:
  \begin{flalign*}
      S(A, Q) & = \card{A} - \sum_{q \in Q} \card{A_q} +
      \sum_{s \neq r \in Q} \card{(A_s \cap A_{r})} - \dots
      \pm \card{\Big(\bigcap_{p \in Q}A_p \Big)}\,\\ 
      & = \card{A_1} - \sum_{q \in Q} \card{A_q} + \sum_{s
      \neq r \in Q} \card{A_{s \cdot r}} - \dots \pm
      \card{A_{\setprod Q}} &\hspace{-12pt}\text{by the
      \emph{independence} property}.
  \end{flalign*}
  Truncating the inclusion-exclusion sequence above, after
  an even (resp.~odd) number of terms results in a lower
  bound (resp.~upper bound) for $S(A,Q)$. Truncating the
  sequence too early would result in a useless bound;
  e.g.,~stopping at the second term might result in a
  negative lower bound for $Q$ sufficiently large.
  Conversely, truncating it too late would make the
  hypotheses of the lemma too weak. To emphasize better this
  point, let us first clarify the truncation. Let $\omega(r)
  \coloneqq \card{\PP(r)}$ be the prime omega function and,
  given $k \in \NN$, define $Q(k) \coloneqq \{ r \in
  \divisors(\setprod Q) : \omega(r) \leq k \}$. Fix $\ell
  \in \NN_+$. We consider the (truncated) sequence
  $T(\ell,A,Q)$ given by 
  \[ 
    T(\ell,A,Q) \coloneqq \card{A_1} - \sum_{q \in Q} \card{A_q} + \sum_{s \neq r \in Q} \card{A_{s \cdot r}} - \dots \pm \sum_{\substack{r  \text{ product of}\\ \text{$\ell$ distinct primes in $Q$}}}
    \card{A_r} 
  \]
  which can be also written as $\sum_{r \in Q(\ell)}
  (-1)^{\omega(r)} \card{A_r}$. From the \emph{density}
  property, $T(\ell,A,Q)$ equals
  \begin{equation}
    \label{eq:trunc}
    \card{A} \cdot \sum_{r \in Q(\ell)} \frac{(-1)^{\omega(r)} m(r)}{r} 
    \, + \, \sum_{r \in Q(\ell)} 
    (-1)^{\omega(r)} \sigma(r).
  \end{equation}
  Note that $\mu(x) \coloneqq (-1)^{\omega(x)}$ is the
  M\"obius function~\cite{Hardy75}, which is multiplicative.
  Let us look at the two sides of the addition above. Note
  that for $\ell= \card{Q}$ the left term $\card A \cdot
  \sum_{r \in Q(\ell)} \frac{(-1)^{\omega(r)}
  m(r)}{r}$ can be factorized as $\card A \cdot \prod_{q \in
  Q} \big(1 + \frac{\mu(q) \cdot m(q)}{q}\big)$, because
  both $\mu$ and $m$ are multiplicative. This is equal to
  $\card A \cdot W_m(Q)$, by definition of $W_m(Q)$ and
  using the fact that $\mu(q) = -1$ for $q$ prime. In
  practice, the higher the~$\ell$, the closer the left term
  of the addition in~\eqref{eq:trunc} becomes to $\card A
  \cdot W_m(Q)$. However, increasing $\ell$ comes at the
  cost of increasing the error term given by the right term
  in the addition. Indeed, note that for $\ell = \card{Q}$
  the sum $\sum_{r \in Q(\ell)} (-1)^{\omega(r)}\sigma(r)$
  can a priori be larger than $\sigma(\setprod Q)$, which
  from the hypotheses can at best be bounded as
  $|\sigma(\setprod Q)| \leq m(\setprod Q) \leq
  d^{\card{Q}}$. Hence, to obtain the bounds in the
  statement of~\Cref{lem:extended-brun}, we need to find a
  value of $\ell$ making the left term in~\eqref{eq:trunc}
  close enough to $\card A \cdot W_m(Q)$ while keeping the
  error term small (in absolute value). Below, we first
  analyze the two terms of the addition in~\eqref{eq:trunc},
  and then optimize the value of~$\ell$. For brevity, we
  focus on computing the lower bound of $S(A,Q)$ (which is
  all we need for~\Cref{thm:mixed-crt}); thus setting $\ell$
  to be odd, so that $S(A,Q) \geq T(\ell,A,Q)$. The
  computation of the upper bound is analogous.

  \proofparagraph{Lower bound on the error term
  of~\eqref{eq:trunc}} Since $|\sigma(r)| \leq m(r) \leq
  d^{\omega(r)} \leq d^\ell$ when $\omega(r) \leq \ell$,
    \begin{equation} 
      \label{eq:error-term}
      \sum_{r \in Q(\ell)}
        \mu(r) \cdot \sigma(r) \geq \sum_{r \in Q(\ell)} - |\sigma(r)|
        \geq \sum_{r \in Q(\ell)} -d^{\ell}
        \geq - \Big(\frac{e \cdot \card{Q}}{\ell}\Big)^\ell d^\ell,
    \end{equation}
  where the rightmost inequality is derived by applying a
  well-known upper bound on the partial sums of binomial
  coefficients: $\card{Q(\ell)} = \sum_{i=0}^{\ell}
  {\card{Q} \choose i} \leq \big(\frac{e \cdot
  \card{Q}}{\ell}\big)^\ell$.

  \proofparagraph{Lower bound on the left term
  of~\eqref{eq:trunc}} Correctly computing a lower bound for
  this term requires a long manipulation using properties of
  the M\"obius function and bounds on prime numbers. The
  following claim (proven in~\Cref{appendix:sieve})
  summarizes this computation.

  \begin{restatable}{claim}{BrunLBLeftTerm}
    \label{lem:extended-brun:left-term}
    $\displaystyle\sum_{r \in Q(\ell)}\frac{\mu(r) \cdot
      m(r)}{r} \geq W_m(Q) \left( 1 - \left(\frac{e \cdot
      \alpha}{\ell}\right)^\ell \alpha \cdot
      e^{\alpha}  \right)$, with $\alpha \coloneqq (d+1)^2
      (2 + \ln \ln (\card{Q} + 1))$.
  \end{restatable}

  \proofparagraph{Optimizing the value of~$\ell$} 
  To obtain the lower bound for $S(A,Q)$ presented in the
  statement of the lemma, we want $\ell$ to be chosen so
  that 
  \[
    \card{A} \cdot \sum_{r \in Q(\ell)}\frac{\mu(r) \cdot m(r)}{r} \geq \frac{1}{2} \cdot \card{A} \cdot W_m(Q).
  \]
  Following~\Cref{lem:extended-brun:left-term}, it suffices
  to pick an $\ell$ making the inequality ${\left(\frac{e
  \cdot \alpha}{\ell}\right)^\ell \alpha \cdot e^\alpha \leq
  \frac{1}{2}}$ true. Note that, since $d \geq 1$ and
  $\card{Q} \geq 1$, we have $\alpha > 6.5$\,. Then, we see
  that $\ell \geq 1.44 \cdot e \cdot \alpha$ does the job:
  \begin{align*}
    \left(\frac{e \cdot \alpha}{\ell}\right)^\ell
    \alpha \cdot e^{\alpha}
    \leq 
    \left(\frac{1}{1.44}\right)^{1.44 \cdot e \cdot \alpha} \cdot e^{\alpha + \ln \alpha}
    \leq \frac{e^{\alpha + \ln \alpha}}{1.44^{1.44 \cdot e \cdot \alpha}}
    \leq \frac{e^{1.3 \cdot \alpha}}{1.44^{1.44 \cdot e \cdot \alpha}}
    \leq
    \left(\frac{e^{1.3}}{1.44^{1.44 \cdot e}} \right)^{6.5}
    \leq \frac{1}{2}.
  \end{align*}
  Hence, we pick $\ell$ to be an odd number in $[1.44 \cdot
  e \cdot \alpha,\ 1.44 \cdot e \cdot \alpha + 2]$.
  From~\Cref{eq:error-term} we obtain 
  \begin{align*}
    \sum_{r \in Q(\ell)} \mu(r) \cdot \sigma(r) 
    \geq -\Big(\frac{e \cdot \card{Q}}{1.44 \cdot e \cdot \alpha +2}\Big)^{1.44 \cdot e \cdot \alpha + 2} \cdot d^{1.44 \cdot e \cdot \alpha + 2}
    \geq -\big(d \cdot \card{Q}\big)^{4 (d+1)^2  (2 + \ln \ln(\card Q+1))+2}.
  \end{align*}
  As $S(A,Q) \geq T(\ell,A,Q) = \card{A} \cdot \sum_{r \in
  Q(\ell)} \frac{\mu(r) \cdot m(r)}{r} \, + \, \sum_{r \in
  Q(\ell)} \mu(r) \cdot \sigma(r)$, that completes the
  proof.
\end{proof}

We now move to the proof of~\Cref{thm:mixed-crt}.

\ThmMixedCRT*


\begin{proof}
  Expanding on the sketch of the proof given
  in~\Cref{sec:intro-crt}, recall that the set of primes $Q$
  and $d \in \pZZ$ defined in the statement
  of~\Cref{thm:mixed-crt} coincide with their homonyms
  in~\Cref{lem:extended-brun}. Furthermore, we let $n
  \coloneqq \setprod Q$, and define: 
  \begin{itemize}
    \item $S_M$ to be the solution set to the system of
    congruences ${\forall m \in M}$, ${x \equiv b_m \pmod
    m}$, which is a shifted lattice with period $\Pi M$ by
    the CRT, 
    \item $B(z) \coloneqq [k,k+z] \cap S_M$, where $k$ is
    the integer in the statement of the theorem,
    \item some integer $z$ to be optimized. We will show
    that $z = \ecrtf(Q,d)$ yield the theorem,
    \item $A \coloneqq B(z)$, and given $q \in Q$, ${A_q
    \coloneqq \{ a \in A : \text{there is } {i \in [1,d]}
    \text{ s.t.}~a \equiv c_{q,i} \pmod  q \}}$,
    \item for $r \in \divisors(n)$ not prime, ${A_r
    \coloneqq A \cap \bigcap_{q \in \PP(r)} A_q}$,
    \item for $r \in \divisors(n)$, ${m(r) \coloneqq
    \prod_{q \in \PP(r)} \#\{ c_{q,i} \mod q : i \in [1,d]
    \}}$, which is a multiplicative function,  
    \item and we take $\sigma(r) \coloneqq \card{A_r} -
    \card{A} \cdot \frac{m(r)}{r}$ as an error function.
  \end{itemize}
  Note that, by definition, $A \setminus \bigcup_{q \in Q}
  A_q$ corresponds to the set of solutions of $\mcS$ that
  belong to $[k,k+z]$. We show that the objects above
  satisfy the hypothesis of~\Cref{lem:extended-brun}, and
  that taking $z = \ecrtf(Q,d)$ makes the cardinality of $A
  \setminus \bigcup_{q \in Q} A_q$ strictly positive,
  yielding~\Cref{thm:mixed-crt}. 

  \proofparagraph{The assumptions
  of~\Cref{lem:extended-brun} hold} By hypothesis $\min(Q) >
  d$, hence $m(q) \leq d \leq q-1$ for every $q \in Q$.
  Below, we show that the \emph{independence} and
  \emph{density} properties are satisfied, and that
  $\abs{\sigma(r)} \leq m(r)$ for every ${r \in
  \divisors(n)}$. This allows us to
  apply~\Cref{lem:extended-brun} in the second part of the
  proof. The \emph{independence} property is trivially
  satisfied: given $r,s \in \divisors(n)$ coprime, we have 
  \[
  A_{r \cdot s} = A \cap \bigcap_{q \in \PP(r
  \cdot s)} A_q = \Big(A \cap \bigcap_{q \in \PP(r)} A_q\Big)
  \cap \Big(A \cap \bigcap_{p \in \PP(s)} A_p\Big) = A_{r} \cap
  A_{s}.
  \]
  Below, fix~${r \in \divisors(n)}$. The \emph{density}
  property and the condition $\abs{\sigma(r)} \leq m(r)$ are
  proved together. By definition of~$A_r$, 
  \[ 
    A_r = \bigcup_{\alpha \colon \PP(r) \to [1,d]} (A \cap S_{\alpha,r})\,,
    \qquad\text{where }
    S_{\alpha,r} \coloneqq \{ \ell \in \ZZ : \text{ for every } q \in \PP(r),\, \ell \equiv c_{q,\alpha(q)} \pmod q \}.
  \]
  The following claim bounds the cardinality of each $(A
  \cap S_{\alpha,r})$. It is proven in~\Cref{appendix:crt}.
  \begin{restatable}{claim}{ThmMixedCRTClaimOne}
    \label{thm:mixed-crt:claim1}
    $\displaystyle \frac{\card{A}}{r} - 1 \leq \card{(A \cap
    S_{\alpha,r})} \leq \frac{\card{A}}{r} + 1$.
  \end{restatable}

  \noindent
  Directly form their definition, given two functions
  $\alpha_1,\alpha_2 \colon \PP(r) \to [1,d]$, the sets
  $S_{\alpha_1,r}$ and $S_{\alpha_2,r}$ satisfy one of the
  two following properties:
  \begin{itemize}
    \item 
    $S_{\alpha_1,r} \cap S_{\alpha_2,r} = \emptyset$ \ (this
    occurs when $c_{q,\alpha_1(q)} \not\equiv
    c_{q,\alpha_2(q)} \pmod q$ for some~$q \in \PP(r)$), or
    \item $S_{\alpha_1,r} = S_{\alpha_2,r}$ \ (this occurs
    when $c_{q,\alpha_1(q)} \equiv c_{q,\alpha_2(q)} \pmod
    q$, for every $q \in \PP(r)$).
  \end{itemize}
  With this in mind, we note that the number of disjoint
  sets in $ \{S_{\alpha,r} : \alpha \colon \PP(r) \to
  [1,d]\}$ corresponds to the value of the multiplicative
  function $m(r)$. Then, by~\Cref{thm:mixed-crt:claim1},
  $(\frac{\card{A}}{r} - 1) \cdot m(r) \leq \card{A_r} \leq
  (\frac{\card{A}}{r} + 1) \cdot m(r)$. This implies that
  $\sigma(r) = \card{A_r} - \card{A} \cdot \frac{m(r)}{r}$
  is such that $\abs{\sigma(r)} \leq m(r)$, as required, and
  also shows that the \emph{density} property holds.

  \proofparagraph{Applying~\Cref{lem:extended-brun}} The
  previous part of the proof shows that we can
  apply~\Cref{lem:extended-brun}, from which  we obtain
  $\card{\Big(A \setminus \bigcup_{q \in Q} A_q\Big)} \geq
  \frac{1}{2} \cdot \card{A} \cdot W_m(Q) - \brunf(Q,d)$.
  Remember that $A = [k,k+z] \cap S_M$ and that $A \setminus
  \bigcup_{q \in Q} A_q$ corresponds to the set of solutions
  of $\mcS$ that belong to $[k,k+z]$. To conclude the proof
  it suffices to make $\frac{1}{2} \cdot \card{A} \cdot
  W_m(Q) - \brunf(Q,d)$ greater or equal to $1$ by
  opportunely selecting the value of the parameter~$z$. We
  want $\card([k,k+z] \cap S_M) \geq 2 \cdot W_m(Q)^{-1}
   (1+\brunf(Q,d))$ which, from the fact that $S_M$ is
  periodic in $\setprod M$, holds as soon as $z \geq 2 \cdot
  W_m(Q)^{-1}  (1+\brunf(Q,d)) \cdot \setprod M$. 

  The following claim on an upper bound for $W_m(Q)^{-1}$ is
  proven in~\Cref{appendix:crt}.

  \begin{restatable}{claim}{ClaimCRTBoundOnW}
    \label{claim:CRT:bound-on-W}
    $\displaystyle W_m(Q)^{-1} \leq (d+1)^{10 d}
    \ln(\card{Q}+1)^{3 d}$. 
  \end{restatable}

\noindent
  \Cref{claim:CRT:bound-on-W} and the definition of $\brunf$
  show that setting ${z \coloneqq \big((d+1) \cdot
  \card{Q}\big)^{4 (d+1)^2 (3 + \ln \ln (\card{Q}+1))}
  \!\cdot \setprod M}$ suffices to satisfy $z \geq 2 \cdot
  W_m(Q)^{-1} (1+\brunf(Q,d)) \cdot \setprod M$,
  concluding the proof.
\end{proof}

%% file: sec-divisibility-algo.tex
\section{A novel strategy for Lipshitz's local-to-global property}
\label{sec:divisibility-algo}

In this section we establish~\Cref{theorem:local-to-global},
providing an asymptotical improvement over the
local-to-global properties for systems of divisibility
constraints discovered by Lipshitz~\cite{Lipshitz78} and
later refined by Lechner et al.~\cite{LechnerOW15}. Most of
the definitions and some intermediate lemmas required for this
result were already formally presented
in~\Cref{sec:intro-local-global}. To avoid repeating them,
we refer the reader to that section, and consider here only
concepts for which further details are required in order to give
the proof of~\Cref{theorem:local-to-global}. On a
high-level, recall that the main concepts discussed
in~\Cref{sec:intro-local-global} are:
\begin{itemize}
  \item The notions of \emph{divisibility module} and
    \emph{$r$-increasing form}. In general, only systems of
    divisibility constraints in increasing form can be
    solved via the local-to-global property.
  \item The notions of \emph{elimination property},
  \emph{$S$-polynomials} and \emph{$S$-terms}. The first
  notion relies on \emph{divisibility modules} to close a system under a finite representation of all its
  entailed divisibilities. The latter two terms are required
  to establish~\Cref{theorem:local-to-global} inductively;
  we will use them to ensure that increasingness is not lost
  after fixing the value of a variable.
  \item The notion of \emph{difficult primes}
  $\pzero(\Phi)$, that is primes~$p$ for which either the
  system of divisibility constraints~$\Phi$ might not have a
  solution modulo $p$, or the solution always exists but
  still influences the minimal integer solution for $\Phi$. 
\end{itemize}

\noindent
Except for~\Cref{theorem:local-to-global}, we defer all
proofs of intermediate results
to~\Cref{appendix:module-basis,appendix:bound-pzero}. 

\paragraph*{Assumptions and further basic definitions.}
Let $\Phi(\vec x) \coloneqq \bigwedge_{i = 1}^m f_i(\vec x)
\div g_i(\vec x)$ be a system of divisibility constraints in
$d$ variables. Throughout the section, wlog.~we tacitly
assume the systems to be non-empty ($m \geq 1$) and
\emph{reduced}, that is such that the GCD of all
coefficients and constants appearing in divisibilities $f
\div g$ is $1$, i.e., $\gcd(\gcd(f),\gcd(g))=1$. Recall
that we assume that $f_i\neq 0$ for all $1\le i\le m$.

Given $\vec b \in \ZZ^{i}$ and a polynomial
$f(x_1,\dots,x_d)$, we write $f(\vec b,x_{i+1},\dots,x_d)$
for the polynomial in variables $(x_{i+1},\dots,x_{d})$
obtained from $f$ by evaluating $x_j$ as the $j$-th entry of
$\vec b$, for all~$j \in [1,i]$.  
Given $\vec v = (v_1,\dots,v_n) \in \ZZ^d$, $\norminf{\vec
v} \coloneqq \max\{ \abs{v_i} : i \in [1,n] \}$ stands for
the \emph{(infinity) norm} of~$\vec v$. 

We define
$\norminf{S} \coloneqq \max\{ \norminf{s} : s \in S\}$, for
every finite set $S$ of objects having a defined notion of
infinity norm. The norm $\norminf{A}$ of a matrix $A$ is the
norm of the set of its columns. Given a polynomial $f = \vec
a^\intercal \vec x + c$, $\norminf{f} \coloneqq
\max(\norminf{\vec a}, \abs{c})$. For a system of
divisibility constraints $\Phi$, $\norminf{\Phi} \coloneqq
\norminf{\terms(\Phi)}$. 

We write $\bitlength{a} \coloneqq 1
+ \ceil{\log_2(\abs{a}+1)}$ for the bit length of $a \in
\ZZ$. The bit length of a set (or vector) $S$ of $n$ objects
$s_1,\dots,s_n$ having a defined notion of bit length
$\bitlength{.}$ is itself defined as $\bitlength{S}
\coloneqq n + \sum_{i=1}^n \bitlength{s_i}$. We define
$\bitlength{f} \coloneqq \bitlength{\vec a} + \bitlength{c}
+ 1$ and $\bitlength{\Phi} \coloneqq
\bitlength{\terms(\Phi)}$ for the bit length of a polynomial
$f = \vec a^\intercal \vec x + c$ and of a system of
divisibility constraints~$\Phi$, respectively. Note that
$\bitlength{\norminf{S}}$ is simply the bit length of the
infinity norm of $S$; where $S$ is any object having a
defined notion of infinity norm.

\subsection{Bounds on divisibility modules, elimination property, $S$-terms, and $\pzero(\Phi)$}
\label{subsec:divisibility:notion-increasing}
\label{subsec:divisibility:notion-elimination}
\label{subsec:divisibility:difficult-primes}

For the proof of~\Cref{theorem:local-to-global} we need to refine some of the bounds given in~\Cref{sec:intro-local-global}.
In that section we have briefly discussed
the existence of an algorithm to close a system of
divisibility constraints under the elimination property
(\Cref{lemma:add-elimination-property-summary}). This
algorithm relies on a procedure computing a span for the
\emph{divisibility module} $\module_f(\Phi)$ of a primitive
polynomial $f$ with respect to a system of divisibility
constraints~$\Phi$. Recall that $\module_f(\Phi)$ is a
vector subspace encoding all the divisibilities of the form
$f \div g$ implied by $\Phi$. From the formal definition of
divisibility module, it is simple to convince
ourselves that a set spanning $\module_f(\Phi)$ can be found
by taking $f$ together with a subset of the right-hand
sides of the divisibilities in~$\Phi$,
possibly scaled. In~\Cref{appendix:module-basis} we show
that computing such a span can be done in polynomial-time
by a fix-point algorithm chaining computations of
integer kernels.

\begin{restatable}{lemma}{LemmaModuleSpan}
  \label{lemma:module-span}
  There is a polynomial-time algorithm that, given a system~$\Phi(\vec x)
  \coloneqq \bigwedge_{i=1}^m f_i \div g_i$
  and a primitive polynomial $f$, computes
   $c_1,\dots,c_m\in \NN^m$ such that
  $\{f,\,c_1 \cdot g_1,\, \dots,\, c_m \cdot g_m\}$ spans
  $\module_f(\Phi)$ and 
  $c_i \leq ((m+3) \cdot
  (\norminf{\Phi}+2))^{(m+3)^3}$ for all $1\le i\le m$.
\end{restatable}

\noindent
Regarding the computation of formulae with the elimination
property, \Cref{lemma:add-elimination-property-summary} is
not precise enough for our purposes to
establish~\Cref{theorem:local-to-global}. We restate it,
tracking the growth of constants and coefficients, as well
as structural properties of the output system of
divisibility constraints.

\begin{restatable}{lemma}{LemmaAddEliminationProperty}
  \label{lemma:add-elimination-property}
  There is a polynomial-time algorithm that, given a system
  of divisibility constraints $\Phi(\vec x) \coloneqq
  \bigwedge_{i=1}^m f_i \div g_i$ and an order $x_1 \incord
  \dots \incord x_d$ for $\vec x$, computes $\Psi(\vec x)
  \coloneqq \bigwedge_{i=1}^n f_i' \div g_i'$ with the
  elimination property for~$\incord$ that is equivalent to
  $\Phi(\vec x)$, both over~$\ZZ$ and modulo each $p \in
  \PP$. The
  algorithm ensures that:
  \begin{enumerate}
  \item\label{lemma:add-elim-property:item-1} For any
    divisibility constraint $f\div g$ such that $f$ is not primitive,
    $f \div g$ occurs in $\Phi$ if and only if $f \div g$
    occurs in $\Psi$. Moreover, for every $f_i'\div g_i'$ in
    $\Psi$ such that $f_i'$ is primitive, there is some
    $f_j\div g_j$ in $\Phi$ such that $f_j'$ is the
    primitive part of $f_j$.
    \item\label{lemma:add-elim-property:item-2} For every
    primitive polynomial $f$, $\module_f(\Phi) =
    \module_f(\Psi)$ (in particular, if $\Phi$ is increasing
    for some order~$\incord'$ then so is $\Psi$, and vice
    versa).
    \item\label{lemma:add-elim-property:item-3}
    $\norminf{\Psi} \leq (d+1)^{\bigO{d}}(m + \norminf{\Phi}
    +2)^{\bigO{m^3 d}}$ \ and \ $n \leq m \cdot (d+2)$.
  \end{enumerate}
\end{restatable}

\noindent
Let us sketch this algorithm.
For every primitive part $f$ of a polynomial
appearing in the left-hand side of a divisibility constraint in~$\Phi$,
the algorithm first computes the set $S \coloneqq \{f,\,c_1
\cdot g_1,\, \dots,\, c_m \cdot g_m\}$ spanning
$\module_f(\Phi)$, using the algorithm
of~\Cref{lemma:module-span}. The set $S$ can be represented
as the matrix~$A \in \ZZ^{(d+1) \times (m+1)}$ in which each
column $(a_d,\dots,a_1,c)$ contains the coefficients and the
constant of a distinct element of $S$, with $a_i$ being the
coefficient of~$x_i$ for $i \in [1,d]$, and $c$ being the
constant of the polynomial. The algorithm puts $A$ in
column-style Hermite normal form, obtaining linearly
independent polynomials $h_1,\dots,h_\ell$ with different
leading variables with respect to $\incord$. Because of how
the coefficients and constants are arranged in $A$, we can
obtain the system $\Psi$ by simply replacing 
divisibility constraints of the form $f \div g$ appearing in~$\Phi$
with the divisibility constraints $f \div h_1,\dots,f \div h_\ell$.
\Cref{lemma:add-elim-property:item-1,lemma:add-elim-property:item-2}
are then easily seen to be satisfied,
whereas~\Cref{lemma:add-elim-property:item-3} follows from
the bound on $c_1,\dots,c_m$ given in~\Cref{lemma:module-span} together with known
bounds for putting an integer matrix in Hermite normal
form~\cite{Vanderkellen00}. Full details are given
in~\Cref{appendix:module-basis}, together with the proof of
the following lemma.

\begin{restatable}{lemma}{LemmaSubstitAndElim}
  \label{lemma:substit-and-elim} 
  Let $\Phi(\vec x,\vec y)$ and $\Psi(\vec x,\vec y)$ be
  input and output of the algorithm
  in~\Cref{lemma:add-elimination-property}, respectively.
  For every~$\vec \nu : \vec x \to \ZZ$ and primitive
  polynomial $f$, $\module_f(\Phi(\vec \nu(\vec x),\vec y))
  \subseteq \module_f(\Psi(\vec \nu(\vec x),\vec y))$.
\end{restatable}

\noindent
This lemma, established by relying on the definition of
divisibility module together
with~\Cref{lemma:add-elim-property:item-1,lemma:add-elim-property:item-2}
of~\Cref{lemma:add-elimination-property}, is used in the
proof of~\Cref{theorem:local-to-global} to
establish that if $\Psi(\vec \nu(\vec x),\vec y)$ is in
increasing form for some order, then so is $\Phi(\vec
\nu(\vec x),\vec y)$.

To prove~\Cref{theorem:local-to-global} we also need a bound
on the number of $S$-terms of a system of divisibility
constraints. We have already claimed
in~\Cref{sec:intro-local-global} that systems with the
elimination property only have polynomially many $S$-terms.
The precise bound, computed following the relevant
definitions, is given in the following lemma
(see~\Cref{appendix:bound-pzero} for the complete proof).

\begin{restatable}{lemma}{LemmaBoundSterms}
  \label{lemma:bound-sterms}
  Let $\Phi \coloneqq
  \bigwedge_{i=1}^m f_i \div g_i$ be a system of divisibility constraints
  in $d$ variables with the elimination property for~$\incord$. Then,
  \begin{enumerate*}[(i)]
    \item\label{lemma:bound-sterms-one}
      $\card{\sterms(\Phi)} \leq 2 \cdot m^2 (d+2)$ and
    \item\label{lemma:bound-sterms-two}
      $\maxbl{\sterms(\Phi)} \leq (d+2) \cdot
      (\bitlength{\norminf{\Phi}}+1)$.
  \end{enumerate*}
\end{restatable}

Lastly, let us restate the two lemmas from~\Cref{sec:intro-local-global} analyzing properties of $\pzero(\Phi)$ and $\pdiff(\Phi)$; they are proven in~\Cref{appendix:bound-pzero} and are fundamental to obtain the upper bound in the statement of~\Cref{theorem:local-to-global}. 
Recall that $\pdiff(\Phi) \coloneqq \{ p \in \PP : p \leq m$ or $p$
divides a coefficient or constant appearing in some $f_i
\}$ is the set of primes~$p$ for which $\Phi$ may not have a solution modulo $p$. 
For primes that lie outside~$\pdiff(\Phi)$ we always have a 
small solution:

\LemmaSimplePrimes*

Following the next lemma,
the bit lengths of $\pdiff(\Phi)$ and $\pzero(\Phi)$ are 
polynomially bounded:

\LemmaBoundOnPzero*

\subsection{Proof of~\Cref{theorem:local-to-global}: the local-to-global property}
\label{section:algo:local-to-global}

We are now ready to formalize the local-to-global property
(\Cref{theorem:local-to-global}). Simliar to
Lipshitz' approach~\cite{Lipshitz78}, the proof of this property is constructive and
yields a procedure that given an $r$-increasing system of
divisibility constraints $\Phi$ and solutions for $\Phi$
modulo $p$ for every $p \in \pdiff(\Phi)$, constructs an
integer solution for $\Phi$. \Cref{func:new-solveincreasing}
provides the pseudocode of this procedure, which we mainly
give as a way of summarizing the various steps of the proof
of~\Cref{theorem:local-to-global}.

\input{divisibility-algo-pseudocode-local-to-global}

\TheoremLocalToGlobal*

\begin{proof}
  \input{local-to-global-proof}\end{proof}

\begin{remark}
  \label{remark:loc-to-glo:inf-solutions}
  Let us briefly discuss how the infinitely many solutions
  of $\Phi$ ensured by~\Cref{theorem:local-to-global} look
  like. Since solutions are constructed by solving the
  systems of (non-)congruences
  in~\Cref{eq:l-t-g:base-case,eq:l-to-g:non-cong-sys}
  (see~\Cref{func:new-solveincreasing} for a summary),
  \Cref{thm:mixed-crt} ensures that $\Phi$ has infinitely
  many solutions. More precisely, the following property
  holds: let $(\incord) \in (X_1 \incord \dots \incord
  X_r)$, ${x \in \bigcup_{j=1}^r X_j}$, and $\vec \nu \colon
  \bigcup_{j=1}^r X_j \to \ZZ$ be the solution of~$\Phi$
  computed by~\Cref{func:new-solveincreasing}. The system
  ${\Phi\substitute{\vec \nu(y)}{y : y \incord x}}$ has a
  solution for infinitely many positive and negative values
  of $x$.
\end{remark}

\subsection{Deciding systems of divisibility constraints in increasing form in NP}
\label{subsec:deciding-increasing-systems-in-NP}

\Cref{theorem:local-to-global} provides a way of
constructing integer solutions of bit length exponential in
$r$ for \mbox{$r$-increasing} systems of divisibility
constraints. A different approach not relying on
constructing integer solutions shows that the feasibility
problem for systems of divisibility constraints in
increasing form is in $\np$. 

Let $\Phi(\vec x) \coloneqq \bigwedge_{i=1}^m f_i \div g_i$
be a formula in increasing form for an order~$\incord$.
According to~\Cref{theorem:local-to-global}, $\Phi$ is
satisfiable over the integers if and only if there are
solutions $\vec b_p$ for $\Phi$ modulo $p$ for every prime
$p$ belonging to $\pdiff(\Phi)$.
From~\Cref{lemma:bound-on-pzero}, the bit length of
$\pdiff(\Phi)$ is polynomial in $\bitlength{\Phi}$, and
therefore only polynomially many primes of polynomial bit
length need to be considered. Recall that $\Phi$ has a
solution modulo~$p$ whenever the system $\bigwedge_{i=1}^m
v_p(f_i(\vec x)) \leq v_p(g_i(\vec x)) \land f_i(\vec x)
\neq 0$ has a solution. In~\cite{GuepinHW19} it is shown
that the feasibility problem for these constraint systems is
in~\np (this result holds for solutions over the integers,
$p$-adic integers, and $p$-adic numbers), and therefore
there are certificates of feasibility having size polynomial
in~$\bitlength{p}$ and $\bitlength{\Phi}$. The set of these
certificates, one for each prime in $\pdiff(\Phi)$, is a
polynomial size certificate for the feasibility of $\Phi$.

\begin{proposition}
  \label{prop:increasing-in-np}
  Feasibility for systems of divisibility constraints in
  increasing form is in~\np.
\end{proposition}

\noindent
Of course, we know from the family of formulae $\Phi_n$
introduced in~\Cref{ssec:np-at-a-glance} (and the one
after~\Cref{theorem:local-to-global}) that systems in
increasing form might have minimal solutions of exponential
bit length. Therefore,~\Cref{prop:increasing-in-np} is of no
use when establishing~\Cref{thm:small-model}. However, it
still has an interesting implication: if the feasibility
problem for systems of divisibility constraints lies outside
$\np$, then there is no polynomial time non-deterministic
Turing machine for finding an equisatisfiable system in
increasing form.

%% file: divisibility-algo-pseudocode-local-to-global.tex
\begin{algorithm}[!t]
  \caption{An algorithmic summary of the local-to-global property}
  \label{func:new-solveincreasing}
  \begin{algorithmic}[1]
    \medskip
    \Require 
      \begin{minipage}[t]{0.8\linewidth}
        a system of divisibility constraints $\Phi(\vec x)$
        increasing for $X_1 \incord \dots \incord X_r$,\\
        and a solution $\vec b_p$ for $\Phi$ modulo $p$ 
        for every $p \in \pdiff(\Phi)$.
      \end{minipage}
    \Ensure
      a solution $\vec \nu \colon \vec x \to \pZZ$ for $\Phi$.
    \medskip
    \State $\vec \nu \coloneqq \vec \epsilon$
    \itComment{empty map}
    \State let $\incord$ be an ordering in $(X_1 \incord \dots \incord X_r)$
    \State $(x_1,\dots,x_d)$ $\coloneqq$ variables in $X_1$, 
      in increasing order for $\incord$
    \If{$r=1$}\label{func:l-t-g:base-case}
      \itComment{base case}
      \For{$p \in \pdiff(\Phi)$}
        \label{func:l-t-g:base-case:mu}
        $\mu_p \coloneqq  \max
        \big\{ v_p(f(\vec b_p)) : f(\vec
        x) \text{ left-hand side of a divisibility in } \Phi
        \}$
      \EndFor
      \For{$\ell$ from $1$ to $d$}
        \For{$p \in \pdiff(\Phi)$}
          $b_{p,\ell}$ $\coloneqq$ value of $\vec b_p$ for the variable $x_\ell$
        \EndFor
        \State insert $(x_\ell \mapsto
        a)$ in $\vec \nu$ \,\textbf{where} $a \in \pZZ$ is a solution for the system
        \itComment{CRT}\label{func:l-t-g:CRT}
        \State \hspace{1.2cm} 
          $\begin{cases} 
            x_{\ell} \equiv b_{p,\ell} \hspace{0.65cm} {\pmod
            {p^{\mu_p+1}}} &\qquad p \in \pdiff(\Phi)
          \end{cases}$
      \EndFor
      \State \textbf{return} $\vec \nu$
    \Else 
      \itComment{$r \geq 2$, recursive case}
      \State $\Psi$ $\gets$ closure of $\Phi$ for the elimination property for the order $\incord$
        \itComment{\Cref{lemma:add-elimination-property}}
        \label{func:l-t-g:ind:elim-prop}
      \For{$p \in \pzero(\Psi) \setminus \pdiff(\Phi)$}
        \label{func:l-t-g:ind:new-b-ps}
        \State $\vec b_p$ $\coloneqq$ solution for $\Phi$ modulo $p$ satisfying $v_p(f(\vec b_p)) = 0$ for every 
        \State \hspace{1.2cm} $f(\vec
        x)$ in the left-hand side of a divisibility in $\Phi$
        \itComment{\Cref{lemma:simple-primes}}
      \EndFor
      \For{$p \in \pzero(\Psi)$}
        $\mu_p \coloneqq  \max
        \big\{ v_p(f(\vec b_p)) : f(\vec
        x) \text{ left-hand side of a divisibility in } \Psi
        \}$
      \EndFor
      \State $Q \coloneqq \emptyset$
      \For{$\ell$ from $1$ to $d$}\label{func:l-t-g:ind:inner-induction}
        \For{$p \in \pzero(\Psi)$}
          $b_{p,\ell}$ $\coloneqq$ value of $\vec b_p$ for the variable $x_\ell$
        \EndFor
        \State insert $(x_\ell \mapsto
        a)$ in $\vec \nu$ \,\textbf{where} $a \in \pZZ$ is a solution for the system
          \itComment{\Cref{thm:mixed-crt}}
          \label{func:l-t-g:ind:mixed-crt}
        \vspace{2pt}
        \State \quad $\begin{cases} x_{\ell} \equiv b_{p,\ell} 
          \hspace{0.84cm} {\pmod
            {p^{\mu_p+1}}} &p \in \pzero(\Psi)\\
            g(\vec \nu(\vec y),x_{\ell}) \not\equiv 0 {\pmod
            q} &q \in Q \setminus \pzero(\Psi),\ g(\vec y,
            x_{\ell}) \in S(\Delta(\Psi)) \text{ with }
            \lv_{\incord}(g) = x_{\ell} \end{cases}$
        \vspace{6pt}
        \State\label{func:l-t-g:ind:Q} 
          $Q \gets Q \cup \{ p \in \PP : \text{there is } 
          h(\vec y) \in S  (\Delta(\Psi)) \text{ such that}~\lv_{\incord}(h) = x_{\ell} \text{ and } p \div h(\vec \nu(\vec y)) \}$
      \EndFor
      \State $\Phi' \coloneqq \Phi\substitute{\vec \nu(x)}{x
      : x \in X_1}$
      \For{$p \in \pdiff(\Phi')$}
          $\vec b_{p}'$ $\coloneqq$ solution for $\Phi'$ modulo $p$ 
          \itComment{\Cref{claim:phi-prime:new-primes-are-ok}}
          \label{func:l-t-g:ind:b-p-for-call}
      \EndFor
      \State\label{func:l-t-g:ind:call}
        $\vec \xi \coloneqq$ result of calling~\Cref{func:new-solveincreasing} on $\Phi'$, $X_2 \incord \dots \incord X_r$ and $\{\vec b_p' : p \in \pdiff(\Phi')\}$
      \State \textbf{return} $\vec \nu \sqcup \vec \xi$
      \itComment{union of disjoint functions}
      \label{func:l-t-g:return}
    \EndIf
  \end{algorithmic}
\end{algorithm}

%% file: local-to-global-proof.tex
Throughout the proof, fix and order $(\incord) \in (X_1 \incord \dots \incord X_r)$.
For simplicity, we focus on the part of the statement that
builds a solution over $\NN$ (in fact, we will build a solution over $\pZZ$).
The fact that there are
infinitely many solutions follows from the
fact that the solution is built by solely relying on systems
of (non-)congruences over the integers.

Let us first expand on the overview of the proof given in~\Cref{sec:intro-local-global} by referring to the pseudocode in~\Cref{func:new-solveincreasing}.
The goal is to compute a map $\vec \nu \colon
\big(\bigcup_{j=1}^{r} X_j\big) \to \pZZ$ such that $\vec
\nu(\vec x)$ is a solution for $\Phi$. The proof proceeds by
induction on $r$, populating the map $\vec \nu$ according
the order $\incord$. 

When $r = 1$ (line~\ref{func:l-t-g:base-case}
in~\Cref{func:new-solveincreasing}) $\vec \nu$ can be
computed using the (standard) Chinese remainder theorem,
with little to no problem (line~\ref{func:l-t-g:CRT}). The
main ingredient here is given by the natural number $\mu_p
\coloneqq  \max \big\{ v_p(f(\vec b_p)) : f(\vec x) \text{
left-hand side of a divisibility in } \Phi \}$
(line~\ref{func:l-t-g:base-case:mu}), that given $p \in
\pdiff(\Phi)$ tells us up to what power of~$p$ should the
integer solution given by $\vec \nu$ agree with the solution
$\vec b_p$.

When $r \geq 2$, the
goal is to define $\vec \nu$ for the variables in $X_1$ in
such a way that the formula $\Phi' \coloneqq
\Phi\substitute{\vec \nu(x)}{x : x \in X_1}$ is increasing
for $X_2 \incord \dots \incord X_r$, and has solutions
modulo $p$ for every $p \in \pdiff(\Phi')$. This allows us
to call for~\Cref{theorem:local-to-global} inductively
(line~\ref{func:l-t-g:ind:call}), obtaining a solution $\vec
\xi  \colon \big(\bigcup_{j=2}^{r} X_j\big) \to \pZZ$ for
$\Phi'$. An integer solution for $\Phi$ is then given by the
union $\vec \nu \sqcup \vec \xi$ of $\vec \nu$ and $\vec
\xi$, i.e., the map defined as $\vec \nu(x)$ for $x \in X_1$
and as $\vec \xi(y)$ for $y \in \bigcup_{j=2}^r X_j$,
(line~\ref{func:l-t-g:return}). To construct $\vec \nu$ for
$X_1$, we first close $\Phi$ under the elimination property
following~\Cref{lemma:add-elimination-property}~(line~\ref{func:l-t-g:ind:elim-prop}),
and extend the solutions $\vec b_p$ to every $p \in
\pzero(\Psi)$ thanks to~\Cref{lemma:simple-primes}
(line~\ref{func:l-t-g:ind:new-b-ps}). We then populate $\vec
\nu$ following the order~$\incord$, starting from the
smallest
variable~(line~\ref{func:l-t-g:ind:inner-induction}). In the
proof, this is done with a second induction. Values for the
variables in $X_1$ are found using~\Cref{thm:mixed-crt}
(line~\ref{func:l-t-g:ind:mixed-crt}). When a new value $a
\in \pZZ$ for a variable $x \in X_1$ is found, new primes
need to be taken into account (line~\ref{func:l-t-g:ind:Q}),
since substituting $a$ for $x$ yields a complete evaluation
of the polynomials in $S(\sterms(\Phi))$ with leading
variable~$x$, and the resulting integers might be divisible by
primes not belonging to~$\pzero(\Psi)$. For subsequent
variables in~$X_1$, we make sure to pick values that keep
the evaluated polynomials as ``coprime as possible'' with
respect to these new primes (see the induction
hypothesis~\eqref{th:l-t-g:internal:IH2} below, as well as the
system of \mbox{(non-)congruences} in
line~\ref{func:l-t-g:ind:mixed-crt}). This condition is
necessary to obtain the new solutions~$\vec b_p$ for the
formula $\Phi'$, modulo every $p \in \pdiff(\Phi')$
(line~\ref{func:l-t-g:ind:b-p-for-call}).

We now formalize the proof. To ease the presentation, we
postpone the analysis on the bound of the minimal positive
solution to after the main induction showing the existence
of such a solution. In a nutshell, the bound fundamentally
comes from repeated applications of~\Cref{thm:mixed-crt}.

\proofparagraph{Base case $r = 1$} 

As $\Phi$ is $1$-increasing, it is of the form
${\bigwedge_{i=1}^\ell c_i \div g_i(\vec x) \land
\bigwedge_{j = \ell+1}^m f_j(\vec x) \div a_j \cdot f_j(\vec
x)}$, where every $c_i$ and $a_j$ are in $\ZZ$. By
hypothesis, every $c_i$ and $f_j$ is non-zero. If $c_i = 1$
for every $i \in [1,\ell]$, then $\vec x = \vec 0$ is
trivially a solution. Otherwise, $\pdiff(\Phi)$ is
non-empty. Let $\vec x = (x_1,\dots,x_d)$ and, given $p \in
\pdiff(\Phi)$, let $\mu_p \coloneqq \max \{ v_p(f(\vec b_p))
: f \text{ is in the left-hand side of a divisibility of }
\Phi \}$. Note that since $\vec b_p$ is a solution for
$\Phi$ modulo $p$, we have $f_j(\vec b_p) \neq 0$ for every
$j \in [\ell+1,m]$, and thus $v_p(f(\vec b_p)) \in \NN$.
Denote with $b_{p,k}$ the value of $\vec b_p$ for the
variable $x_k$, with $p \in \pdiff(\Phi)$ and $k \in [1,d]$.
Consider the system of congruences 
\begin{align}
  \label{eq:l-t-g:base-case}
  x_k \equiv b_{p,k} & \pmod {p^{\mu_p+1}} & p \in \pdiff(\Phi),\ 1 \leq k \leq d.
\end{align}
According to the Chinese remainder theorem, this system has
a positive solution $\vec a = (a_1,\dots,a_d)$. To conclude
the base case, it suffices to show that $f_j(\vec a) \neq 0$
for every $j \in [\ell+1,m]$, and that $c_i \div g_i(\vec
a)$ for every $i \in [1,\ell]$. First, consider $j \in
[\ell+1,m]$ and pick a prime $p \in \pdiff(\Phi)$. From the
system of congruences in~\Cref{eq:l-t-g:base-case} we have
$f_j(\vec a) \equiv f_j(\vec b_p) \pmod {p^{\mu_p+1}}$, and
by definition of~$\mu_p$, $f_j(\vec b_p) \not\equiv 0 \pmod
{p^{\mu_p+1}}$. We conclude that $f_j(\vec a) \not\equiv 0
\pmod {p^{\mu_p+1}}$, and so $f_j(\vec a) \neq 0$.

Consider now~$i \in [1,\ell]$. To prove that $c_i \div
g_i(\vec a)$, concluding the base case, we show that for
every prime $p$ dividing~$c_i$, ${v_p(c_i) \leq v_p(g_i(\vec
a))}$. By definition, any such prime $p$ satisfies $p \in
\pdiff(\Phi)$ and moreover ${v_p(c_i) \leq \mu_p}$. We
distinguish two cases:
\begin{itemize}
  \item if $v_p(g_i(\vec b_p)) \leq \mu_p$, then according
  to~\Cref{eq:l-t-g:base-case} we have $v_p(g_i(\vec b_p)) =
  v_p(g_i(\vec a))$. Since $\vec b_p$ is a solution for
  $\Phi$ modulo $p$, this implies~$v_p(c_i) \leq
  v_p(g_i(\vec a))$.
  \item If $v_p(g_i(\vec b_p)) > \mu_p$, then $g_i(\vec b_p)
  \equiv 0 \pmod {p^{\mu_p+1}}$ and so~$g_i(\vec a) \equiv 0
  \pmod {p^{\mu_p+1}}$ by~\Cref{eq:l-t-g:base-case}.
  Therefore $v_p(g_i(\vec a)) > \mu_p$ and by definition of
  $\mu_p$ we get $v_p(c_i) \leq v_p(g_i(\vec a))$.
\end{itemize}

\proofparagraph{Induction step $r \geq 2$} 
by induction hypothesis, we assume the theorem to be true
for every $s$-increasing system with $s < r$.
By~\Cref{lemma:simple-primes}, for every prime $p \in \PP
\setminus \pdiff(\Phi)$ there is a solution $\vec b_p$ for
$\Phi$ modulo $p$ such that $\max \{ v_p(f(\vec b_p))  : f
\text{ in the left-hand side of a divisibility of } \Phi \}
= 0$. Together with the solutions $\vec b_p$ for primes $p
\in \pdiff(\Phi)$, this means that $\Phi$ has solutions
modulo every prime. We
apply~\Cref{lemma:add-elimination-property} in order to
obtain from $\Phi$ a system $\Psi$ with the elimination
property for~$\incord$. The system $\Psi$ is used to produce
the map $\vec \nu$ for the variables in $X_1$. Adding the
elimination property does not change the set of solutions
(neither over the integers nor modulo a prime), and
therefore the above solutions $\vec b_p$ are still solutions
for $\Psi$ modulo~$p$. Below, among these solutions we only
consider the ones for primes $p \in \pzero(\Psi)$. Given
such a prime~$p \in \pzero(\Psi)$, define $\mu_p \coloneqq
\max \{ v_p(f(\vec b_p)) : f \text{ is in the left-hand side
of a divisibility of } \Psi \}$. As already observed in the
base case, given $f$ left-hand side of a divisibility in
$\Psi$, $f(\vec b_p) \neq 0$ and so $v_p(f(\vec b_p)) \in
\NN$. Moreover, from~\Cref{lemma:add-elim-property:item-1}
in~\Cref{lemma:add-elimination-property} we conclude that
$\mu_p = 0$ for every $p \in \pzero(\Psi) \setminus
\pdiff(\Phi)$.

As $\Psi$ is $r$-increasing
(see~\Cref{lemma:add-elim-property:item-1}
in~\Cref{lemma:add-elimination-property}), it is of the form 
\begin{equation}
  \label{eq:l-t-g:outer-induction}
  \left(\bigwedge_{i = 1}^\ell c_i \mid g_i(\vec  x) \right) \wedge \left( \bigwedge_{i=\ell+1}^{n}f_i(\vec  x) \mid g_i(\vec  x) + g_i'(\vec  y) \right) \wedge \left( \bigwedge_{i=n+1}^{t} f_i(\vec  x) + f_i'(\vec  y) \mid g_i(\vec  x) + g_i'(\vec  y) \right),
\end{equation}
where $\vec x$ are the variables appearing in $X_1$, $\vec
y$ are the variables appearing in $\bigcup_{j=2}^{r} X_j$,
$\ell \leq n \leq t$, and for every $i \in [n+1,t]$,
$f_i'(\vec y)$ and $g_i'(\vec y)$ have $0$ as a constant and
are non-constant. Moreover, since $\Psi$ is increasing, for
every $i \in [\ell+1,n]$ $g_i(\vec x)$ and $g_i'(\vec y)$
are such that either $g_i' = 0$ and $g_i = a \cdot f_i$ for
some $a \in \ZZ$, or $g_i'$ is non-constant and has $0$ as a
constant. Let $X_1 = \{x_1,\dots, x_d \}$, with $x_1
\incord \dots \incord x_d$. Denote by $b_{p,k}$ the value of
$\vec b_p$ for the variable $x_k$, with $p \in \pzero(\Psi)$
and $k \in [1,d]$. We build the map $\vec \nu$ defined on
the variables in $X_1$, inductively starting from $x_1$. In
the induction step, when searching for a value to the
variable $x_{k+1}$, the following induction hypotheses hold:
\begin{description}
  \item[\textlabel{IH1}{th:l-t-g:internal:IH1}:] 
  For every $p \in \pzero(\Psi)$ and $j \in [1,k]$, \ $\vec
  \nu(x_j) \equiv b_{p,j} \pmod {p^{\mu_p + 1}}$,
  \item[\textlabel{IH2}{th:l-t-g:internal:IH2}:] 
  For every prime $p \notin \pzero(\Psi)$, for every $h,h'
  \in \sterms(\Psi)$ with leading variable at most $x_k$, if
  $S(h, h')$ is not identically zero, then $p$ does not
  divide both $h(\vec \nu (x_1,\dots,x_k))$ and $h'(\vec \nu
  (x_1,\dots,x_k))$. 
  \item[\textlabel{IH3}{th:l-t-g:internal:IH3}:]
  $h(\vec \nu (x_1,\dots,x_k)) \neq 0$ for every $h \in
  \sterms(\Psi)$ that is non-zero and with $\lv(h) \incordeq
  x_k$.
\end{description}

\begin{description}
  \item[base case $k = 0$.]  
    In this case, \eqref{th:l-t-g:internal:IH1}
    and~\eqref{th:l-t-g:internal:IH3} trivially hold
    (for~\eqref{th:l-t-g:internal:IH3} note that $h$ is
    constant). In~\eqref{th:l-t-g:internal:IH2} we only
    consider constant polynomials $h,h'$, hence $S(h,h') =
    0$ by definition.

  \item[induction step.] Let us assume that $\vec \nu$ is
  defined for the variables $x_1,\dots,x_{k}$ with $k \in
  [0,d-1]$, so that the induction hypotheses hold. Let us
  provide a value for $x_{k+1}$ so that $\vec \nu$ still
  fulfils the induction hypotheses. We define the following
  set of primes: 
  \[
    P_k \coloneqq \left\{ p \in \PP : p \in \pzero(\Psi) \text{ or } p \mid h(\vec  \nu(x_1,\dots,x_k)) \text{ for }h \in S(\sterms{(\Psi)}) \backslash \{0\} \text{ with }\lv(h) \incordeq x_k \right\}.
  \]
  In the hypothesis that $P_k = \pzero(\Psi)$, we add to
  $P_k$ the smallest prime not in $\pzero(\Psi)$. Hence,
  below, assume $P_k \neq \pzero(\Psi)$. We consider the
  following system of (non-)congruences:
  \begin{align*}
    x_{k+1} &\equiv b_{p, k+1} && \pmod {p^{\mu_p+1}} & p \in \pzero(\Psi)\\
    h(\vec \nu(x_1,\dots,x_k),\,x_{k+1}) &\not\equiv 0 && \pmod q &
    q \in P_k \setminus \pzero(\Psi) \text{ and}\\
    &&&& h \in S(\sterms(\Psi)) \text{ s.t.}~\lv(h) = x_{k+1}.
  \end{align*}
  With respect to the $h$ above, let us write $h(\vec
  \nu(x_1,\dots,x_k),x_{k+1}) = c_h + a_{h} \cdot x_{k+1}$,
  where $c_h$ is the constant term obtained by partially
  evaluating $h$ with respect to $\vec \nu(x_1,\dots,x_k)$,
  and $a_h$ is the coefficient of $x_{k+1}$ in $h$. Since $q
  \in P_k \setminus \pzero(\Psi)$, then $q \nmid a_h$ from
  Condition~\ref{pzero:2}. Then $a_h$ has an inverse
  $a_h^{-1}$ modulo $q$, and the system of (non-)congruences
  above is equivalent to
  \begin{equation}
    \label{eq:l-to-g:non-cong-sys}
    \begin{aligned}
      x_{k+1} &\equiv b_{p, k+1} && \pmod {p^{\mu_p+1}} & p \in \pzero(\Psi)\\
      x_{k+1} &\not\equiv -a_h^{-1} c_h && \pmod q &
      \hspace{-10pt}q \in P_k \setminus \pzero(\Psi) \text{ and } h \in S(\sterms(\Psi)) \text{ s.t.}~\lv(h) = x_{k+1}.
    \end{aligned}
  \end{equation}
  In this system of (non-)congruences, elements in
  $\pzero(\Psi)$ and $P_k \setminus \pzero(\Psi)$ are
  pairwise coprime, $P_k \setminus \pzero(\Psi)$ is a set of
  primes, and moreover $\min(P_k \setminus \pzero(\Psi)) >
  \card{S(\sterms(\Psi))}$ by Condition~\ref{pzero:1}.
  Hence, we can apply~\Cref{thm:mixed-crt} and conclude
  that~\Cref{eq:l-to-g:non-cong-sys} has a solution $w \in
  \pZZ$. Let us update $\vec \nu$ so that $\vec \nu(x_{k+1})
  = w$. We show that $\vec \nu$ satisfies the induction
  hypotheses.
  \begin{enumerate}
    \item By the congruences
    in~\Cref{eq:l-to-g:non-cong-sys}, $\vec \nu(x_{k+1})
    \equiv b_{p,k+1} \pmod {p^{\mu_p+1}}$,
    hence~\eqref{th:l-t-g:internal:IH1} holds.
    \item Consider $h,h' \in \Delta(\Psi)$ such that $\lv(h)
    \incordeq \lv(h') = x_{k+1}$ and $S(h,h')$ is not
    identically zero. Note that the case where $\lv(h')
    \incordeq \lv(h) = x_{k+1}$ is analogous, whereas if
    both $\lv(h)$ and $\lv(h')$ are at most $x_{k}$
    then~\eqref{th:l-t-g:internal:IH2} already holds by
    induction hypothesis. We divide the proof into two
    cases, depending on $\lv(h)$. 
    \begin{itemize}
      \item if $\lv(h) \incord x_{k+1}$, consider $p \not\in
      \pzero(\Psi)$ such that $p \div h(\vec
      \nu(x_1,\dots,x_k))$. By definition, $p \in P_k$, and
      thus from the non-congruences
      in~\Cref{eq:l-to-g:non-cong-sys}, $p \nmid  h(\vec
      \nu(x_1,\dots,x_{k+1}))$.
      \item if $\lv(h) = \lv(h') = x_{k+1}$, assume \emph{ad
      absurdum} that there is $p \not \in \pzero(\Psi)$ such
      that $p \div h(\vec \nu(x_1,\dots,x_{k+1}))$ and $p
      \div h'(\vec \nu(x_1,\dots,x_{k+1}))$. Then, $p \div
      S(h,h')$ by definition of~$S$. However, $S(h,h') \in
      S(\sterms(\Psi)) \setminus \{0\}$ and $\lv(S(h,h'))
      \incordeq x_{k}$, from which we conclude that $p \in
      P_k$. Again from the non-congruences
      in~\Cref{eq:l-to-g:non-cong-sys}, this implies $p
      \nmid h(\vec \nu(x_1,\dots,x_{k+1}))$ and $p \nmid
      h'(\vec \nu(x_1,\dots,x_{k+1}))$, a contradiction.
    \end{itemize}
    In both cases, we conclude
    that~\eqref{th:l-t-g:internal:IH2} holds.
    \item Let $h \in \sterms(\Psi)$ with $\lv(h) = x_{k+1}$
    (else~\eqref{th:l-t-g:internal:IH3} directly holds by
    induction hypothesis). As there is a prime $p \in P_k
    \setminus \pzero(\Psi)$, from the non-congruences
    of~\Cref{eq:l-to-g:non-cong-sys} we conclude $p \nmid
    h(\vec \nu(x_1,\dots,x_{k+1}))$, and thus $h(\vec
    \nu(x_1,\dots,x_{k+1}))$ cannot be $0$.
    Hence,~\eqref{th:l-t-g:internal:IH3} holds.
  \end{enumerate}
\end{description}

The innermost induction we have just completed yields a map
$\vec \nu$ defined for the variables in $X_1$ and
satisfying~\eqref{th:l-t-g:internal:IH1}--\eqref{th:l-t-g:internal:IH3}
for every $k \in [1,d]$. Consider the system $\Psi'(\vec y)
\coloneqq \Psi\substitute{\vec \nu(x)}{x : x \in X_1}$
obtained from $\Psi$ by evaluating as $\vec \nu(x)$ every
variable $x$ in $X_1$. With reference
to~\Cref{eq:l-t-g:outer-induction}, we note that the
subsystem $\bigwedge_{i=1}^\ell c_i \div g_i(\vec \nu(\vec
x))$ evaluates to true (proof as in the base case $r=1$ of
the induction and by using~\eqref{th:l-t-g:internal:IH1}).
Then, $\Psi'(\vec y)$ is of the form 
\begin{equation}
  \label{eq:l-t-g:outer-induction-psi}
  \left( \bigwedge_{i=\ell+1}^{n} \alpha_i \mid \beta_i + g_i'(\vec  y) \right) \wedge 
  \left( \bigwedge_{i=n+1}^{t} \alpha_i + f_i'(\vec  y) \mid \beta_i  + g_i'(\vec  y) \right),
\end{equation}
where $\alpha_i = f_i(\vec \nu(\vec x)) \in \ZZ$ and
$\beta_i = g_i(\vec \nu(\vec x)) \in \ZZ$, for every $i \in
[\ell+1,t]$. Note that $\alpha_i \neq 0$ for every $i \in
[\ell+1,n]$, thanks to~\eqref{th:l-t-g:internal:IH3},
so~$\vec \nu$ satisfies all trivial divisibilities of the
form $f(\vec x) \div a \cdot f(\vec x)$.

The next step is to show that $\Psi'$ is increasing for
$(X_2 \incord \dots \incord X_r)$ and to provide solutions
modulo~$p$ for every $p \in \pzero(\Psi')$. These two
properties, formalized below
in~\Cref{claim:still-increasing}
and~\Cref{claim:new-primes-are-ok}, follow from the
induction
hypotheses~\eqref{th:l-t-g:internal:IH1}--\eqref{th:l-t-g:internal:IH3}
we kept during the construction of $\vec \nu$, together with
the fact that the system~$\Psi$ has the elimination
property. Their proofs are very technical and lengthy, and
we therefore defer them
to~\Cref{appendix:lemmas-local-to-global}. 
Observe that the condition~\ref{pzero:3} of the difficult primes is required to establish~\Cref{claim:new-primes-are-ok}, but otherwise does not appear anywhere else in this proof.

\begin{restatable}{claim}{ClaimStillIncreasing}
  \label{claim:still-increasing}
  The system $\Psi'$ is increasing for $(X_2 \incord \dots
  \incord X_r)$.
\end{restatable}

\begin{restatable}{claim}{ClaimNewPrimesAreOk}
  \label{claim:new-primes-are-ok}
  For every $p \in \pzero(\Psi)$, the solution~$\vec b_p$
  for $\Psi$ modulo $p$ is, when restricted to~$\vec y$, a
  solution for~$\Psi'(\vec y)$ modulo~$p$. For every
  prime~${p \not \in \pzero(\Psi)}$, there is a solution
  $\vec b_p$ for $\Psi'$ modulo $p$ such that (i)~every
  entry of~$\vec b_p$ belongs to~${[0,p^{u+1}-1]}$, where
  ${u \coloneqq \max\{ v_p(\alpha_i) : i \in
  [\ell\,{+}\,1,n] \}}$, and (ii)~for every $g \in
  \terms(\Psi')$, $v_p(g(\vec b_p))$ is either $0$ or $u$.
\end{restatable}

Thanks to~\Cref{claim:still-increasing}
and~\Cref{claim:new-primes-are-ok}, we can inductively apply
the statement of~\Cref{theorem:local-to-global} on~$\Psi'$
in order to obtain an integer solution for $\Psi$, and thus
a solution for the original system $\Phi$. While this would
prove the local-to-global property, it is not enough to
obtain the upper bound on the size of the minimal positive
solution stated in~\Cref{theorem:local-to-global}. Instead,
we wish to apply the induction hypothesis on the system
$\Phi'(\vec y) \coloneqq \Phi\substitute{\vec \nu(\vec x)}{x
: x \in X_1}$, hence disregarding the work done to close
$\Phi$ under the elimination property. The main point in
favour of this strategy is that the subsequent applications
of~\Cref{lemma:add-elimination-property}, required to
inductively construct the integer solutions for the
remaining variables $\vec y$, yield smaller systems of
divisibility constraints (for instance, note that~$\Phi'$ has at
most~$m$ divisibilities, whereas~$\Psi'$ can have close to
$m \cdot (d+2)$ divisibilities).

To prove that we can apply the induction hypothesis on
$\Phi'$, we need to show that this system satisfies
properties analogous to the ones
in~\Cref{claim:still-increasing}
and~\Cref{claim:new-primes-are-ok}. While the proofs of
these claims require the elimination property to be
established, we can transfer them to $\Phi'$ thanks to the
fact that $\Psi$ is defined from~$\Phi$ following the
algorithm of~\Cref{lemma:add-elimination-property}.

\begin{claim}
  \label{claim:phi-prime:still-increasing}
  The system $\Phi'$ is increasing for $(X_2 \incord \dots
  \incord X_r)$.
\end{claim}

\begin{proof}
  \emph{Ad absurdum}, assume that $\Phi'(\vec y)$ is not
  increasing for some order $(\incord') \in (X_2 \incord
  \dots \incord X_r)$. Let $\vec y = (y_1,\dots,y_j)$ with
  $y_1 \incord' \dots \incord' y_j$. There is $i \in [1,j]$
  and a primitive term $f$ with $\lv(f) = y_i$ such that
  $\ZZ f \varsubsetneq \module_f(\Phi') \cap
  \ZZ[y_1,\dots,y_i]$. By~\Cref{lemma:substit-and-elim} we
  get $\ZZ f \varsubsetneq \module_f(\Psi') \cap
  \ZZ[y_1,\dots,y_i]$. However, this implies that $\Psi'$ is
  not increasing for $\incord'$,
  contradicting~\Cref{claim:still-increasing}.
\end{proof}

\begin{claim}
  \label{claim:phi-prime:new-primes-are-ok}
  For every $p \in \PP$, the solution $\vec b_p$ for $\Psi'$
  modulo $p$ ensured in~\Cref{claim:new-primes-are-ok} is
  also a solution for $\Phi'$ modulo $p$. If $p \not \in
  \pzero(\Psi)$, then for every polynomial $f'$ appearing in
  the left-hand side of a divisibility of $\Phi'$, we have
  either $v_p(f'(\vec b_p)) = 0$ or $v_p(f'(\vec b_p)) =
  \max\{ v_p(\alpha_i) : i \in [\ell\,{+}\,1,n] \}$.

\end{claim}

\begin{proof}
  For the first statement of the claim, consider a solution
  $\vec b_p$ for $\Psi'(\vec y)$ modulo $p$ (such as the one
  ensured by~\Cref{claim:new-primes-are-ok}). From the
  definition of $\Psi'$, the tuple $(\vec \nu(\vec x), \vec
  b_p)$ is a solution for $\Psi(\vec x, \vec y)$ modulo $p$.
  Then, by~\Cref{lemma:add-elimination-property}, $(\vec
  \nu(\vec x), \vec b_p)$ is a solution for $\Phi(\vec x,
  \vec y)$ modulo $p$; and so by definition of $\Phi'$,
  $\vec b_p$ is a solution for $\Phi'(\vec y)$ modulo $p$.

  The second statement of this claim follows
  from~\Cref{claim:new-primes-are-ok} together with the
  property~\eqref{lemma:add-elim-property:item-1} of
  \Cref{lemma:add-elimination-property}, and  
  by definitions of $\Psi'$ and $\Phi'$. In particular, for
  every polynomial $f'(\vec y)$ occurring in a left-hand
  side of a divisibility of~$\Phi'$, there is a polynomial
  $f(\vec x, \vec y)$ occurring in a left-hand side of
  $\Phi$ such that $f'(\vec y) = f(\vec \nu(\vec x), \vec
  y)$. From~\eqref{lemma:add-elim-property:item-1} of
  \Cref{lemma:add-elimination-property}, $f$ occurs in a
  left-hand side of $\Psi$ and thus $f'$ occurs in a
  left-hand side of $\Psi'$. The statement then follows
  by~\Cref{claim:new-primes-are-ok}. 
\end{proof}

From~\Cref{claim:phi-prime:still-increasing}
and~\Cref{claim:phi-prime:new-primes-are-ok}, and by
induction hypothesis, there is a map $\vec \xi \colon
\big(\bigcup_{j=2}^r X_j \big) \to \ZZ_+$ such that $\vec
\xi(\vec y)$ is a solution for $\Phi'$. Note that in
constructing $\vec \xi$ we can rely on the order~$\incord$
restricted to~$\bigcup_{j=2}^r X_j$; since $\Phi'$ is
increasing for that order. Then, by definition of $\Phi'$, a
positive integer solution for $\Phi$ is given by the union
$\vec \nu \sqcup \vec \xi$ of $\vec \nu$ and $\vec \xi$.
This concludes the proof of existence of a solution. We now
study its bit length.

In what follows, let $\underline{O} \in \pZZ$ be the minimal
positive integer greater or equal than $4$ such that the map
${x \mapsto \underline{O}(x+1)}$ upper bounds the linear
functions hidden in the $\bigO{.}$ appearing
in~\Cref{lemma:add-elimination-property}. We
write~$\Gamma(r,\ell,w,m,d)$, with $r,\ell,w,m,d \in \pZZ$
and $r \leq d$, for the maximum bit length of the minimal
positive solution of any system of divisibility constraints $\Phi$
such that:
\begin{itemize}
  \item $\Phi$ is $r$-increasing.
  \item The maximum bit length of a coefficient or constant
  appearing in $\Phi$, i.e., $\maxbl{\Phi}$, is at most
  $\ell$.
  \item For every $p \in \pdiff(\Phi)$, consider a solution
  $\vec b_p$ of $\Phi$ modulo $p$ minimizing $\mu_p
  \coloneqq \max\{{v_p(f(\vec b_p))} : f \text{ is in the
  left-hand side of a divisibility in } \Phi \}$. Then,
  $\log_2\left(\prod_{p \in \pdiff(\Phi)} p^{\mu_p + 1}
  \right) \leq  w$.
  \item $\Phi$ has at most $m$ divisibilities.
  \item $\Phi$ has at most $d$ variables.
\end{itemize}
The constraint $r \leq d$ is without loss of
generality, as every increasing formula is $d$-increasing.

Since we want to find an upper bound for $\Gamma$, assume
without loss of generality that $\Gamma(r,\ell,w,m,d)$ is
always at least $\min(\ell,w)$.
In~\Cref{appendix:l-t-g:bounds-on-gamma} we study the growth
of $\Gamma$ and prove the following claim.

\begin{restatable}{claim}{EqGammaInductiveBound}
  \label{eq:gamma-inductive-bound}
  $\begin{cases} \Gamma(1,\ell,w,m,d) \leq w + 3\\
    \Gamma(r+1,\ell,w,m,d) \leq  
    \begin{aligned}[t]
        \Gamma(&r,\\
          & 2^{105}  m^{27}  (d+2)^{38} 
          \underline{O} \cdot \log_2( \underline{O})^6 
          (\ell + w) \cdot (\log_2(\ell + w))^6,\\
          &  2^{109}  m^{29}  (d+2)^{39} 
          \underline{O} \cdot \log_2( \underline{O})^6 
          (\ell + w) \cdot (\log_2(\ell + w))^6,\\
          &m,\\
          &d).
    \end{aligned}
  \end{cases}$
\end{restatable}

\noindent
Let us show that the recurrence system above yields the
bound in the statement of the theorem. Remark that $\Gamma$ is
monotonous by definition. By induction on $k \in [0,r-1]$ we
show that 
\[ 
  \Gamma(r,\ell,w,m,d) \leq \Gamma(r-k,\delta_k,\delta_k,m,d) \text{ where } 
  \delta_k \coloneqq \frac{1}{2} \cdot (2^{110}  m^{29}  (d+2)^{39}  \underline{O} \cdot \log_2( \underline{O})^6  (\ell+w))^{2(k+1)}.
\]
\begin{description}
  \item[base case $k = 0$.] Directly follows from $\delta_0
  \geq \max(\ell,w)$ and the fact that $\Gamma$ is
  monotonous.
  \item[induction case $k \geq 1$.] Let us define $C
  \coloneqq 2^{110}  m^{29}  (d+2)^{39} 
  \underline{O} \cdot \log_2( \underline{O})^6$. By
  induction hypothesis, $\Gamma(r,\ell,w,m,d) \leq
  \Gamma(r-(k-1), \delta_{k-1},\delta_{k-1}, m, d)$. 
  By~\Cref{eq:gamma-inductive-bound} 
  and the monotonicity of $\Gamma$:
  \begin{align*}
    & {\Gamma(r-(k-1),
    \delta_{k-1},\delta_{k-1}, m, d)}\\
    \leq\,&\Gamma(r-k, \ \frac{C}{2} \cdot (2 \cdot \delta_{k-1}) \cdot (\log_2(2 \cdot \delta_{k-1}))^6, \ \frac{C}{2} \cdot (2 \cdot \delta_{k-1}) \cdot (\log_2(2 \cdot \delta_{k-1}))^6, \ m, \ d)\\
    \leq\,&\Gamma(r-k, \ \delta_{k},\ \delta_{k} \ m, \ d),
  \end{align*}
  as indeed 
  \begin{align*}
    & \frac{C}{2} \cdot (2 \cdot \delta_{k-1}) \cdot (\log_2(2 \cdot \delta_{k-1}))^6\\
    \leq\,& 
    \frac{C}{2} \cdot \big(C \cdot (\ell+w)\big)^{2k}  \big(\log_2((C \cdot (\ell+w))^{2k})\big)^6\\
    \leq\,& 
    \frac{C}{2} \cdot \big(C \cdot (\ell+w)\big)^{2k}  (2 \cdot k)^6 \log_2(C \cdot (\ell+w))^6\\
    \leq\,& 
    \frac{C}{2} \cdot \big(C \cdot (\ell+w)\big)^{2k} \cdot \sqrt{C} \cdot \log_2(C \cdot (\ell+w))^6
    &\text{from $k < r \leq d$ and $(2 \cdot d)^6 \leq \sqrt{C}$}\\
    \leq\,& 
    \frac{C}{2} \cdot \big(C \cdot (\ell+w)\big)^{2k} \cdot \sqrt{C} \cdot \sqrt{C \cdot (\ell+w)}
    &\text{from $\log_2(x)^6 \leq \sqrt{x}$ for $x \geq 2^{75}$}\\
    \leq\,& 
    \frac{1}{2} \cdot \big(C \cdot (\ell+w)\big)^{2(k+1)} = \delta_{k}.
  \end{align*}
\end{description}
The inequality we just showed, together with the base case
of the recurrence system, entails 
\begin{equation}
  \label{equation:the-ultimate-bound-on-Gamma}
  \Gamma(r,\ell,w,m,d) \leq (2^{110} m^{29} (d+2)^{39} \underline{O} \cdot \log_2( \underline{O})^6 (\ell+w))^{2 \cdot r}.
\end{equation}

Take now the formula $\Phi$ in the statement of the theorem.
This formula belongs to $\Gamma(r,\ell,w,m,d)$ where $\ell
\coloneqq \maxbl{\Phi}$ and $w \coloneqq \log_2\big(\prod_{p
\in \pdiff(\Phi)} p^{\mu_p + 1} \big)$. We have
\begin{align*}
  w &\leq \max\{ 1+v_p(f(\vec b_p)) : f \text{ is in a left-hand side of } \Phi,\, p \in \pdiff(\Phi)\} \cdot \log_2\Big(\prod_{p \in \pdiff(\Phi)} p\Big)\\
  &\leq \max\{ \bitlength{f(\vec b_p)} : f \text{ is in a left-hand side of } \Phi,\, p \in \pdiff(\Phi)\} \cdot \log_2\Big(\prod_{p \in \pdiff(\Phi)} p\Big)\\
  &\leq (\max\{ \maxbl{\vec b_p} : p \in \pdiff(\Phi)\} + \maxbl{\Phi} + d+1) \cdot \log_2\Big(\prod_{p \in \pdiff(\Phi)} p\Big)\\  
  &\leq (\max\{ \maxbl{\vec b_p} : p \in \pdiff(\Phi)\} + \maxbl{\Phi} + d+1) \cdot m^2 (d+2) \cdot (\maxbl{\Phi}+2)
  &\text{\Cref{lemma:bound-on-pzero}}\\
  &\leq (\max\{ \maxbl{\vec b_p} : p \in \pdiff(\Phi)\} + 1)  \cdot m^2 (d+2)^2 (\maxbl{\Phi}+2)^2.
\end{align*}
Then, following~\Cref{equation:the-ultimate-bound-on-Gamma},
the minimal positive solution of $\Phi$ is bounded by 
\begin{align*}
  \left(2^{111} \underline{O} \cdot \log_2( \underline{O})^6 m^{31} (d+2)^{41} (\maxbl{\Phi}+2)^2 (\max\{ \maxbl{\vec b_p} : p \in \pdiff(\Phi)\} + 2)\right)^{2 r},
\end{align*}
which is in $(\bitlength{\Phi} + \max\{ \maxbl{\vec b_p} : p
\in \pdiff(\Phi) \})^{O(r)}$.

%% file: sec-ip-gcd.tex
\section{IP-GCD systems have polynomial size solutions}
\label{sec:ip-gcd-small-model}

In this section we expand the summary
provided~\Cref{sec:intro-gcd-ip} and
establish~\Cref{thm:small-model}, i.e., that every feasible
IP-GCD system has solutions of polynomial bit length, and
that this polynomial bound still holds when looking at
minimization or maximization of linear objectives. As
explained in~\Cref{sec:intro-gcd-ip},
we prove~\Cref{thm:small-model} by designing an algorithm
that reduces an IP-GCD system into a disjunction of
(exponentially many) $3$-increasing systems of
divisibility constraints with coefficients and constants of polynomial
size, to then study bounds on their solutions modulo primes.
Then, the polynomial small witness property follows
from~\Cref{theorem:local-to-global}.

Without loss of generality, throughout the section we
consider IP-GCD systems of the form 
\[ 
  A \cdot \vec x \le \vec b \land \bigwedge_{i=1}^k \gcd(y_i,z_i) \sim_i c_i\,,
\]
where, $A\in \ZZ^{m\times d}$, $\vec b\in \ZZ^m$, $c_i \in
\pZZ$, $\vec x=(x_1,\ldots,x_d)$ is a vector of variables,
$\sim_i{} \in \{{\leq},{=},{\neq},$ ${\geq}\}$, and the
$y_i$ and $z_i$ are variables occurring in $\vec x$. Systems
with GCD constraints~$\gcd(f(\vec w), g(\vec w)) \sim c$ can
be put into this form by replacing~$\gcd(f(\vec w), g(\vec
w)) \sim c$ with $y = f(\vec w) \land z = g(\vec w) \land
\gcd(y,z) \sim c$, where $y$ and $z$ are fresh variables.

\subsection{Translation into 3-increasing systems}
\label{sec:trans-into-3-inc-sys}

The procedure generating the $3$-increasing systems of
divisibility constraints starting from an IP-GCD system~$\Phi$ is
 divided into two steps: we first (\Cref{algorithm:from-IP-GCD-to-systems-of-divisibilities}) compute several
systems of divisibility constraints whose disjunction is equivalent
to~$\Phi$ (under some changes of variables). We now
describe these two steps in detail, and study bounds on the
obtained $3$-increasing formulae
(\Cref{lemma:from-IP-GCD-to-increasing-systems}).
Both~steps
rely on the following notion of \gcdtodiv triple, which
highlights properties of the system of divisibility constraints
obtained by translation from IP-GCD systems. A triple
$(\Psi,\vec u, E)$ is said to be a \emph{\gcdtodiv} triple
whenever there are $d,m \in \NN$ and three disjoint families
of variables~$\vec z$, $\vec y$ and $\vec w$ for which the following properties hold:

\begin{enumerate}
  \item\label{lemma:gcd-to-div:item0} $\Psi(\vec z, \vec y,
  \vec w)$ is a system of divisibility constraints in $m$ variables,
  $\vec u \in \ZZ^d$ and $E \in \ZZ^{d \times m}$, where
  each column of $E$ (implicitly) corresponds to a variable
  in $\Psi$.
  \item\label{lemma:gcd-to-div:item1} Each divisibility in
  $\Psi$ is of the form $h(\vec z) \div f(\vec y)$ or of the
  form $f(\vec y) \div g(\vec w)$, with $g$ being a
  non-constant polynomial. Each polynomial only features
  non-negative coefficients and constants, and each
  left-hand side of a divisibility has a (strictly) positive
  constant.
  \item\label{lemma:gcd-to-div:item2} In $\Psi$, each
    variable in $\vec z$ appears in a single
    polynomial~$h(\vec z)$, where $h(\vec z)$ is of the form
    $z+c$, for some $c \in \pZZ$, and occurs in precisely
    two divisibilities (as left-hand side).
  \item\label{lemma:gcd-to-div:item2.5} In $\Psi$, each
    variable in $\vec w$ appears in exactly two polynomials
    $g_1(\vec w)$ and $g_2(\vec w)$, each occurring in
    $\Psi$ exactly once (as right-hand sides). They have the
    form $g_1(\vec w) = w$ and $g_2(\vec w) = w+c$, for some
    $c \in \pZZ$.
  \item\label{lemma:gcd-to-div:item3} Every column in~$E$
  corresponding to a variable in $\vec z$ or $\vec w$ is
  zero (see line~\ref{algo:ip-to-gcd-a1:add-to-B}
  of~\Cref{algorithm:from-IP-GCD-to-systems-of-divisibilities}).
\end{enumerate}
For a set of \gcdtodiv triples $S$, let $\sem{S} \coloneqq
\{\vec u + E \cdot \vec \lambda \, : \, {(\Psi,\vec u, E)
\in B}$ $\text{and } \vec \lambda \in \NN^{m} \text{
solution to } \Psi\}$.

\paragraph*{Step I: from IP-GCD to systems of divisibility constraints.}

\input{gcd-to-divisibility-algo}

This step is implemented
by~\Cref{algorithm:from-IP-GCD-to-systems-of-divisibilities}.
As highlighted in its signature, given as input an IP-GCD
system~$\Phi(\vec x)$ having $d$ variables and~$k$
GCD~constraints, this procedure returns a set~$B$
of~\gcdtodiv triples satisfying the equivalence $\{\vec a
\in \ZZ^d : \vec a \text{ solution to } \Phi \} = \sem{B}$.
This equivalence clarifies the role of the vector $\vec u$
and matrix $E$ of a~\gcdtodiv triple~$(\Psi, \vec u, E)$:
they are used to perform a change of variables between the
variables~$(\vec z, \vec y, \vec w)$ in $\Psi$ and the
variables~$\vec x$ in $\Phi$. Note that, according to the
definition of $\sem{B}$, the values of $(\vec z, \vec y,
\vec w)$ range over $\NN$ instead of $\ZZ$. This discrepancy
stems from the use of the forthcoming~\Cref{prop:vzGS}. 

Let us discuss
how~\Cref{algorithm:from-IP-GCD-to-systems-of-divisibilities}
computes $B$. As a preliminary step, the procedure computes
the formula $\bigvee_{i = 1}^\ell \Psi_i$ in
line~\ref{algo:ip-to-gcd-a1:setA}. The role of this formula
is to reduce the problem of translating IP-GCD systems into
systems of divisibility constraints to only those systems in which the
GCD~constraints $\gcd(y,z) \leq c$ and $\gcd(y,z) \neq c$ do
not appear, and given a GCD~constraint $\gcd(y,z) \sim c$
(with $\sim$ either $=$ or $\geq$), the variables $y$ and
$z$ are forced to be positive or negative (in particular,
non-zero). The formula $\bigvee_{i = 1}^\ell \Psi_i$ can be
computed from $\Phi$ by opportunely applying the following
tautologies:
\begin{flalign*}
  &y \leq -1 \lor y = 0 \lor y \geq 1\,, &\hspace{-4cm}
  \gcd(y,z) \neq c+2 \iff \gcd(y,z) \leq c+1 \lor \gcd(y,z)
  \geq c+3 \ \ (c \in \NN)\,,\\
  &\gcd(y,z) \neq 1 \iff y = z = 0 \lor \gcd(y,z) \geq 2\,,
  &\hspace{-15cm}\gcd(y,z) \leq c \iff \bigvee_{j=1}^{c}
  \gcd(y,z) = j\,,\\
  &y = 0 \implies (\gcd(y,z) \sim c \iff |z| \sim c)\,,
  &\hspace{-15cm}y \neq 0 \land z = 0 \implies (\gcd(y,z)
  \sim c \iff |y| \sim c)\,,
\end{flalign*}
where in the last two tautologies $\sim$ is $=$ or $\geq$,
and $|x| \sim c \, \coloneqq \, (x \geq 0 \land x \sim c)
\lor (x < 0 \land -x \sim c)$. Let $G \coloneqq \{ \Psi_1,
\dots, \Psi_\ell \}$ (as defined in
line~\ref{algo:ip-to-gcd-a1:setA}). The next step of the
algorithm is to remove the system of inequalities from every
formula $\Psi \in G$ via changes of variables
(lines~\ref{algo:ip-to-gcd-a1:lineqPsi}--\ref{algo:ip-to-gcd-a1:change-of-variables}).
This can be done thanks to a fundamental result by von zur
Gathen and Sieveking~\cite{vzGS78} that characterises the
set of solutions of linear inequalities as a union of
discrete shifted cones. The following formulation of this
result is from~\cite[Theorem~3]{LechnerOW15}.

\begin{proposition}[{\cite{vzGS78}}]
  \label{prop:vzGS}
  Consider matrices $A \in \ZZ^{m \times d}$, $C \in \ZZ^{n
      \times d}$, and vectors $\vec b \in \ZZ^m$, $\vec d
      \in \ZZ^n$. Let $r \coloneqq \rank(A)$ and $s
      \coloneqq \rank{\left(\begin{smallmatrix} A\\
      C \end{smallmatrix}\right)}$. Then,

  \vspace{-5pt}
  \[ 
  \{ \vec x \in \ZZ^d : A \cdot \vec x = \vec b \land C \cdot \vec x \leq \vec d \} 
  \,=\, \bigcup_{i \in I} \{ \vec u_i + E_i \cdot \vec y \,:\, \vec y \in \NN^{d-r} \}\,,
  \]
  where $I$ is a finite set, $\vec u_i \in \ZZ^d$, $E_i \in
  \ZZ^{d \times (d-r)}$\! and $\norminf{\vec
  u_i},\norminf{E_i} \leq (d+1) (s \cdot
  \max(2,\norminf{A},\norminf{C},\norminf{\vec
  b},\norminf{\vec d}))^s$.
\end{proposition}

\noindent
Let $S = \{ (\vec u_i, E_i) : i \in I\}$ be the set of pairs
given by~\Cref{prop:vzGS} on the linear inequalities
of~$\Psi$, as written in line~\ref{algo:ip-to-gcd-a1:VZG},
and given $(\vec u, E) \in S$ consider the system~$\Psi''$
defined in line~\ref{algo:ip-to-gcd-a1:change-of-variables}.
Following~\Cref{prop:vzGS}, $\Psi''$ is interpreted over
$\NN$. By definition of $G$, in $\Psi$, every variable~$x$
appearing in a GCD~constraint also appears in a (non-zero)
sign constraint $x \leq -1$ or $x \geq 1$. This means that
in the system $\vec x = \vec u + E \cdot \vec y$, the row
corresponding to $x$ is of the form $x = f(\vec y)$ where
$f$ is a linear polynomial having coefficients and constant
with the same polarity, i.e., they are all negatives (if $x
\leq -1$) or positives (if $x \geq 1$). Therefore, all
GCD~constraints in~$\Psi''$ are of the form $\gcd(f,g) \sim
c$ where $f$ and $g$ are polynomials with coefficients and
constant having the same polarity.
Line~\ref{algo:ip-to-gcd-a1:adjust-polarity} modifies
$\Psi''$ so that every polynomial in it becomes of positive
polarity, thanks to the equalities $\gcd(f,g) = \gcd(-f,g)$
and $\gcd(f,g) = \gcd(g,f)$. 
What is left is to translate~$\Psi''$ into a system of
divisibilities. This is done in
lines~\ref{algo:ip-to-gcd-a1:trans-positive}
and~\ref{algo:ip-to-gcd-a1:trans-negative} by simply relying
on the following two tautologies: 
\begin{equation}\label{eqn:tautologies}
\begin{aligned}
  \gcd(f,g) = c \land f \neq 0 \land g \neq 0
  &\iff 
  \exists w \in \NN :~ c \div f ~\land~ c \div g ~\land~ f \div w ~\land~ g \div w+c\,,
  \\[4pt]
  \gcd(f,g) \geq c 
  &\iff 
  \exists z \in \NN :~ z + c \div f ~\land~ z+c \div g.
\end{aligned}
\end{equation}
Above, note that we can assume $f \neq 0 \land g \neq 0$ in
$\Psi''$, again because of the sign constraints appearing in
$\Psi$. While the second tautology should be
self-explanatory, the first one merits a formal proof:%
\begin{flalign*}
  &\gcd(f,g) = c \land f \neq 0 \land g \neq 0\\
  \iff\,& \exists a , b \in \ZZ :~ c \div f ~\land~ c \div g
  ~\land~ a \cdot f + b \cdot g = c &\text{B\'ezout's
  identity} \\
  \iff\,& \exists w, z \in \ZZ :\, w \leq 0 \,\land\, c \div
  f \,\land\, c \div g \,\land\, f \div w \,\land\, g \div z
  \,\land\, w + z = c &\text{set $w = a \cdot f$ and $z = b
  \cdot g$} \\
  &&\hspace{-3cm}\text{B\'ezout's identity allows picking $w
  \leq 0$} \\
  \iff\,& \exists w \in \ZZ :\, w \leq 0 \,\land\, c \div f
  \,\land\, c \div g \,\land\, f \div -w \,\land\, g \div
  c-w &\hspace{-3cm}\text{eliminate $z$, and $f \div w
  \Leftrightarrow f \div -w$} \\
  \iff\,& \exists w \in \NN :\, c \div f \,\land\, c \div g
  \,\land\, f \div w \,\land\, g \div w+c
  &\hspace{-3cm}\text{change of variable $-w \gets w$.}
\end{flalign*}
Note that the divisibilities in \eqref{eqn:tautologies}
ensure that~$\Psi''$ satisfies the constraints required
by~\gcdtodiv triples. After translating GCDs into
divisibilities, the procedure computes a matrix $E'$ by
enriching~$E$ with zero columns corresponding to the new
variables~$z$ and~$w$, and adds the resulting triple
$(\Psi'',\vec u, E')$ to $B$
(line~\ref{algo:ip-to-gcd-a1:add-to-B}). We obtain the
following result:

\begin{lemma}
  \label{lemma:gcd-to-div}
  \Cref{algorithm:from-IP-GCD-to-systems-of-divisibilities}
  respects its specification. Given as input a system $\Phi$
  with $d$ variables and~$k$ GCD~constraints, every triple
  $(\Psi,\vec u, E)$ in the output set $B$ is such that
  $\Psi$ has at most $d+k$ variables and $4 k$
  divisibilities, $E$ has at most $d$ non-zero columns,
  and~$\norminf{\Psi}, \norminf{\vec u}, \norminf{E} \leq
  (d+1)^{d+2} (\norminf{\Phi}+1)^{d+1}$\!.
\end{lemma}

\begin{proof}
  The fact that
  \Cref{algorithm:from-IP-GCD-to-systems-of-divisibilities}
  respects its specification follows from the discussion
  given above. In particular, $\{\vec a \in \ZZ^d : \vec a \text{ solution of }
  \Phi \} = \sem{B}$ stems from the fact that the
  procedure treats the original formula $\Phi$ by only
  relying on tautologies and on~\Cref{prop:vzGS}.

  Let us study the bounds on $(\Psi,\vec u, E)$. For the
  bound on the number of variables in $\Psi$ and non-zero
  columns in $E$, note that by~\Cref{prop:vzGS}, the change
  of variables of
  line~\ref{algo:ip-to-gcd-a1:change-of-variables} does not
  increase the number of variables, and therefore the only
  lines where the number of variables increases are
  lines~\ref{algo:ip-to-gcd-a1:trans-positive}
  and~\ref{algo:ip-to-gcd-a1:trans-negative}. Overall, these
  two lines introduce $k$ many variables, one for each
  GCD~constraint; so the number of variables in~$\Psi$ is
  bounded by $d+k$. Each new variable introduces a zero
  column in $E$, which has thus at most $d$ non-zero columns
  (line~\ref{algo:ip-to-gcd-a1:add-to-B}). For the bound on
  the number of divisibilities, only
  lines~\ref{algo:ip-to-gcd-a1:trans-positive}
  and~\ref{algo:ip-to-gcd-a1:trans-negative} matter, and
  they introduce at most $4$ divisibilities per
  GCD~constraint; hence $\Psi$ has at most $4k$
  divisibilities. Lastly, let us derive the bound on the
  infinity norm of $\Psi$, $\vec u$ and $E$. The rewritings
  done in line~\ref{algo:ip-to-gcd-a1:setA} increase the
  infinity norm by at most $1$; this occurs when relying on
  the tautology $\gcd(y,z) \neq c+2 \iff \gcd(y,z) \leq c+1
  \lor \gcd(y,z) \geq c+3$. The bound on $\vec u$ and $E$
  then follows from a simple application
  of~\Cref{prop:vzGS}: $\norminf{\vec u},\norminf{E} \leq
  (d+1) \cdot (d \cdot (\norminf{\Phi}+1))^d$. The change of
  variables in
  line~\ref{algo:ip-to-gcd-a1:change-of-variables} yields
  $\norminf{\Psi''} \leq (d+1) \cdot \max( \norminf{\vec u},
  \norminf{E} ) \cdot (\norminf{\Phi} + 1)$.
  Lines~\ref{algo:ip-to-gcd-a1:adjust-polarity}--\ref{algo:ip-to-gcd-a1:add-to-B}
  do not change the infinity norm, and therefore we obtain
  the bound in the statement of the lemma.
\end{proof}

\paragraph{Step II: force increasingness.}
We now move
to~\Cref{algorithm:from-IP-GCD-to-increasing-systems}, whose
role is to translate the systems of divisibility constraints  computed
by~\Cref{algorithm:from-IP-GCD-to-systems-of-divisibilities}
into~$3$-increasing systems. As such, the procedure takes as
input a set~$B$ of~\gcdtodiv triples. We first need the
following result:

\input{gcd-to-increasing-algo}

\begin{lemma}
  \label{lemma:gcd-non-increasing}
  Let $(\Psi,\vec u, E)$ be a \gcdtodiv triple. If the
  system~$\Psi$ is not in increasing form, then there is a
  non-constant polynomial~$f$ primitive part of a left-hand
  side in $\Psi$ such that $\module_f(\Psi) \cap \ZZ \neq
  \{0\}$. If~$\Psi$ is in increasing form, then it is
  increasing for $\vec z \incord \vec y \incord \vec w$,
  where $\vec z$, $\vec y$ and $\vec w$ are the families of
  variables appearing in the definition of \gcdtodiv triple.
\end{lemma}

\begin{proof}
  For the first statement, we prove a stronger result: if
  $\Psi$ is not increasing for $\vec z \incord \vec y
  \incord \vec w$, then there is a non-constant
  polynomial~$f$ primitive part of a left-hand side in
  $\Psi$ s.t.~${\module_f(\Psi) \cap \ZZ \neq \{0\}}$.
  Observe that then, by definition of divisibility module
  and increasing form, $\Psi$ cannot be in increasing form
  for any order; which shows the second statement in the
  lemma by contrapositive. 

  Consider an order $x_1 \incord \dots \incord x_d$ of the
  variables in $\Psi$ that belongs to $\vec z \incord \vec y
  \incord \vec w$, and suppose that~$\Psi$ is not in
  increasing form for this order. Therefore, there is a
  primitive part $f$ of a left-hand side of a divisibility
  in $\Psi$ such that $\module_f(\Psi) \cap
  \ZZ[x_1,\dots,x_j] \neq \ZZ f$, where $x_j = \lv(f)$. Let
  $g \in \module_f(\Psi) \cap \ZZ[x_1,\dots,x_j] \setminus
  \ZZ f$. We show that $g$ must be a constant polynomial. We
  distinguish two cases, depending on whether the leading
  variable of $f$ belongs to $\vec z$ or to $\vec y$ (note
  that it cannot belong to $\vec w$, as no left-hand sides
  with variables from this family exists).

  \begin{description}
    \item[case {\rm{$\lv(f)$}} is in $\vec z$.] Since
    $\lv(g) \incordeq \lv(f)$, all variables in $g$ are from
    $\vec z$. By~Property~\ref{lemma:gcd-to-div:item1}
    of~\gcdtodiv triple, each divisibility in $\Psi$ is of
    the form $h(\vec z) \div h'(\vec y)$ or of the form
    $h(\vec y) \div h'(\vec w)$.
    By~\Cref{lemma:module-span}, a set spanning
    $\module_f(\Psi)$ is given by $\{f,c_1 \cdot g_1, \dots,
    c_m \cdot g_m\}$ where $c_i \in \NN$ and $g_i$ is a
    right-hand side of a divisibility in $\Psi$, for every
    $i \in [1,m]$. This means that every $g_i$ has variables
    from~$\vec y$ or~$\vec w$. Since~$g$ does not have any
    variable from $\vec y$ or $\vec w$ and belongs to~$\ZZ
    f$, we conclude that it must be a constant polynomial.
    
    \item[case {\rm{$\lv(f)$}} is in $\vec y$.] Again
    from~Property~\ref{lemma:gcd-to-div:item1} of~\gcdtodiv
    triple, $f$ only appears as left-hand side in
    divisibilities of the form $a \cdot f(\vec y) \div
    h(\vec w)$, with $a \in \ZZ \setminus \{0\}$. Since no
    non-constant polynomial $h(\vec w)$ appears in a
    left-hand side of $\Psi$, the set $\{f,c_1 \cdot g_1,
    \dots, c_m \cdot g_m\}$ spanning $\module_f(\Psi)$
    computed via~\Cref{lemma:module-span} is such that $c_i
    \neq 0$ if and only if $g_i$ only has variables from
    $\vec w$, for every $i \in [1,m]$. Since~$\incord$
    belongs to $\vec z \incord \vec y \incord \vec w$, from
    $\lv(g) \incord \lv(f)$ we then conclude that $g$ must
    be a constant polynomial.%
    \qedhere
  \end{description}
\end{proof}

Consider $(\Psi,\vec u, E) \in B$.
\Cref{algorithm:from-IP-GCD-to-increasing-systems} relies on
\Cref{lemma:gcd-non-increasing} to test whether $\Psi$ is
increasing
(line~\ref{gcd-to-increasing:line-test-increasing}). If it
is not, it computes the minimum positive integer $c \in
\module_f(\Psi)$, for some $f$ non-constant
(line~\ref{gcd-to-increasing:line-compute-c}). By definition
of divisibility module, for every primitive polynomial $f$
and polynomial $g \in \module_f(\Psi)$, we have that $\Psi$
entails $f \div g$, that is for every $\vec a \in \ZZ^m$
solution to $\Psi$, $f(\vec a)$ divides $g(\vec a)$; and
therefore $\Psi$ entails $f \div c$. We can now eliminate
all variables that occur in $f$: by definition of~\gcdtodiv
triple, $f$ has coefficients and constant that are all
positive, and $\Psi$ is interpreted over $\NN$. We conclude
that every solution of $\Psi$ is such that it assigns an
integer in $[0,c]$ to every variable in $f$. The
\textbf{for} loop in
line~\ref{gcd-to-increasing:line-for-loop} iterates over the
subset of these (partial) assignments satisfying $f \div c$.
Each of these assignments~$\vec \nu$ yields a new triple
$(\Psi_{\vec \nu}, \vec u_{\vec \nu}, E_{\vec \nu})$,
defined as in
lines~\ref{gcd-to-increasing:line-new-Psi}--\ref{gcd-to-increasing:line-new-E},
which is a \gcdtodiv triple thanks to the lemma below (that
follows directly from the definition of \gcdtodiv triple).

\begin{lemma}
  \label{lemma:gcd-to-div-substitution}
  Let $(\Psi,\vec u, E)$ be a \gcdtodiv triple, with $\vec u
  \in \ZZ^d$, and $X$ be a subset of the variables appearing
  in left-hand sides of $\Psi$. Consider a map $\vec \nu
  \colon X \to \ZZ$. Let $\Psi' \coloneqq
  \Psi\substitute{\vec \nu(x)}{x : x \in X}$, $\vec u' \in
  \ZZ^d$, and $E'$ be the matrix obtained from $E$ by
  removing the columns corresponding to variables in $X$.
  The triple $(\Psi', \vec u', E')$ is a \gcdtodiv triple. 
\end{lemma}


The key equivalence, from which the correctness of the
algorithm directly stems, is:
\begin{equation}
  \label{eq:to-increasing:key-equivalence}
  \{ 
    \vec u + E \cdot \vec \lambda : \vec \lambda \in \NN^m \text{ solution for } \Psi
  \} 
  = 
  \!\bigcup_{\substack{\vec \nu~\text{substitution}\\ \text{considered in line~\ref{gcd-to-increasing:line-for-loop}}}}\! 
  \{
    \vec u_{\vec \nu} + E_{\vec \nu} \cdot \vec \lambda : \vec \lambda \in \NN^{m-j} \text{ solution for } \Psi_{\vec \nu}
  \},
\end{equation}
where $j \geq 1$ is the number of variables in $f$. The
procedure adds each triple $(\Psi_{\vec \nu}, \vec u_{\vec
\nu}, E_{\vec \nu})$ to the set~$B$
(line~\ref{gcd-to-increasing:line-new-triple}), so that it
will be tested for increasingness in a later iteration of
the \textbf{while} loop of
line~\ref{gcd-to-increasing:line-pop-B}. Termination is
guaranteed from the fact that~$f$ is non-constant and so
each $\Psi_{\vec \nu}$ has strictly fewer variables
than~$\Psi$.

\begin{lemma}
  \label{lemma:from-IP-GCD-to-increasing-systems}
  \Cref{algorithm:from-IP-GCD-to-increasing-systems}
  respects its specification. On input~$B$ such that, for
  every $(\Psi, \vec u, E) \in B$, $\Psi$ has at most $d$
  variables and $k$ GCD~constraints, and~$E$ has at most
  $\ell$ non-zero columns, each output triple $(\Psi',\vec
  u', E') \in C$ is such that $\Psi'$ has at most~$d$
  variables and~$k$ GCD~constraints, $E'$ has at most $\ell$
  non-zero columns, $\norminf{\Psi'} \leq  2^{15} (d+1)
  \cdot (\norminf{B}+1)^7$, $\norminf{\vec u'} \leq (\ell+1)
  \cdot \norminf{B}^2$, and $\norminf{E'} \leq \norminf{B}$.
\end{lemma}

\noindent
Above, $\norminf{B}$ is the maximum among $\norminf{\Psi}$,
$\norminf{\vec u}$, and $\norminf{E}$, over all \gcdtodiv
triples $(\Psi,\vec u, E) \in B$. The most difficult parts
of the proof are the bounds on $\Psi'$ and $\vec u'$. These,
however, follow from the properties of \gcdtodiv triples
and, in particular, from the special shape of the
divisibility constraints that they allow.
Together,~\Cref{lemma:gcd-to-div,lemma:from-IP-GCD-to-increasing-systems}
imply~\Cref{prop:to-three-increasing} in~\Cref{sec:intro-gcd-ip}.

\begin{proof}
  The fact
  that~\Cref{algorithm:from-IP-GCD-to-increasing-systems}
  respects its specification follows from the discussion
  given above, and in particular
  from~\Cref{lemma:gcd-non-increasing}
  and~\Cref{eq:to-increasing:key-equivalence}. Let us then
  focus on the bounds on an output triple $(\Psi', \vec u',
  E')$. Note that $\norminf{B} \geq 1$, if $B$ contains at
  least one divisibility. Following the \textbf{while} loop
  of~\Cref{gcd-to-increasing:line-pop-B}, there is a
  sequence of triples 
  \[ 
    (\Psi_1, \vec u_1, E_1) \ \to \ (\Psi_2, \vec u_2, E_2) \ \to \ \dots \ \to (\Psi_k, \vec u_k, E_k) = (\Psi', \vec u', E'),
  \]
  where $(\Psi_1, \vec u_1, E_1) \in B$ and for every $i \in
  [1,k-1]$, the triple $(\Psi_{i+1}, \vec u_{i+1}, E_{i+1})$
  is computed from $(\Psi_{i}, \vec u_{i}, E_i)$ following
  lines~\ref{gcd-to-increasing:line-find-f}--\ref{gcd-to-increasing:line-new-triple}.
  In particular, given $i \in [1,k-1]$:
  
  \begin{itemize}
    \item  there is a non-constant
    polynomial~$\widehat{f}_i$ that is the part of a
    left-hand side~in~$\Psi_i$
    satisfying~${\module_{\widehat{f}_i}(\Psi_i) \cap \ZZ
    \neq \!\{0\}}$ and with variables $\widehat{\vec
    \lambda}_i \coloneqq
    (\lambda_{i,1},\dots,\lambda_{i,j_i})$, and 
    \item there is a map $\vec \nu_i :
    \{\lambda_{i,1},\dots,\lambda_{i,j_i}\} \to
    [0,\widehat{c}_i]$ such that  $\widehat{f}_i(\vec
    \nu_i(\widehat{\vec \lambda}_i))$ divides
    $\widehat{c}_i$, where $\widehat{c}_i$ is the minimum
    positive integer in $\module_{\widehat{f}_i}(\Psi_i)$,
  \end{itemize}
  such that $\Psi_{i+1} = \Psi_i\substitute{\vec
  \nu_i(\lambda_{i,r})}{\lambda_{i,r} : r \in [1,j_i]}$,
  $\vec u_{i+1} = \vec u_{i} + \sum_{r=1}^j \vec
  \nu_i(\lambda_{i,r}) \cdot \vec p_r$, where $\vec p_r$ is
  the column of $E_i$ corresponding to the
  variable~$\lambda_{i,r}$, and $E_{i+1}$ is obtained from
  $E_i$ by removing the columns corresponding to variables
  in~$\widehat{\vec \lambda}_i$. Note that this implies that
  $\norminf{E'} \leq \norminf{E_i} \leq \norminf{B}$ and
  that $E'$ and $E_i$ have at most $\ell$ non-zero columns,
  as required by the lemma.
  
  We show the remaining bounds in the statement of the lemma
  by induction on $i \in [1,k]$, with the induction
  hypothesis stating that $(\Psi_i, \vec u_i, E_i)$ is
  a~\gcdtodiv triple where:
  \begin{enumerate}[(A)]
    \item\label{ih:to-increasing-bounds:0} $\Psi_i$ is a
    system with at most $d$ variables and $k$
    GCD~constraints, having the form 
    \[
      \Psi_i = \bigwedge_{j=1}^l c_j \div f_j(\vec y) \land 
      \bigwedge_{j=l+1}^n 
      \Big(z_j + c_j \div f_j(\vec y) \land z_j + c_j \div g_j(\vec y) \Big) 
      \land
      \bigwedge_{j=n+1}^m 
      \Big(f_j(\vec y) \div w_j \land g_j(\vec y) \div w_j+c_j\Big),
    \]
    where $\vec y$, $\vec z = (z_{l+1},\dots,z_n)$ and $\vec
    w = (w_{n+1},\dots,w_m)$ are disjoint families of
    variables (according to the definition of~\gcdtodiv
    triple), $c_j \in \pZZ$ for every $j \in [1,m]$, and 
    \item\label{ih:to-increasing-bounds:1} for every $j \in
    [1,l]$, $c_j \leq 2^{15} \cdot (2 + \norminf{B})^7$, and
    for every $j \in [l+1,m]$, $c_j \leq \norminf{B}$, and
    \item\label{ih:to-increasing-bounds:2} for every $j \in
    [l+1,m]$,  $h(\vec y) \in \{f_j(\vec y),\, g_j(\vec
    y)\}$ has variable coefficients bounded
    by~$\norminf{B}$, and constant bounded by $(d+1-d')
    \cdot \norminf{B}^2$, where $d'$ is the number of
    variables in $h$, and
    \item\label{ih:to-increasing-bounds:3} if $i \in [2,k]$,
    then for every $r \in [1,j_{i-1}]$, if $\lambda_{i-1,r}$
    belongs to $\vec y$ then $\vec \nu_i(\lambda_{i-1,r})
    \leq \norminf{B}$, and if~$\lambda_{i-1,r}$ belongs to
    $\vec z$ then $\vec \nu_i(\lambda_{i-1,r}) \leq 2^{14}
    (2 + \norminf{B})^7$.
  \end{enumerate}
  Note that \Cref{ih:to-increasing-bounds:3} implies
  $\norminf{\vec u'} \leq (\ell+1) \cdot \norminf{B}^2$,
  since all non-zero columns of $E_1$ correspond to
  variables in $\vec y$, by definition of~\gcdtodiv triple.
  \Cref{ih:to-increasing-bounds:1,ih:to-increasing-bounds:2}
  imply $\norminf{\Psi'} \leq 2^{15} (d+1) \cdot
  (\norminf{B}+1)^7$.

  \begin{description}
    \item[base case $i = 1$.] In this case $(\Psi_1, \vec
    u_1, E_1) \in B$ and the hypothesis above trivially
    holds since $(\Psi_1, \vec u_1, E_1)$ is a~\gcdtodiv
    triple and
    Properties~\ref{lemma:gcd-to-div:item1}--\ref{lemma:gcd-to-div:item2.5}
    ensure that $\Psi_1$ has the form
    in~\Cref{ih:to-increasing-bounds:0}.
    \item[induction step $i+1 \geq 2$.] 
      We assume the induction hypothesis for $(\Psi_i,\vec
      u_i, E_i)$, and establish it for $(\Psi_{i+1}, \vec
      u_{i+1}, E_{i+1})$.
      By~\Cref{lemma:gcd-to-div-substitution}, $(\Psi_{i+1},
      \vec u_{i+1}, E_{i+1})$ is a~\gcdtodiv triple,
      hence~\Cref{ih:to-increasing-bounds:0} follows. So,
      let us focus on establishing the part of the induction
      hypothesis related to the infinity norm of
      $\Psi_{i+1}$ and $\vec \nu_{i}$
      (\Cref{ih:to-increasing-bounds:1,ih:to-increasing-bounds:2,ih:to-increasing-bounds:3}).
      Let $\vec z$, $\vec y$ and $\vec w$ be the families of
      variables witnessing that $(\Psi_{i}, \vec u_{i},
      E_{i})$ is a~\gcdtodiv triple, according to the
      definition of such triples.
      By~Property~\ref{lemma:gcd-to-div:item1},
      $\widehat{f}_i$ has variables from either $\vec z$ or
      $\vec y$ (not both). We divide the proof depending on
      these two cases.
      \begin{description}
        \item[case $\widehat{f}_i$ has only variables from $\vec y$.] 
          From Property~\ref{lemma:gcd-to-div:item1} of
          \gcdtodiv triples, $\widehat{f}_i$ only appears as
          a left-hand side in divisibilities of the form $a
          \cdot \widehat{f}_i(\vec y) \div h(\vec w)$, with
          $a \in \ZZ \setminus \{0\}$. From
          Property~\ref{lemma:gcd-to-div:item2.5} of
          \gcdtodiv triples together with the fact
          that~$\module_{\widehat{f}_i}(\Psi_{i}) \cap \ZZ
          \neq \{0\}$, we conclude that there must be a
          variable $w$ in $\vec w$ and $c \in \pZZ$ such
          that $a_1 \cdot \widehat{f}_i \div w$ and $a_2
          \cdot \widehat{f}_i \div w + c$ are divisibilities
          in~$\Psi_i$, for some $a_1,a_2 \in \ZZ \setminus
          \{0\}$. Then, $c \in
          \module_{\widehat{f}_i}(\Psi_i)$ and by definition
          $\widehat{c}_i \leq c$. By induction hypothesis
          (\Cref{ih:to-increasing-bounds:1}), $\widehat{c}_i
          \leq \norminf{B}$, which shows
          \Cref{ih:to-increasing-bounds:3} directly by
          definition of~$\vec \nu_i$.
          \Cref{ih:to-increasing-bounds:1} is also trivially
          satisfied: since we are replacing only variables
          in~$\vec y$, all polynomials in $\Psi_{i+1}$ with
          variables from $\vec z$ or $\vec w$ are
          polynomials in~$\Psi_{i}$, and no new coefficient
          $c'$ can appear in divisibilities of the form~$c'
          \div f(\vec y)$. 

          To prove~\Cref{ih:to-increasing-bounds:2}, let
          $h'$ be a polynomial obtained from some $h(\vec
          y)$ in $\Psi_i$ by evaluating each $\lambda_{i,r}$
          as $\vec \nu_i(\lambda_{i,r})$ ($r \in [1,j]$). By
          induction hypothesis
          (\Cref{ih:to-increasing-bounds:2}), $h$ has
          variable coefficients bounded by $\norminf{B}$,
          and constants bounded by $(d+1-d') \cdot
          \norminf{B}^2$,  
          where $d'$ is the number of variables in $h$.
          Let~$d''$ be the number of variables in $h'$.
          Because of the substitutions done by $\vec \nu_i$,
          we conclude that the coefficients of $h'$ are
          bounded by $\norminf{B}$, whereas its constant is
          bounded by $(d+1-d') \cdot \norminf{B}^2 +
          (d'-d'') \cdot \norminf{B}^2 = (d+1-d'') \cdot
          \norminf{B}^2$.
        \item[case $\widehat{f}_i$ has only variables from $\vec z$.]
          In this case, $\widehat{f}_i$ is of the form $z +
          c$ for some $c \in \pZZ$, and by
          Property~\ref{lemma:gcd-to-div:item2} of~\gcdtodiv
          triple it appears in exactly two divisibilities~$z
          + c \div f(\vec y)$ and $z + c \div g(\vec y)$. In
          order to upper bound $\widehat{c}_i$, we divide
          the proof in two cases, depending on whether $(\ZZ
          f + \ZZ g) \cap \ZZ = \{0\}$. 
          \begin{description}
            \item[case $(\ZZ f + \ZZ g) \cap \ZZ = \{0\}$.] 
              Since $\module_{\widehat{f}_i}(\Psi_i) \cap
              \ZZ \neq \{0\}$, by
              Properties~\ref{lemma:gcd-to-div:item1}
              and~\ref{lemma:gcd-to-div:item2.5} of
              \gcdtodiv triples there must be two
              polynomials $f'(\vec y)$ and $g'(\vec y)$, a
              variable $w$ in $\vec w$ and $a',b',c' \in
              \pZZ$ such that $f'(\vec y) \div w$, $g'(\vec
              y) \div w + c'$ and $\{a' \cdot f',\, b' \cdot
              g'\} \subseteq (\ZZ f + \ZZ g)$. Then, by
              definition of divisibility module, $a' \cdot
              b' \cdot c' \in
              \module_{\widehat{f}_i}(\Psi_i)$. By induction
              hypothesis $c' \leq \norminf{B}$
              (\Cref{ih:to-increasing-bounds:1}), and
              therefore to find a bound on $\widehat{c}_i$
              is suffices to bound $a'$ and $b'$. Let us
              study the case of $a'$ (the bound is the same
              for $b'$). The set $S \coloneqq \{-f',f,g\}$
              can be represented as a matrix $A \in
              \ZZ^{(d+1) \times 3}$ in which each column
              contains the coefficients and the constant of
              a distinct element of $S$. We
              apply~\Cref{prop:vzGS} on the system $A \cdot
              (x_1,x_2,x_3) = \vec 0$, and conclude that
              $a'$ is bounded by $4 \cdot (3 \cdot \max(2,
              \norminf{A}))^3 \leq 108 \cdot
              (\norminf{B}+1)^3$. Therefore, $\widehat{c}_i
              \leq 2^{14} (\norminf{B}+1)^7$.
            \item[case $(\ZZ f + \ZZ g) \cap \ZZ \neq \{0\}$.] 
              In this case, we consider the set $S \coloneqq
              \{f,g\}$ and the matrix $A \in \ZZ^{(d+1)
              \times 2}$ in which each column contains the
              coefficients and the constant of a distinct
              element in $S$, with the constant being places
              in the last row. To find a non-zero value $c'
              \in (\ZZ f + \ZZ g) \cap \ZZ$, we solve the
              system $A \cdot (x_1,x_2) + x_3 \cdot (\vec 0,
              1) = \vec 0$. By~\Cref{prop:vzGS},
              $\widehat{c}_i \leq |c'| \leq 4 \cdot (3 \cdot
              \max(2, \norminf{A}))^3 \leq 108 \cdot
              (\norminf{B}+1)^3$.
          \end{description}
          Therefore, $\vec \nu_i(z) \leq \widehat{c}_i \leq
          2^{14} (\norminf{B}+1)^7$, which
          shows~\Cref{ih:to-increasing-bounds:3} of the
          induction hypothesis.
          \Cref{ih:to-increasing-bounds:2} is trivially
          satisfied, since $\vec \nu_i$ replaces only the
          variable $z$, which does not belong to $\vec y$.
          \Cref{ih:to-increasing-bounds:1} follows from the
          fact that in the polynomial $z + c$ the integer
          $c$ is bounded by $\norminf{B}$ by induction
          hypothesis, and therefore $\vec \nu(z) + c \leq
          2^{15} (\norminf{B}+1)^7$.
          \qedhere
      \end{description}
  \end{description}
\end{proof}

\subsection{Bound on the solutions modulo primes}
\label{sec:ip-gcd:small-bound-p-adic-solutions}

Through~\Cref{algorithm:from-IP-GCD-to-systems-of-divisibilities,algorithm:from-IP-GCD-to-increasing-systems}
we are able to compute from a IP-GCD system $\Phi$ a set of
\gcdtodiv triples $C$ such that $\{\vec a \in \ZZ^d : \vec a
\text{ is a solution to } \Phi \} = \sem{C}$. To
apply~\Cref{theorem:local-to-global}, what is left is to
study bounds on the solutions modulo primes in
$\pdiff(\Psi)$, for every $(\Psi, \vec u, E) \in C$.

\LemmaIPGCDSmallSolutionsModulo*

\begin{proof}
  Let us assume there exists a solution $\vec \nu \colon
  \vec \lambda \to \ZZ$ to $\Psi(\vec \lambda)$ modulo $p$.
  We build another solution $\vec \nu' \colon \vec \lambda
  \to \ZZ$ to $\Psi(\vec \lambda)$ modulo $p$ such that
  $\norminf{\vec \nu'(\vec \lambda)} \leq (d+1) \cdot
  \norminf{\Psi}^3 p^2$. According to
  Properties~\ref{lemma:gcd-to-div:item1}--\ref{lemma:gcd-to-div:item2.5}
  of \gcdtodiv triples, the formula $\Psi$ is of the form:
  \[
    \Psi = \bigwedge_{i=1}^l c_i \div f_i(\vec y) \land 
    \bigwedge_{i=l+1}^n \Big(z_i + c_i \div f_i(\vec y) \land z_i + c_i \div g_i(\vec y) \Big) \land
    \bigwedge_{i=n+1}^m \Big(f_i(\vec y) \div w_i \land g_i(\vec y) \div w_i+c_i\Big),
  \]
  where $\vec y$, $\vec z = (z_{l+1},\dots,z_n)$ and $\vec w
  = (w_{n+1},\dots,w_m)$ are disjoint families of variables,
  and $c_i \in \pZZ$ for every $i \in [1,m]$. Recall that
  the variables $z_i$ ($i \in [l+1,n]$) are all distinct,
  and the same holds true for the variables $w_i$ $(i \in
  [n+1,m]$). We define $\mu_i \coloneqq v_p(c_i)$, $\mu
  \coloneqq \max_{i=1}^m \mu_i$, and $\vec \nu'$ as:
  \begin{align*}
    \vec \nu'(x) \coloneqq \!
    \begin{cases}
      (\vec \nu (x) \text{ modulo } p^\mu) &\text{if } x \text{ belongs to } \vec y,\\
      1 &\text{if } x = z_i \text{ for some } i \in [l+1,n] \text{ and } p \text{ divides } c_i,\\
      0 &\text{if } x = z_i \text{ for some } i \in [l+1,n] \text{ and } p \text{ does not divide } c_i,\\
      p^{\mu+1} g_i(\vec \nu'(\vec y)) \,{-}\, c_i
      &\text{if } x = w_i \text{ for some } i \,{\in}\, [n\,{+}\,1,m] \text{ and } p^{\mu_i+1} \,\text{does not divide}\, f_i(\vec \nu(\vec y)),\\
      p^{\mu+1} f_i(\vec \nu'(\vec y))
      &\text{otherwise } (x = w_i \text{ for some } i \in [n+1,m]).\\
    \end{cases}
  \end{align*}
  Note that $\vec \nu'$ is defined recursively in the last
  two cases; this recursion is on variables from $\vec y$
  and thus $\vec \nu'$ is well-defined. By definition,
  $p^{\mu+1} \leq \norminf{\Psi} \cdot p$, and therefore
  $\norminf{\vec \nu'(x)} \leq (d+1) \cdot \norminf{\Psi}^3
  p^2$ for every variable~$x$ in $\vec \lambda$. To conclude
  the proof, let us show that $\vec \nu'$ is a solution for
  $\Psi$ modulo $p$. The fact that $f(\vec \nu'(\vec
  \lambda)) \neq 0$ for every polynomial $f$ in the
  left-hand side of a divisibility in $\Psi$ stems from
  $\vec \nu'$ being defined to be non-negative for every
  variable in $\vec z$ and $\vec y$, and $f$ having a
  positive constant by Property~\ref{lemma:gcd-to-div:item1}
  of \gcdtodiv triples. So, we focus on showing that
  $v_p(f(\vec \nu'(\vec \lambda))) \leq v_p(g(\vec \nu'(\vec
  \lambda)))$ for every divisibility $f \div g$ in~$\Psi$.

  Let $i \in [1,l]$, and consider~$c_i \div f_i(\vec y)$. By
  definition of $\vec \nu'$, ${f_i(\vec \nu'(\vec y)) \equiv
  f_i(\vec \nu(\vec y)) \pmod {p^{\mu+1}}}$, and therefore
  $v_p(f_i( \vec \nu'(\vec y))) = \min(\mu+1, v_p(f_i(\vec
  \nu(\vec y))))$. By definition of~$\mu$, we have $c_i
  \not\equiv 0 \pmod {p^{\mu+1}}$, i.e., $v_p(c_i) < \mu+1$.
  We conclude that $v_p(c_i) \leq v_p(f_i( \vec \nu'(\vec
  y)))$.

  Let $i \in [l+1,n]$, and consider the divisibilities $z_i
  + c_i \div f_i(\vec y)$ and $z_i + c_i \div g_i(\vec y)$.
  By definition of $\vec \nu'$ we have $v_p(\vec
  \nu'(z_i)+c_i) = 0$, and so~$v_p(\vec \nu'(z_i)+c_i) \leq
  v_p(f_i(\vec \nu'(\vec y)))$ and $v_p(\vec \nu'(z_i)+c_i)
  \leq v_p(f_i(\vec \nu'(\vec y)))$.

  Let $i \in [n+1, m]$. Assume first that $p^{\mu_i+1}$ does
  not divide $f_i(\vec \nu(\vec y))$, and so $\vec \nu'$ is
  defined so that~$\vec \nu'(w_i) = p^{\mu+1} g_i(\vec
  \nu'(\vec y)) - c_i$. The divisibility $g_i(\vec y) \div
  w_i+c$ is trivially satisfied by~$\vec \nu'$ over the
  integers, and thus also modulo $p$. By definition of $\vec
  \nu'$ we have ${f_i(\vec \nu'(\vec y)) \equiv f_i(\vec
  \nu(\vec y)) \pmod {p^{\mu+1}}}$ and therefore
  $p^{\mu_i+1}$ does not divide $f_i(\vec \nu'(\vec y))$. By
  definition of~$\mu_i$, this implies ${v_p(f_i(\vec
  \nu'(\vec y))) \leq v_p(c_i)}$. From the definition of
  $\mu$, $v_p(p^{\mu+1} g_i(\vec \nu'(\vec y))) > v_p(c_i)$
  and therefore $v_p(\vec \nu'(w_i)) = v_p(c_i)$, which
  yield $v_p(f_i(\vec \nu'(\vec y))) \leq v_p(\vec
  \nu'(w_i))$. Let us now assume that $p^{\mu_i+1}$ divides
  $f_i(\vec \nu(\vec y))$, and so $\vec \nu'$ is defined so
  that~$\vec \nu'(w_i) = p^{\mu+1} f_i(\vec \nu'(\vec y))$.
  Clearly, the divisibility $f_i(\vec y) \div w_i$ is
  satisfied by $\vec \nu'$ over the integers, and thus also
  modulo $p$. Since $\vec \nu$ is a solution for $\Psi$
  modulo $p$, and $p^{\mu+1}$ divides $f_i(\vec \nu(\vec
  y))$, we conclude that $p^{\mu+1}$ divides $\vec
  \nu(w_i)$. Then, by definition of $\mu$, $v_p(\vec
  \nu(w_i)) > v_p(c_i)$ and therefore $v_p(g_i(\vec \nu(\vec
  y))) \leq v_p(\vec \nu(w_i)+c_i) = v_p(c_i)$. By
  definition of $\vec \nu'$, ${g_i(\vec \nu'(\vec y)) \equiv
  g_i(\vec \nu(\vec y)) \pmod {p^{\mu+1}}}$ and ${v_p(\vec
  \nu'(w_i)+c_i) = v_p(c_i)}$. We conclude that
  ${v_p(g_i(\vec \nu'(\vec y))) \leq v_p(\vec
  \nu'(w_i)+c_i)}$.
\end{proof}

\subsection{Proof
of~\Cref{thm:small-model}}\label{sec:proof-of-small-model}

Thanks
to~\Cref{lemma:bound-on-pzero,lemma:gcd-to-div,lemma:from-IP-GCD-to-increasing-systems,lemma:ip-gcd:small-solutions-modulo},
we obtain the part of~\Cref{thm:small-model} not concerning
optimization as a corollary
of~\Cref{theorem:local-to-global}.

\begin{restatable}{corollary}{CorollaryIPGCDSmallModel}
  \label{corollary:IP-GCD-small-model}
  Each feasible IP-GCD system has a solution of polynomial
  bit length.
\end{restatable}

Let us now discuss the related integer programming
optimization problem. Consider an \mbox{IP-GCD} system
$\Phi(\vec x)$ and the problem of minimizing (or maximizing)
a linear objective $\vec c^\intercal \vec x$ subject
to~$\Phi(\vec x)$. We
apply~\Cref{lemma:gcd-to-div,lemma:from-IP-GCD-to-increasing-systems}
on $\Phi(\vec x)$ to obtain a set $C$ of \gcdtodiv triples
only featuring $3$-increasing systems of divisibility constraints , and
with $\{ \vec a \in \ZZ^d : \vec a \text{ solution to }\Phi\} = \sem{C}$. We
show the following characterization that implies the
optimization part of~\Cref{thm:small-model}:
\begin{enumerate}[I.]
  \item\label{ip-gcd:optim:case1} if for every $(\Psi,\vec
  u, E) \in C$, $\Psi$ is unsatisfiable over $\NN$, then
  $\Phi$ is unsatisfiable;
  \item\label{ip-gcd:optim:case2} else, if there is
  $(\Psi,\vec u, E) \in C$ such that $\Psi$ is satisfiable
  over $\NN$ and the linear polynomial~${\vec c^\intercal
  (\vec u + E \cdot \vec \lambda)}$ has a variable in~$\vec
  \lambda$ with strictly negative (resp.~positive)
  coefficient, then an optimal solution minimizing
  (resp.~maximizing)~$\vec c^\intercal \vec x$ subject to
  $\Phi(\vec x)$ does not exist;
  \item\label{ip-gcd:optim:case3} else, an optimal solution
    does exist, and in particular one with polynomial bit
    length with respect to~$\bitlength{\Phi}$ and
    $\bitlength{\vec c}$.
\end{enumerate}
\Cref{ip-gcd:optim:case1} follows directly from the
equivalence $\{ \vec a \in \ZZ^d : \vec a \text{ solution ot }\Phi\} =
\sem{C}$. Let us focus on~\Cref{ip-gcd:optim:case2}, which
we show for the case of minimization (the case of
maximization being analogous). Consider a triple $(\Psi,\vec
u, E) \in C$ such that $\Psi$ is satisfiable and the linear
polynomial~$f(\vec \lambda) \coloneqq \vec c^\intercal (\vec
u + E \cdot \vec \lambda)$ has a variable in~$\vec \lambda$
with strictly negative coefficient. Let~$\vec z$, $\vec y$
and $\vec w$ be the disjoint families of variable witnessing
the fact that $(\Psi, \vec u, E)$ is a \gcdtodiv triple,
according to the definition of such triples.
By~\Cref{lemma:gcd-non-increasing}, $\Psi$ is increasing for
$\vec z \incord \vec y \incord \vec w$, and from
Property~\ref{lemma:gcd-to-div:item3} of~\gcdtodiv triples,
all variables appearing in $f(\vec \lambda)$ with a non-zero
coefficient are from $\vec y$. Let $\widehat{y}$ be a
variable appearing in $f$ with a negative coefficient, and
consider an order $(\incord) \in (\vec z \incord \vec y
\incord \vec w)$ for which $\widehat{y}$ is the largest of
the variables appearing in $\vec y$. Since $\Psi$ is
satisfiable over $\NN$, it is satisfiable modulo every prime
in $\PP(\Psi)$, and we can
apply~\Cref{func:new-solveincreasing} to compute a
solution~$\vec \nu$ over $\NN$ satisfying the property
highlighted in~\Cref{remark:loc-to-glo:inf-solutions}: the
formula $\Psi\substitute{\vec \nu(x)}{ x : x \incord
\widehat{y}\,}$ has a solution for infinitely many positive
values of~$\widehat{y}$. Since $\widehat{y}$ is the largest
(for~$\incord$) variable appearing in $f$, and its
coefficient in $f$ is negative, we conclude that $\min\{
f(\vec \lambda) \in \ZZ : \vec \lambda \text{ is a solution
to } \Psi \}$ is undefined, which in turn implies that  an
optimal solution minimizing $\vec c^\intercal \vec x$
subject to $\Phi(\vec x)$ does not exist.

Lastly, let us consider~\Cref{ip-gcd:optim:case3}. Again we
focus on the case of minimization. Below, let $C' \coloneqq
\{(\Psi, \vec u, E) \in C : \Psi \text{ is satisfiable over
} \NN \}$ and note that $\{ \vec x \in \ZZ^d : \Phi(\vec
x)\} = \sem{C'}$.
As~\Cref{ip-gcd:optim:case1,ip-gcd:optim:case2} do not hold,
$C' \neq \emptyset$ and every \gcdtodiv triple $(\Psi,\vec
u,E) \in C'$ is such that the linear polynomial $\vec
c^\intercal (\vec u +E \cdot \vec \lambda)$ only has
non-negative coefficients. Since the variables $\vec
\lambda$ are interpreted over $\NN$, this means that $\ell
\coloneqq \min\{ \vec c^\intercal \vec u : (\Psi, \vec u, E)
\in C' \}$ is a lower bound to the values that $\vec
c^{\intercal} \vec x$ can take when $\vec x$ is a solution
to $\Phi$; i.e., the optimal solution exists.
\Cref{lemma:gcd-to-div,lemma:from-IP-GCD-to-increasing-systems}
ensure that the lower bound~$\ell$ has polynomial bit length
with respect to $\bitlength{\Phi}$ and $\bitlength{\vec c}$.
We also have an upper bound~$u$ to the optimal solution: it
suffices to take the minimum of the values $ (\vec u +E
\cdot \vec \lambda)$, where $(\Psi, \vec u, E) \in C'$ and
$\vec \lambda$ is the positive integer solution to~$\Psi$
computed with~\Cref{func:new-solveincreasing} using the
solutions modulo $p \in \PP(\Psi)$
of~\Cref{lemma:ip-gcd:small-solutions-modulo}. Again, $u$
has polynomial bit length with respect to $\bitlength{\Phi}$
and $\bitlength{\vec c}$, thanks
to~\Cref{lemma:bound-on-pzero,lemma:gcd-to-div,lemma:from-IP-GCD-to-increasing-systems},
and~\Cref{theorem:local-to-global}.
\Cref{ip-gcd:optim:case3} then follows by reduction from the
feasibility problem of IP-GCD systems: it suffices to find
the minimal $v \in [\ell,u]$ such that the IP-GCD system
$\Phi_v(\vec x) \coloneqq \Phi(\vec x) \land (\vec
c^\intercal \vec x \leq v)$ is feasible. Since every $v \in
[\ell,u]$ is of polynomial bit length,
by~\Cref{corollary:IP-GCD-small-model} if $\Phi_v(\vec x)$
is satisfiable, then it has a solution $\vec x \in \ZZ^d$
such that $\bitlength{\vec x} \leq
\poly{\bitlength{\Phi},\bitlength{\vec c}}$.

%% file: gcd-to-divisibility-algo.tex
\begin{algorithm}[t]
  \caption{Translate a IP-GCD system into \gcdtodiv triples}
  \label{algorithm:from-IP-GCD-to-systems-of-divisibilities}
  \begin{algorithmic}[1]
    \medskip
    \Require An IP-GCD system $\Phi(\vec x) = A \cdot \vec x \leq \vec b \land \bigwedge_{i=1}^k \gcd(y_i,z_i) \sim_i c_i$ with $\vec x = (x_1,\dots,x_d)$.
    \Ensure 
    \begin{minipage}[t]{0.92\linewidth}
      A finite set $B$ of~\gcdtodiv triples 
      satisfying~$\{\vec a \in \ZZ^d : \vec a \text{ solution to }\Phi \} = \sem{B}$.
    \end{minipage}
    \medskip
    \State\label{algo:ip-to-gcd-a1:setA}%
      $G$ $\coloneqq$ 
      \begin{minipage}[t]{0.9\linewidth}
        $\{\Psi_1(\vec x),\dots,\Psi_\ell(\vec x)\}$ such that $\Phi$ is equivalent to $\bigvee_{i=1}^\ell \Psi_i$ and every $\Psi \in G$ is a IP-GCD\\[-1pt]
        system in which every GCD~constraint $\gcd(y,z) \sim c$ is such that (i) for both $w \in \{y,z\}$
        \\[-1pt]
        either $w \leq -1$ or $w \geq 1$ appear in $\Psi$, and (ii) the relation $\sim$ is either $=$ or $\geq$
      \end{minipage}
    \State $B \coloneqq \emptyset$
    \itComment{Set to be returned by the procedure}
    \For{$\Psi$ \textbf{in} $G$}
      \State\label{algo:ip-to-gcd-a1:lineqPsi}%
        $\Psi'$ $\coloneqq$ linear inequalities in $\Psi$
      \State\label{algo:ip-to-gcd-a1:VZG}%
        $S \coloneqq \{(\vec u_i, E_i) : i \in I \}$ s.t.~$\bigcup_{i \in I} \{ \vec u_i + E_i \cdot \vec y : \vec y \in \NN^{\ell}\}$ solutions set of $\Psi'$
      \itComment{\Cref{prop:vzGS}}
      \For{$(\vec u, E)$ \textbf{in} $S$}
        \State\label{algo:ip-to-gcd-a1:change-of-variables}%
          $\Psi'' \coloneqq $ 
            \begin{minipage}[t]{0.8\linewidth}
              system of GCD~constraints obtained from $\Psi$ by performing the change of\\[-1pt]
              variables $\vec x \gets \vec u + E \cdot \vec y$, where $\vec y$ is a vector of fresh variables (over $\NN$)
            \end{minipage}\vspace{4pt}
        \State\label{algo:ip-to-gcd-a1:adjust-polarity}%
          replace every polynomial~$f$ in $\Psi''$ having only negative coefficients or constant with $-f$\vspace{4pt}
        \State\label{algo:ip-to-gcd-a1:trans-positive}%
          \begin{minipage}[t]{0.9\linewidth}
            replace every constraint $\gcd(f,g) = c$ in $\Psi''$  with 
              $(c \div f) \land (c \div g) \land (f \div w) \land (g \div w+c)$,\\[-1pt]
              where $w$ is a fresh variable
              (distinct GCD~constraints gets distinct fresh variables)
          \end{minipage}\vspace{4pt}
        \State\label{algo:ip-to-gcd-a1:trans-negative}%
          \begin{minipage}[t]{0.9\linewidth}
            replace every constraint $\gcd(f,g) \geq c$ in $\Psi''$ with 
              $(z+c \div f) \land (z+c \div g)$,\\[-1pt]
              where $z$ is a fresh variable
              (distinct GCD~constraints gets distinct fresh variables)
          \end{minipage}\vspace{4pt}
        \State\label{algo:ip-to-gcd-a1:add-to-B}%
          \begin{minipage}[t]{0.9\linewidth}
            add to $B$ the triple $(\Psi'', \vec u, E')$ where $E'$ is obtained form $E$ by adding a zero column for each auxiliary variable $z$ and $w$ added in lines~\ref{algo:ip-to-gcd-a1:trans-positive} and~\ref{algo:ip-to-gcd-a1:trans-negative}
          \end{minipage}
      \EndFor
    \EndFor
    \State \textbf{return} $B$
  \end{algorithmic}
\end{algorithm}

%% file: gcd-to-increasing-algo.tex
\begin{algorithm}[t]
  \caption{Translates the systems in \gcdtodiv triples into $3$-increasing form}
  \label{algorithm:from-IP-GCD-to-increasing-systems}
  \begin{algorithmic}[1]
    \medskip
    \Require A finite set $B$ of \gcdtodiv triples.
    \Ensure 
    \begin{minipage}[t]{0.92\linewidth}
      A finite set $C$ of \gcdtodiv triples such that~$\sem{B} = \sem{C}$\\ and for every $(\Psi,\vec u, E) \in C$, $\Psi$~is a~$3$-increasing system of divisibility constraints.
    \end{minipage}
    \medskip
    \State $C \coloneqq \emptyset$ 
    \label{gcd-to-increasing:line-D}
    \itComment{Set to be returned by the procedure}
    \While{$(\Psi,\vec u, E) \gets \text{pop}(B)$} 
    \label{gcd-to-increasing:line-pop-B}
    \itComment{exits when $B$ becomes empty}
        \If{ $\module_f(\Psi) \cap \ZZ = \{0\}$ for every non-constant $f$ primitive part of some l.h.s.~in~$\Psi$}
        \label{gcd-to-increasing:line-test-increasing}
        \State add to $C$ the triple $(\Psi, \vec u, E)$ 
        \itComment{$\Psi$ in increasing form}
        \Else
          \State $f \coloneqq$ non-constant primitive part of some l.h.s.~in~$\Psi$, satisfying~$\module_f(\Psi) \cap \ZZ \neq \{0\}$
          \label{gcd-to-increasing:line-find-f}
          \State $\lambda_1,\dots,\lambda_j \coloneqq $ the variables appearing in $f$
          \State $c \coloneqq $ minimum positive integer in $\module_f(\Psi)$ 
          \label{gcd-to-increasing:line-compute-c}
          \For{$\vec \nu \colon \{\lambda_1,\dots,\lambda_j\} \to [0,c]$ such that $f(\vec \nu(\lambda_1),\dots,\vec \nu(\lambda_j))$ divides $c$}
          \label{gcd-to-increasing:line-for-loop}
            \State $\Psi_{\vec \nu }\coloneqq \Psi\substitute{\vec \nu(\lambda_i)}{\lambda_i : i \in [1,j]}$
            \label{gcd-to-increasing:line-new-Psi}
            \itComment{$\Psi_{\vec \nu}$ has fewer variables than $\Psi$}
            \State $\vec u_{\vec \nu} \coloneqq \vec u + \sum_{i=1}^j \vec \nu(\lambda_i) \cdot \vec p_i$ where $\vec p_i$ is the column of $E$ corresponding to the 
            variable~$\lambda_i$%
            \label{gcd-to-increasing:line-new-u}
            \State $E_{\vec \nu} \coloneqq$ $E$ without the columns corresponding to $\lambda_1,\dots,\lambda_j$
            \label{gcd-to-increasing:line-new-E}
            \State add to $B$ the triple $(\Psi_{\vec \nu}, \vec u_{\vec \nu}, E_{\vec \nu})$
            \itComment{triple to be considered again in line~\ref{gcd-to-increasing:line-pop-B}}
            \label{gcd-to-increasing:line-new-triple}
          \EndFor
        \EndIf
    \EndWhile
    \State \textbf{return} $C$
    \label{gcd-to-increasing:line-return}
  \end{algorithmic}
\end{algorithm}

%% file: appendix-sieve.tex
\section{\Cref{lem:extended-brun}: proof of~\Cref{lem:extended-brun:left-term}}
\label{appendix:sieve}

In this appendix, we present the technical manipulation yielding~\Cref{lem:extended-brun:left-term}, hence finishing the proof of~\Cref{lem:extended-brun}. Below, $\mu$ and $\omega$ stand for the M\"obius function and the prime omega function, respectively. Recall that $\mu(n) = (-1)^{\omega(n)}$ and $\omega(n) = \card{\PP(n)}$, for every $n \in \pZZ$.

\begin{proposition}[M\"obius inversion~{\cite[Theorem~266]{Hardy75}}]
    \label{lemma:Mobius-inversion}
    Consider two functions $f,g \colon \pZZ \to \RR$ such that for every $n \in \pZZ$, $f(n) = \sum_{d \in \divisors(n)} g(d)$. For every $m \in \pZZ$,  
    $g(m) = \sum_{d \in \divisors(m)} f(d) \cdot \mu(\frac{m}{d})$.
\end{proposition}

\begin{proposition}[M\"obius sums~{\cite[Theorem~263]{Hardy75}}]
    \label{lemma:Mobius-sums}
    For $n \in \pZZ$ greater than~$1$, $\sum_{s \in \divisors(n)} \mu(s) = 0$.
\end{proposition}

The following lemma tells us what to expect when we truncate the sum of the previous proposition so that it only considers elements with at most $\ell$ divisors.

\begin{lemma}
    \label{lemma:Mobius-Pascal}
    Let $n, \ell\,{\in}\,\NN$ with $n$ square-free. If $\omega(n) > \ell$ then  $\sum_{r \in \divisors(n),\, \omega(r) \leq \ell}\,\mu(r) = (-1)^\ell {\omega(n)-1 \choose \ell}$.
\end{lemma}

\begin{proof}
    We write LHS (resp.~RHS) for the left-hand (resp.~right-hand) side of the equivalence in the statement. 
    Note that $\omega(n) > \ell$ implies $n \geq 1$.
    The proof is by induction on $\ell$. 

    \proofparagraph{Base case: $\ell = 0$} In this case, 
    $\text{LHS}  = \mu(1) = 1 = (-1)^0 {\omega(n)-1 \choose 0} = \text{RHS}$.

    \proofparagraph{Induction step: $\ell \geq 1$}We have,
    {\allowdisplaybreaks
    \begin{align*}
        \text{LHS} 
        &= \sum_{r \in \divisors(n),\, \omega(r) < \ell} \mu(r) 
        + \sum_{s \in \divisors(n),\, \omega(r) = \ell} \mu(s)\\
        &= (-1)^{\ell-1}{\omega(n) - 1 \choose \ell-1}
        + \sum_{s \in \divisors(n),\, \omega(r) = \ell} \mu(s)
        &
        \begin{aligned}
            \text{by induction hypothesis;}\\ 
            \text{recall $\omega(n) > \ell$}
        \end{aligned}\\
        &= (-1)^{\ell-1} \left( {\omega(n) - 1 \choose \ell-1} - \sum_{r \in \divisors(n),\, \omega(r) = \ell} 1 \right)
        &\text{since $\mu(r) = (-1)^\ell$ iff $\omega(r) = \ell$}\\
        &= (-1)^{\ell-1} \left( {\omega(n) - 1 \choose \ell-1} - 
        {\omega(n) \choose \ell} \right)
        &\text{from $n$ square-free}\\ 
        & = (-1)^{\ell} {\omega(n) -1 \choose \ell} = \text{RHS}
        &\text{Pascal's rule.}
        & \qedhere
    \end{align*}
    }
\end{proof}



We are now ready to prove~\Cref{lem:extended-brun:left-term}:

\BrunLBLeftTerm*

\noindent
Let us recall the hypothesis under which this claim must be proved: $\ell \in \NN_+$ is odd, $d \geq 1$, $Q$ is a non-empty finite set of primes, $Q(\ell) \coloneqq \{ r \in \divisors(\setprod Q) : \omega(r) \leq \ell \}$, 
$m$ is a multiplicative function such that $m(q) \leq q-1$ and $m(q) \leq d$ on all $q \in Q$, and $W_m(Q) \coloneqq \prod_{q \in Q} \left( 1 - \frac{m(q)}{q} \right)$.

\begin{proof}
We start by defining the truncated M\"obius function $\mu_{\ell}$ and its companion function $\psi_\ell$:
\[
    \mu_{\ell}(x) \coloneqq 
    \begin{cases}
        \mu(x) &\text{if $\omega(x) \leq \ell$}\\
        0       &\text{otherwise}
    \end{cases}
    \qquad\text{ and }\qquad
    \psi_{\ell}(x) \coloneqq \sum_{r \in \divisors(x)} \mu_{\ell}(x).
\]
The proof proceeds by performing two term manipulations. 
In the first one, we use the fact that $m$ is multiplicative, 
together with properties of the M\"obius function (e.g.~\Cref{lemma:Mobius-inversion}), to show that
\begin{equation}
    \label{eq:claim-3:first}
    \sum_{r \in Q(\ell)} \frac{\mu(r) \cdot m(r)}{r} = W_m(Q) \cdot 
    \left( 
        1    
        +
        \sum_{\substack{s \in \divisors(\setprod Q)\\ \omega(s) > \ell}}
        \frac{\psi_{\ell}(s) \cdot m(s)}{s \cdot W_m(\PP(s))}
    \right).
\end{equation}
In the second manipulation, we look at the sum $\sum_{s \in \divisors(\setprod Q) \setminus \{1\}}
\frac{\psi_{\ell}(s) \cdot m(s)}{s \cdot W_m(\PP(s))}$ from the equation above, and (also thanks to~\Cref{lemma:Mobius-Pascal}) bound it in absolute terms as follows: 
\begin{equation} 
    \label{eq:claim-3:second}
    \left| 
        \sum_{\substack{s \in \divisors(\setprod Q)\\ \omega(s) > \ell}}
        \frac{\psi_{\ell}(s) \cdot m(s)}{s \cdot W_m(\PP(s))}
    \right|
    \leq \left(\frac{e \cdot \alpha}{\ell}\right)^\ell \cdot \alpha \cdot e^{\alpha}, \ 
    \text{ where }\alpha \coloneqq (d+1)^2 (2 + \ln \ln (\card{Q} + 1)).
\end{equation}

\noindent
\Cref{lem:extended-brun:left-term}
follows directly from~\Cref{eq:claim-3:first} and~\Cref{eq:claim-3:second}. Note that these equations can be used to also establish the upper bound to $\sum_{r \in Q(\ell)} \frac{\mu(r) \cdot m(r)}{r}$ required for the upper bound of~\Cref{lem:extended-brun}.

\proofparagraph{Manipulation resulting in~\Cref{eq:claim-3:first}}
{\allowdisplaybreaks
\begin{flalign*}
    &\ \ \sum_{r \in Q(\ell)}\frac{\mu(r) \cdot m(r)}{r}\\ 
    =\,& \sum_{r \in \divisors(\setprod Q)}\frac{\mu_\ell(r) \cdot m(r)}{r}
    &\text{by def.~of $\mu_\ell$}\\
    =\,& \sum_{r \in \divisors(\setprod Q)}\frac{\left(\sum_{s \in \divisors(r)}\psi_\ell(s) \cdot \mu \left(\frac{r}{s}\right)\right) \cdot m(r)}{r} 
    &\text{by~\Cref{lemma:Mobius-inversion}}\\
    =\,& \sum_{r \in \divisors(\setprod Q)}\sum_{s \in \divisors(r)}\frac{\psi_\ell(s) \cdot \mu\left(\frac{r}{s}\right) \cdot m(r)}{r}\\
    =\,& \sum_{s \in \divisors(\setprod Q)} \sum_{r \in \divisors\left(\frac{\setprod Q}{s}\right)}\frac{\psi_\ell(s) \cdot \mu (r) \cdot m(r \cdot s)}{r \cdot s} 
    &
    \begin{aligned}
        \text{invert summations using the}\\
        \text{change of variable $r \gets r \cdot s$}
    \end{aligned}\\ 
    =\,& \sum_{s \in \divisors(\setprod Q)}\frac{\psi_\ell(s) \cdot m(s)}{s} \cdot \sum_{r \in \divisors\left(\frac{\setprod Q}{s}\right)}\frac{ \mu(r) \cdot m(r)}{r}
    &\text{multiplicity of~$m$}\\
    =\,& \sum_{s \in \divisors(\setprod Q)}\frac{\psi_\ell(s) \cdot m(s)}{s} \cdot \prod_{q \in Q \setminus \divisors(s)}\left(1 + \frac{\mu(q) \cdot m(q)}{q}\right)
    &
    \begin{aligned}
        \text{multiplicity of~$\mu$ and~$m$;}\\
        \text{factorization thanks to $r$ being}\\
        \text{square-free, for all $r \in \divisors\!\textstyle\left(\frac{\setprod Q}{s}\right)$}
    \end{aligned}\\
    =\,& \sum_{s \in \divisors(\setprod Q)}\frac{\psi_\ell(s) \cdot m(s)}{s} \cdot \frac{ \prod_{q \in Q}\left(1 - \frac{m(q)}{q}\right) }{\prod_{q \in \PP(s)}\left(1 - \frac{m(q)}{q}\right)}
    &
    \begin{aligned}
        \text{$\mu(q) = -1$ for $q$ prime}\\
        \text{and simple manipulation}
    \end{aligned}\\
    =\,& \sum_{s \in \divisors(\setprod Q)}\frac{\psi_\ell(s) \cdot m(s)}{s} \cdot \frac{ W_m(Q) }{W_m(\PP(s))}
    &
    \text{by def.~of~$W_m$}\\
    =\,&\, W_m(Q) \cdot \sum_{s \in \divisors(\setprod Q)}\frac{\psi_\ell(s) \cdot m(s)}{s \cdot W_m(\PP(s))}\\
    =\,&\,W_m(Q) \cdot 
    \left( 
        \sum_{s \in Q(\ell)}
        \frac{\psi_{\ell}(s) \cdot m(s)}{s \cdot W_m(\PP(s))}    
        +
        \sum_{\substack{s \in \divisors(\setprod Q)\\ \omega(s) > \ell}}
        \frac{\psi_{\ell}(s) \cdot m(s)}{s \cdot W_m(\PP(s))}
    \right) 
    &
    \begin{aligned}
        \text{split depending on $\omega(s) \leq \ell$,}\\
        \text{and by def.~of~$Q(\ell)$}
    \end{aligned}
    \\
    =\,&\,W_m(Q) \cdot 
    \left( 
        \sum_{s \in Q(\ell)}
        \frac{\left(\sum_{r \in \divisors(s)} \mu(r)\right) \cdot m(s)}{s \cdot W_m(\PP(s))}    
        +
        \sum_{\substack{s \in \divisors(\setprod Q)\\ \omega(s) > \ell}}
        \frac{\psi_{\ell}(s) \cdot m(s)}{s \cdot W_m(\PP(s))}
    \right) 
    \hspace{-1.4cm}
    &\text{def.~of~$\psi_\ell$}\\
    =\,&\,W_m(Q) \cdot 
    \left( 
        1    
        +
        \sum_{\substack{s \in \divisors(\setprod Q)\\ \omega(s) > \ell}}
        \frac{\psi_{\ell}(s) \cdot m(s)}{s \cdot W_m(\PP(s))}
    \right) 
    &
    \begin{aligned}
        \text{in the left summation:}\\ 
        \text{for $s = 1$ the addend is $1$,}\\
        \text{and for $s > 1$ the addend is $0$}\\ 
        \text{by~\Cref{lemma:Mobius-sums}.}
    \end{aligned}
\end{flalign*}
}%
\proofparagraph{Manipulation resulting in~\Cref{eq:claim-3:second}}
{\allowdisplaybreaks
    \begin{flalign*}
        \left|
            \sum_{\substack{s \in \divisors(\setprod Q)\\ \omega(s) > \ell}}
            \frac{\psi_{\ell}(s) \cdot m(s)}{s \cdot W_m(\PP(s))}
        \right|
        \leq\,&
        \sum_{\substack{s \in \divisors(\setprod Q)\\ \omega(s) > \ell}} 
        {\omega(s) - 1 \choose \ell} \cdot \frac{m(s)}{s \cdot W_m(\PP(s))}
        &\text{by~\Cref{lemma:Mobius-Pascal} and def.~of~$\psi_\ell$}\\ 
        =\,& 
        \sum_{k = \ell + 1}^{\card{Q}} \left(\binom{k-1}{\ell} \cdot \sum_{\substack{s \in \divisors(\setprod Q)\\ 
        \omega(s) = k}} \frac{m(s)}{s \cdot W_m(\PP(s))} \right)
        &\text{split on the value of $\omega(s)$.} 
    \end{flalign*}
}%
    \noindent
    We focus on the summation~$\sum_{s \in \divisors(\setprod Q)
    ,\,\omega(s) = k} \frac{m(s)}{s \cdot W_m(\PP(s))}$.
    Since the function $m$ is multiplicative, and similarly~$W_m(A \cup B) = W_m(A) \cdot W_m(B)$ for $A,B$ disjoint finite sets of primes (and $W_m(\emptyset) = 1$ by definition), for $k \geq 1$ we have:
    \begin{align*}
        \sum_{\substack{s \in \divisors(\setprod Q)\\ 
        \omega(s) = k}} 
        \frac{m(s)}{s \cdot W_m(\PP(s))}
        &= \sum_{q_1<...<q_k \in Q}\left( \prod_{i=1}^k\frac{m(q_i)}{q_i \cdot W_m(\{q_i\})} \right)
        \leq \frac{1}{k!} \sum_{q_1,...,q_k \in Q}\left( \prod_{i=1}^k\frac{m(q_i)}{q_i \cdot W_m(\{q_i\})}\right)\\
        &= \frac{1}{k!}\left(\sum_{q \in Q} \frac{m(q)}{q \cdot W_m(\{q\})}\right)^k
        = \frac{1}{k!}\left(\sum_{q \in Q} \frac{m(q)}{q - m(q)} \right)^k.
    \end{align*}
    We further analyse the summation $\sum_{q \in Q} \frac{m(q)}{q-m(q)}$. Below, we write $Q_{d+1}$ for the set of the first $\min(\card{Q},d+1)$ many primes in $Q$ (recall $d \geq 1$), and denote by $p_i$ the $i$-th prime.
{\allowdisplaybreaks
    \begin{flalign*}
        \sum_{q \in Q} \frac{m(q)}{q - m(q)}
        &= 
        \sum_{q \in Q_{d+1}} \frac{m(q)}{q - m(q)} 
        + 
        \sum_{q \in Q \setminus Q_{d+1}} \frac{m(q)}{q - m(q)} \\
        &\leq \sum_{q \in Q_{d+1}} d 
        + 
        \sum_{q \in Q \setminus Q_{d+1}} \frac{m(q)}{q - m(q)}
        &
        \begin{aligned}
            \text{since } m(q) \leq d\\
            \text{and } q-m(q) \geq 1
        \end{aligned}\\
        &\leq d \cdot (d+1) + 
        \sum_{q \in Q \setminus Q_{d+1}} \frac{d}{q - d}
        & \text{$m(q) \leq d < q$, for all $q \in Q \setminus Q_{d+1}$}\\
        &\leq d \cdot (d+1) +
        \sum_{i=d+2}^{\card{Q}}
        \frac{d}{p_i - d}
        & \begin{aligned}
            \text{def.~of~$Q \setminus Q_{d+1}$}\\
            \text{and $p_i > d$ for $i \geq d+2$}
        \end{aligned}\\
        & \leq d \cdot (d+1) + 
        d \cdot \sum_{i=d+2}^{\card{Q}} \frac{1}{(i \ln i) - d}
        & 
        \begin{aligned}
            p_i \geq i \ln i \text{ \cite{Rosser39}}\\
            \text{and $i \ln i > d$ for $i \geq d+2$}
        \end{aligned}\\
        & \leq d \cdot (d+1) + d \cdot (d+1) \sum_{i=d+2}^{\card{Q}} \frac{1}{i \ln i}
        &
        \begin{aligned}
            \text{since } \ \frac{1}{x \ln x - y} \leq \frac{y+1}{x \ln x}\\ 
            \text{for all } x \geq 3 \text{ and } 0 \leq y \leq x-1 
        \end{aligned}\\
        & \leq d \cdot (d+1) \cdot \Big( 1 + \sum_{i=3}^{\card{Q}} \frac{1}{i \ln i} \Big)\\
        & \leq d \cdot (d+1) \cdot \Big( 1 + 
        \int_{2}^{\card{Q}+1} \frac{1}{x \ln x}\textup{d}x\Big)
        &\begin{aligned}
            \text{Riemann over-approximation}\\ 
            \text{note: $\card Q + 1 \geq 2$} 
        \end{aligned}\\
        & \leq d \cdot (d+1) \cdot (1 + \ln \ln (\card Q+1) - \ln \ln 2)\\
        & \leq (d+1)^2 (2 + \ln \ln (\card{Q} + 1)) \ = \ \alpha.
    \end{flalign*}%
}
We combine this bound with the previous two to obtain complete the proof of~\Cref{eq:claim-3:second}:
{\allowdisplaybreaks
\begin{flalign*}
    \left|
        \sum_{\substack{s \in \divisors(\setprod Q)\\ \omega(s) > \ell}}
        \frac{\psi_{\ell}(s) \cdot m(s)}{s \cdot W_m(\PP(s))}
    \right|
    &\leq
    \sum_{k = \ell + 1}^{\card{Q}} \left(\binom{k-1}{\ell} \cdot \frac{1}{k!} \cdot \alpha^k \right)\\
    &= 
    \sum_{j = 0}^{\card{Q}-\ell-1} \left(\binom{\ell+j}{\ell} \cdot \frac{1}{(\ell+1+j)!} \cdot \alpha^{\ell+1+j} \right)
    &
    \begin{aligned}
        \text{change of variable}\\ 
        k \gets \ell + 1 + j
    \end{aligned}
    \\
    &= 
    \sum_{j = 0}^{\card{Q}-\ell-1} \left(\frac{(\ell+j)!}{\ell! \cdot j! \cdot (\ell+1+j)!} \cdot \alpha^{\ell+1+j} \right)\\
    & \leq 
    \frac{\alpha^{\ell+1}}{\ell!} \cdot \sum_{j = 0}^{\infty} \frac{\alpha^{j}}{j!}
    &
    \hspace{-0.5cm}
    \begin{aligned}
        \text{note: all terms in the}\\
        \text{summation are non-negative}
    \end{aligned}
    \\
    & \leq \frac{\alpha^{\ell+1}}{\ell!} \cdot e^{\alpha} 
    & \begin{aligned}
        \text{def.~of $e^x$ as a series}\\
        \text{i.e.,~$\textstyle e^x = \sum_{i=0}^\infty \frac{x^i}{i!}$}
    \end{aligned}\\
    & \leq \left(\frac{e \cdot \alpha}{\ell}\right)^\ell \cdot \alpha \cdot e^{\alpha}
    & \text{from } x! \geq \frac{x^x}{e^x}.
\end{flalign*}
}
This completes the proof 
of~\Cref{lem:extended-brun:left-term}.
\end{proof}

%% file: appendix-crt.tex
\section{\Cref{thm:mixed-crt}: proofs of~\Cref{thm:mixed-crt:claim1} and~\Cref{claim:CRT:bound-on-W}}
\label{appendix:crt}

The mathematical objects appearing in the statements of the two claims below are defined in the proof of~\Cref{thm:mixed-crt} and the statement of~\Cref{lem:extended-brun}; see~\Cref{section:CRT}.

\ThmMixedCRTClaimOne*

\begin{proof} 
  Recall that $A = [k,k+z] \cap S_M$, and so $A \cap
  S_{\alpha,r} = [k,k+z] \cap S_M \cap S_{\alpha,r}$. Since that elements in $M \cup Q$
  are pairwise coprime and $M \cap Q = \emptyset$, 
  we can apply the CRT and conclude that 
  $S_M \cap S_{\alpha,r}$ is an
  arithmetic progression with period $r \cdot \setprod M$. Let
  $u$ be the largest element of $S_M  \cap S_{\alpha,r}$
  that is strictly smaller than $k$. By definition of $u$
  and from the fact that $S_M  \cap S_{\alpha,r}$ has period
  $r \cdot \setprod M$, we get $\card{(A \cap S_{\alpha,r})} =
  \floor{\frac{k+z-u}{r \cdot \setprod M}}$. Similarly, because
  $S_M$ is periodic in $\setprod M$, $\floor{\frac{k+z-u}{\setprod
  M}}$ is over counting $\card{A}$ by at most $r-1$, i.e.,
  there is $\tau_{\alpha,r} \in [0,r-1]$ such that $\card{A} =
  \floor{\frac{k+z-u}{\setprod M}} - \tau_{\alpha,r}$. Since
  $\floor{\frac{a}{b}} = \floor{\frac{\floor{a}}{b}}$ for
  every $a \in \RR$ and $b \in \pZZ$, we get $\card{(A \cap
  S_{\alpha,r})} = \floor{\frac{1}{r} \cdot (\card{A} +
  \tau_{\alpha,r})}$. With a simple manipulation using
  $\floor{a} + \floor{b} \leq \floor{a+b} \leq \floor{a} +
  \floor{b} + 1$ and $\floor{\frac{\tau_{\alpha,r}}{r}} = 0$,
  we derive $\frac{\card{A}}{r} - 1 \leq \card{(A \cap
  S_{\alpha,r})} \leq \frac{\card{A}}{r} + 1$.
  \end{proof}

  \ClaimCRTBoundOnW*

  \begin{proof}
    Let $Q_d$ be the set containing the $\min(\card{Q},d)$ smallest primes in $Q$.
    Recall
    that by definition $m(q) \leq d \leq q-1$ for every $q \in Q$. We have,
    \[ 
      W_m(Q)^{-1} = \prod_{q \in Q} \frac{q}{q-m(q)} 
      \ \leq \prod_{q \in Q} \frac{q}{q-d}
      \ \leq \ \prod_{q \in Q_d} \frac{q}{q-d} \cdot\! \prod_{q \in Q \setminus Q_d} \frac{q}{q-d}
      \ \leq \ (d+1)^d \cdot\! \prod_{q \in Q \setminus Q_d} \frac{q}{q-d},
    \]
    where the last inequality holds because $\frac{x}{x-c} \leq c+1$ for every $x \geq c+1$ and $c \in \pZZ$. 
    Below, let us denote by $p_i$ the $i$-th prime.
    We further inspect the product $\prod_{q \in Q \setminus Q_d} \frac{q}{q-d}$:

    \begin{flalign*}
      &\prod_{q \in Q \setminus Q_d} \frac{q}{q-d} 
      \leq \prod_{i = d+1}^{\card{Q}}\frac{p_i}{p_i-d}
      \leq \prod_{i = d+1}^{\card{Q}}\frac{i \cdot \ln i}{i \cdot \ln i-d}
      &
      \hspace{-5cm}
      \begin{aligned}
        \text{$p_i \geq i \cdot \ln i$ for all $i \in \pZZ$, see~\cite{Rosser39};}\\
        \text{$x \mapsto \frac{x}{x-d}$ decreasing for $x > 1$}
      \end{aligned}\\
      \leq \,
      &\exp\left({\sum_{i=d+1}^{\card{Q}} \ln\Big(\frac{i \cdot \ln i}{i \cdot \ln i-d}\Big) }\right)
      = \exp\left({-\sum_{i=d+1}^{\card{Q}} \ln\Big(1-\frac{d}{i \cdot \ln i}\Big) }\right)
      \\
      \leq \, 
      & 
      \exp\left({\sum_{i=d+1}^{\card{Q}} \frac{3 \cdot d}{i \cdot \ln i} }\right) 
      \leq \exp\left({\sum_{i=2}^{\card{Q}} \frac{3 \cdot d}{i \cdot \ln i} }\right)
      &
      \hspace{-10cm}
      \begin{aligned}
        \text{first term from $\ln\Big(1-\frac{1}{x}\Big) \geq -\frac{3}{x}$ for all $x \geq \ln 3$;}\\
        \text{for corner case $d = 1$ and $i=2$, note $2 \ln 2 > \ln 3$}
      \end{aligned}\\
      \leq \, 
      & 
      \exp\left(\frac{3 \cdot d}{2 \cdot \ln 2} + {\sum_{i=3}^{\card{Q}} \frac{3 \cdot d}{i \cdot \ln i} }\right)
      \leq 
      \exp\left(\frac{3 \cdot d}{2 \cdot \ln 2} + 
      \int_2^{\card{Q}+1} \frac{3 \cdot d}{x \ln x} \textup{d}x
      \right)
      &
      \hspace{-5cm}
      \begin{aligned}
        \text{Riemann over-approximation}\\
        \text{note: $\card{Q}+1 \geq 2$}
      \end{aligned}\\
      \leq \, 
      & 
      \exp\left(\frac{3 \cdot d}{2 \cdot \ln 2} + 
      3 \cdot d \cdot \big( \ln \ln (\card{Q}+1) - \ln \ln 2\,\big)
      \right)
      \leq 
      \exp\left(
      3 \cdot d \cdot \big(2 + \ln \ln (\card{Q}+1)\big)
      \right).
    \end{flalign*}
    We plug this bound on the afore-derived 
    bound for $W_m(Q)^{-1}$ to complete the proof
    of~\Cref{claim:CRT:bound-on-W}:
    \begin{align*}
      W_m(Q)^{-1} 
      &\leq\, (d+1)^d  \exp\left(
        3 \cdot d \cdot \big(2+ \ln \ln (\card{Q}+1)\big)
        \right)
        \,\leq\, 
        (d+1)^d \cdot e^{6 \cdot d} \ln(\card{Q}+1)^{3 \cdot d}\\
      &\leq\, (d+1)^d \cdot 2^{9 \cdot d} \ln(\card{Q}+1)^{3 \cdot d}
      \,\leq\, (d+1)^{10 \cdot d} \ln(\card{Q}+1)^{3 \cdot d}.
      \qedhere
    \end{align*}
  \end{proof}

%% file: appendix-module-basis.tex
\section{Algorithms related to the elimination property}
\label{appendix:module-basis}

In this appendix we establish~\Cref{lemma:module-span} and~\Cref{lemma:add-elimination-property}. 
Proving these lemmas require the standard notion of kernel and Hermite normal form of a matrix, which we now recall for completeness. 
Consider a matrix~${A \in \ZZ^{n \times d}}$. 
The \emph{kernel} of $A$ is the vector space $\ker(A) \coloneqq \{ \vec v \in \ZZ^d : A \cdot \vec v = \vec 0 \}$. 
We represent~\emph{bases} of $\ker(A)$ as matrices~$K \in \ZZ^{d \times (d-r)}$, where $r$ is the rank of $A$ and $\ker(A) = \{ K \cdot \vec v : \vec v \in \ZZ^{d-r} \}$.
A matrix $H \in \ZZ^{n \times d}$ is said to be the \emph{column-style Hermite normal form} of $A$ (\textit{HNF}, in short) if there is a square unimodular matrix $U \in \ZZ^{d \times d}$ such that $H = A \cdot U$ and
\begin{enumerate}
  \item $H$ is lower triangular,  
  \item the \emph{pivot} (i.e., the first non-zero entry in a column, from the top) of a non-zero column is positive and it is strictly below the pivot of the column before it, and
  \item elements to the right of pivots are $0$ and elements to the left are non-negative and smaller than the pivot. 
\end{enumerate} 
Recall that $U$ being unimodular means that it is invertible over the integers.

Given a vector $\vec v$, we write $\vec v[i]$ 
for the $i$-th entry of $\vec v$, starting at $i = 1$. 
Similarly, for a matrix~$A$, we write $A[i]$ for its $i$-th row, again starting at $i = 1$.

\begin{proposition}[{\cite[Section~4.2]{Schrijver99}}]
  \label{propositoin:HNF-lattice}
  The HNF~$H$ of a matrix $A \in \ZZ^{n \times d}$ always exits, it is unique, and $A$ and $H$ generate the same lattice, i.e., $\{ A \cdot \vec \lambda : \vec \lambda \in \ZZ^d \} = \{ H \cdot \vec \lambda : \vec \lambda \in \ZZ^d\}$. 
\end{proposition}

The following proposition refers to the {LLL-based} algorithm for the HNF~in~\cite{HavasMM98}. A basis for the integer kernel can be retrieved from the HNF~together with the associated unimodular matrix.

\begin{proposition}[\cite{Vanderkellen00}]
  \label{prop:HNF}
  There is a \ptime algorithm computing a basis $K$
  of the integer kernel and the HNF~$H$ 
  of an input matrix $A \in \ZZ^{n \times d}$. 
  The algorithm yields ${\norminf{K},\norminf{H}
  \leq (n \cdot
  \norminf{A} + 1)^{\bigO{n}}}$.
\end{proposition}

Note that we can also upper bound the GCDs of the rows of the integer kernel $K$ in terms of the rank of $A$ by appealing to~\Cref{prop:vzGS}.

\begin{corollary}
  \label{corollary:kernel-gcd}
  Consider a basis $K$ of the integer kernel of a matrix $A \in \ZZ^{n \times d}$. Let $r \coloneqq \rank(A)$. 
  For every $i \in [1,d]$, $\norminf{\gcd(K[i])} \leq (d+1) \cdot (r \cdot \max(2,\norminf{A}))^r$.
\end{corollary}

\subsection{Computing a set spanning the divisibility module}

\LemmaModuleSpan*

This lemma follows from the forthcoming~\Cref{prop:correctness-span-divmod} and~\Cref{theorem:runtime-span-divmod}.

For the whole section, let $\Phi \coloneqq \bigwedge_{i=1}^m f_i \div g_i$ and $f$ be a primitive polynomial.
As already explained in~\Cref{sec:divisibility-algo}, 
the algorithm~\Cref{lemma:module-span} refers to performs a fix-point
computation where, at the $\ell$-th iteration, the values contained in $\vec
v$ characterize a spanning set of a particular submodule $\module_f^\ell(\Phi)$ of $\module_f(\Phi)$.
More precisely, we define $\module^0_f(\Phi) \subseteq
\module^1_f(\Phi) \subseteq \dots \subseteq \module^\ell_f(\Phi)
\subseteq \dots$ to be the sequence of sets given by 
\begin{enumerate}
  \item\label{Mellf:base-case} $\module_f^0(\Phi) \coloneqq \ZZ f$, and
  \item\label{Mellf:ind-case} for $\ell \in \NN$, $\module_f^{\ell+1}(\Phi) \coloneqq \module_f^{\ell}(\Phi) + \left\{ \textstyle\sum_{j=1}^m a_j \cdot
  g_j : \text{ for all } i \in [1,m],\, a_i \in \ZZ \text{ and } a_i \cdot f_i \in
  \module_f^{\ell}(\Phi)\right\}\!.$
\end{enumerate}
Let $\ell \in \NN$.
Note that, by definition, $\module_f^{\ell}(\Phi)$ is a $\ZZ$-module and moreover if $\ZZ f_i \cap \module_{f}^\ell(\Phi) = \{0\}$ for some $i \in [1,m]$, then $a_i$ in the definition of $\module_f^{\ell+1}(\Phi)$ equals $0$.
We define the \emph{canonical representation} of $\module_f^\ell(\Phi)$ as the vector $(v_1,\dots,v_m) \in \NN^{m}$ such that for every $i \in [1,m]$,
\begin{itemize}
  \item if $\ell = 0$ then $v_i \coloneqq 0$, 
  \item if $\ell \geq 1$ then $v_i \coloneqq \gcd\{\lambda \in \NN : \lambda \cdot f_i \in \module_f^{\ell-1}(\Phi)\}$.
\end{itemize}
\Cref{lemma:spanMellf} shows that this vector represents a spanning set of~$\module_f^\ell(\Phi)$, but first we need an auxiliary lemma.

\begin{lemma}
  \label{lemma:Mellf2}
  Let $\ell \in \NN$. Let $(v_1,\dots,v_m)$ and
  $(v_1',\dots,v_m')$ be the canonical
  representations of $\module_f^\ell(\Phi)$ and
  $\module_f^{\ell+1}(\Phi)$, respectively. For every $i \in
  [1,m]$, $v_i = v_i' = 0$ or $v_i'$ divides
  $v_i$ (so, $v_i' \neq 0$ if $v_i \neq 0$).
\end{lemma}

\begin{proof}
  Let $i \in [1,m]$. If $v_i = 0$ then either $v_i'$ is
  $0$ or it divides $v_i$, hence the statement is
  trivially satisfied for that particular $i$.  
  Suppose that $v_i \neq 0$. By definition of canonical
  representation, $\ell \geq 1$ and $v_i \cdot f_i \in M^{\ell-1}_f(\Phi)$. By definition of~$\module_f^{\ell}(\Phi)$, 
  we conclude that $v_i \cdot f_i
  \in M^{\ell}_f(\Phi)$. By definition of canonical
  representation $v_i' = \gcd\{\lambda \in \NN
  : \lambda \cdot f_i \in \module_f^{\ell}(\Phi)\}$, and
  therefore $v_i'$ divides $v_i$. 
\end{proof}

\begin{lemma}
  \label{lemma:spanMellf}
  Let $\ell \in \NN$ and let $(v_1,\dots,v_m) \in
  \NN^{m}$ be the canonical representation of
  $\module_f^{\ell}(\Phi)$. Then, the set of linear
  polynomials $\{f,\,v_1 \cdot g_1,\,\dots,\,v_m \cdot g_m\}$
  spans $M^\ell_f(\Phi)$.
\end{lemma}

\begin{proof}
  The statement follows by induction on $\ell \in \NN$. 
  \begin{description}
    \item[base case $\ell = 0$.]
    From $\module_f^0(\Phi) = \ZZ f$ we have
    $(v_1,\dots,v_m) = (0,\dots,0)$ and $\{f\}$ spans
    $M^0_f(\Phi)$.
    \item[induction step $\ell \geq 1$.] From the induction
      hypothesis, ${\{f,\,v_1^* \cdot g_1,\,\dots,\,v_m^* \cdot g_m\}}$ 
      spans $M^{\ell-1}_f(\Phi)$; with $(v_1^*,\dots,v_m^*)$ being the canonical representation
       of $M^{\ell-1}_f(\Phi)$.
      We consider the two inclusions of the equivalence
      $\ZZ f + \ZZ (v_1 \cdot g_1) + \dots + \ZZ(v_m \cdot g_m) = M^\ell_f(\Phi)$.

      \proofSubset This direction follows directly by definition of $M^\ell_f(\Phi)$.

      \proofSupset
      Let $h \in \module_f^{\ell}(\Phi)$.
      By definition, 
      $h = h_1 + h_2$ where ${h_1 \in \ZZ f + \ZZ (v_1^* \cdot g_1) + \dots + \ZZ(v_m^* \cdot g_m)}$
      and $h_2 = \sum_{i=1}^m a_i \cdot g_i \in
      \module_f^\ell(\Phi)$ satisfying ${a_i \cdot f_i \in
      \module_f^{\ell-1}(\Phi)}$ for every~$i \in [1,m]$. 
      By~\Cref{lemma:Mellf2} $\ZZ(v_i^* \cdot g_i) \subseteq \ZZ(v_i \cdot g_i)$ and therefore $h_1 \in \ZZ f + \ZZ (v_1 \cdot g_1) + \dots + \ZZ(v_m \cdot g_m)$.
      By definition $v_i = \gcd\{\lambda \in \NN : \lambda
      \cdot f_i \in \module_f^{\ell-1}(\Phi)\}$ and thus $v_i \div a_i$.
      So, $h \in \ZZ f + \ZZ (v_1 \cdot g_1) + \dots + \ZZ(v_m \cdot g_m)$.
      \qedhere
  \end{description}
\end{proof}


\begin{lemma}
  \label{lemma:Mellf3}
  \tabto{2.3cm} {(\textlabel{A}{lemma:Mellf3A})} \ For
  every $\ell \in \NN$, $\module_f^{\ell} \subseteq
  \module_f^{\ell+1} \subseteq \module_f(\Phi)$.

  \vspace{2pt}
  \tabto{2.3cm} {(\textlabel{B}{lemma:Mellf3B})} \
  There is $\ell \in \NN$ such that $\module_f^\ell(\Phi) =
  \module_f^{\ell+1}(\Phi)$.

  \vspace{2pt}
  \tabto{2.3cm} {(\textlabel{C}{lemma:Mellf3C})} \ For
  every $\ell \in \NN$, if $\module_f^\ell(\Phi) =
  \module_f^{\ell+1}(\Phi)$ then $\module_f^\ell(\Phi) =
  \module_f(\Phi)$.
\end{lemma}

\begin{proof}
  \inproof{\eqref{lemma:Mellf3A}} By definition,
  $\module_f^{\ell} \subseteq \module_f^{\ell+1}$. 
  An
  induction on~$\ell \in \NN$ shows $M^\ell_f(\Phi)
  \subseteq \module_f(\Phi)$:
  \begin{description}
    \item[base case $\ell = 0$:]
     By~definition of $\module_f^\ell(\Phi)$ and of divisibility module, $M^0_f(\Phi) = \ZZ f
     \subseteq \module_f(\Phi)$.
    \item[induction case $\ell \geq 1$:]
      From the induction hypothesis,
      $M^{\ell-1}_f(\Phi) \subseteq \module_f(\Phi)$. 
      By definition, 
      $M^\ell_f(\Phi)$ is defined from
      $M^{\ell-1}_f(\Phi)$ by taking linear
      combinations of elements in $M^{\ell-1}_f(\Phi)$
      together with elements $b \cdot h$ such that $b
      \cdot g \in M^{\ell-1}_f(\Phi)$ and $g \div h$ is
      a divisibility of $\Phi$. From the definition of
      divisibility module, $\module_f(\Phi)$ is closed under
      such combinations, since for every $b \cdot g \in
      \module_f(\Phi)$ and $g \div h$ divisibility of $\Phi$,
      $b \cdot h \in \module_f(\Phi)$ (see
      Property~\ref{divmod:prop3} in the def.~of
      divisibility module). From $M^{\ell-1}_f(\Phi)
      \subseteq \module_f(\Phi)$ we then conclude that
      $M^{\ell}_f(\Phi) \subseteq \module_f(\Phi)$.
  \end{description}

  \inproof{\eqref{lemma:Mellf3B}} This statement
  follows from~\Cref{lemma:Mellf2}. Indeed, for a given
  $\ell \in \NN$, consider the canonical
  representations $(v_1,\dots,v_m)$ and
  $(v_1',\dots,v_m')$ of $\module_f^\ell(\Phi)$ and
  $\module_f^{\ell+1}(\Phi)$, respectively.
  By~\Cref{lemma:Mellf2}, if $\module_f^\ell(\Phi) \neq
  \module_f^{\ell+1}(\Phi)$ then one of the following holds:
  \begin{enumerate}
    \item\label{Mellf3BI1} there is $i \in [1,m]$ such
    that $v_i = 0$ and $v_i' \neq 0$, or
    \item\label{Mellf3BI2} there is $i \in [1,m]$ such
    that $v_i \neq 0$, $v_i' \neq v_i$ and $v_i'$
    divides $v_i$.
  \end{enumerate}
  Again from~\Cref{lemma:Mellf2}, for
  every $j \in [1,m]$, if $v_j \neq 0$ then $v_j'$
  divides $v_j$. This implies that both
  Items~\eqref{Mellf3BI1} and~\eqref{Mellf3BI2} cannot
  occur infinitely often, and therefore $\module_f^{r}(\Phi)
  = \module_f^{r+1}(\Phi)$ for some $r \in \NN$.\\[6pt]
  \inproof{\eqref{lemma:Mellf3C}}  
  From Part~\eqref{lemma:Mellf3A}, $\module_f^\ell(\Phi) \subseteq \module_f(\Phi)$. 
  We show that $\module_f^\ell(\Phi)$ satisfies the Properties~\mbox{\ref{divmod:prop1}--\ref{divmod:prop3}}
  of divisibility modules. Then, $\module_f(\Phi) \subseteq \module_f^\ell(\Phi)$
  follows from the minimality condition required by these modules.
  Properties~\ref{divmod:prop1}
  and~\ref{divmod:prop2} are trivially satisfied. To establish
  Property~\ref{divmod:prop3}, consider $b \cdot g
  \in \module_f^\ell(\Phi)$ and a divisibility $g \div h$ 
  of~$\Phi$. By~definition $b
  \cdot h \in \module_f^{\ell+1}(\Phi)$, and from $\module_f^{\ell}
  = \module_f^{\ell+1}(\Phi)$ we get $b \cdot h \in
  \module_f^{\ell}(\Phi)$. Therefore, $\module_f^{\ell}(\Phi)$
  satisfies Property~\ref{divmod:prop3}.
\end{proof}

In view of~\Cref{lemma:spanMellf,lemma:Mellf3}, the
algorithm required by~\Cref{lemma:module-span} presents
itself: it suffices to iteratively compute canonical
representations of every $\module_f^\ell(\Phi)$ until
reaching a fix-point. \Cref{procedure:span-divmod} performs
this computation.
\begin{algorithm}[t]
  \caption{Computes a set spanning a divisibility module}
  \label{procedure:span-divmod}
  \begin{algorithmic}[1]
      \Require A system of divisibility constraints $\Phi(\vec x) = \bigwedge_{i=1}^m f_i(\vec x) \div g_i(\vec x)$ and 
      a primitive polynomial $f$.
      \Ensure A tuple $(c_1,\dots,c_m) \in \NN^{m}$ such that $\{f,\,c_1 \cdot g_1,\,\dots,\,c_m \cdot g_m\}$ spans $\module_f(\Phi)$.
  \State\label{span-divmod:Vect}%
    $\vec v \coloneqq (0,\dots,0) \in \NN^{m}$
  \While{\textbf{true}}
    \label{span-divmod:While}
    \State\label{span-divmod:CopyT}%
    $\vec u \coloneqq \vec v$
    \For{$i$ \textbf{in} $[1,m]$}
      \label{span-divmod:For}
      \State\label{span-divmod:Psi}%
        $F_i \coloneqq \{-f_i,\,f,\,\vec u[1] \cdot g_1,\,\dots,\,\vec u[m] \cdot g_m \}$
      \State\label{span-divmod:Kernel}%
        $K_i \coloneqq$ basis of the integer kernel of the matrix representing $F_i$
      \State\label{span-divmod:GCD}%
        $\vec v[i] \gets \gcd(\text{row of $K_i$ corresponding to $-f_i$})$ 
    \EndFor
    \If{$\vec v = \vec u$}
        \label{span-divmod:If}%
        \textbf{return} $\vec v$
      \EndIf
  \EndWhile
\end{algorithmic}
\end{algorithm}
In a nutshell, during the $\ell$-th iteration ($\ell \geq 1$) of the \textbf{while} loop of
line~\ref{span-divmod:While}, the variable $\vec u$ contains the canonical representation of $\module_f^{\ell-1}(\Phi)$, and the algorithm updates the vector $\vec v$ with the canonical representation of $\module_f^{\ell}(\Phi)$. 
To update the value $\vec v[i]$ associated to~$g_i$ 
the algorithm needs to compute $\gcd\{\lambda \in \NN : \lambda \cdot f_i \in \module_f^{\ell-1}(\Phi) \}$ \mbox{(line~\ref{span-divmod:GCD})}.
This is done by finding a finite representation for all the scalars $\lambda$, which is given by those entries corresponding to $-f_i$ in a basis of the integer kernel of the matrix for the set $F_i$ defined in line~\ref{span-divmod:Psi}. 
As explained in~\Cref{subsec:divisibility:notion-elimination}, 
a set of polynomials $F \coloneqq \{h_1,\dots,h_\ell\}$ in variables $x_1 \incord \dots \incord x_d$ (where $\incord$ is an arbitrary order) can be represented as the matrix~$A \in \ZZ^{(d+1) \times \ell}$ in which each column $(a_d,\dots,a_1,c)$ contains the coefficients and the constant of a distinct element~$h$ of $F$, 
with $a_i$ being the coefficient of~$x_i$ for $i \in [1,d]$, and $c$ being the constant of~$h$. This matrix is unique up-to permutation of columns. 

It might not be clear for the moment whether \Cref{procedure:span-divmod}
runs in \ptime: in each iteration, the integer kernel computation
done in line~\ref{span-divmod:Kernel} might a priori
increase the bit length of the entries in the canonical
representation by a polynomial factor, 
yielding entries of exponential bit length
after polynomially many iterations -- an effect similar to
na\"ive implementations of Gaussian elimination or kernel
computations via suboptimal algorithms for the Hermite
normal form of a matrix. We show later that our worries are
unjustified, as the GCD~computed in
line~\ref{span-divmod:GCD} prevents this blow-up. For the
moment, let us formally argue on the correctness of~\Cref{procedure:span-divmod}.

\begin{proposition} 
  \label{prop:correctness-span-divmod}
  \Cref{procedure:span-divmod} respects its specification.
\end{proposition}

\begin{proof}
  We write $\vec u_\ell$ for the value that the tuple
  $\vec u$ declared in line~\ref{span-divmod:CopyT}
  of~\Cref{procedure:span-divmod} takes during the
  $(\ell+1)$-th iterations of the \textbf{while} loop of
  line~\ref{span-divmod:While}, with $\ell \in \NN$
  and assuming that the \textbf{while} loop is iterated
  at least $\ell+1$ times. We show the following claim: 
  \begin{claim}
    \label{claim:canonical-representation}
    For every $\ell \in \NN$, the tuple $\vec u_\ell$ is
    the canonical representation of~$\module_f^\ell(\Phi)$.
  \end{claim}
  Since~\Cref{procedure:span-divmod} terminates when
  $\vec u_{\ell-1}$ is found to be equal to $\vec
  u_{\ell}$ for some $\ell \geq 1$, its correctness
  follows directly from~\Cref{lemma:spanMellf} and~\Cref{lemma:Mellf3}.
  The proof of this claim is by induction on $\ell$.
  \begin{description}
    \item[base case.] 
    We have $\vec u_{0} = (0,\dots,0)
    \in \NN^{m}$, which is the canonical representation of $\module_f^0(\Phi)$. 
    \item[induction step.] 
      By induction hypothesis, let us assume that 
      $\vec u_{\ell} = (v_1,\dots,v_m)$ is the canonical representation of $\module_f^\ell(\Phi)$. 
      We show that when exiting the \textbf{for} loop of line~\ref{span-divmod:For}, for every $i \in [1,m]$, $\vec v[i]$ equals $v_i' \coloneqq \gcd\{\lambda \in \NN : \lambda \cdot f_i \in
      \module_f^{\ell}(\Phi)\}$.
      Thanks to the declaration of line~\ref{span-divmod:CopyT},
      this implies that $\vec u_{\ell+1}$ is the canonical representation of $\module_f^{\ell+1}(\Phi)$. 
      Since $\vec u_{\ell} = (v_1,\dots,v_m)$ is the canonical representation of $\module_f^\ell(\Phi)$, by~\Cref{lemma:spanMellf} we have $\module_f^\ell(\Phi) = \ZZ f + \ZZ (v_1 \cdot g_1) + \dots + \ZZ (v_m \cdot g_m)$.
      Therefore, $v_i' = \gcd\{\lambda \in \NN : \lambda \cdot f_i = \mu_0 \cdot f + \textstyle\sum_{i=1}^m \mu_i \cdot (v_i \cdot g_i) \text{ for some } \mu_0,\dots,\mu_m \in \ZZ\}$.
      The set of tuples $(\lambda,\mu_0,\dots,\mu_m) \in \ZZ^{m+2}$ such that~$\lambda \cdot f_i = \mu_0 \cdot f + \textstyle\sum_{i=1}^m \mu_i \cdot (v_i \cdot g_i)$
      corresponds to the solutions to the system of equations $A \cdot (\lambda,\mu_0,\dots,\mu_m) = \vec 0$ over the integers,
      where~$A$ is the matrix representing the set $\{-f_i,\,f,\, v_i \cdot g_1,\,\dots,\,\vec v_m \cdot g_m \}$, i.e., $F_i$ in
      line~\ref{span-divmod:Psi}.
      This set corresponds to $\ker(A)$, and so can be finitely represented with an integer kernel basis, i.e.,~$K_i$ in line~\ref{span-divmod:Kernel}. Computing $v_i'$ only requires to compute the 
      GCD~of the row of $K_i$ corresponding to the variable $\lambda$ of $-f_i$. This is exactly how $\vec v[i]$ is defined in line~\ref{span-divmod:GCD}.
      \qedhere
  \end{description}
\end{proof}

We move to the runtime analysis
of~\Cref{procedure:span-divmod}. We need the following lemma
studying the growth of the GCDs~of the rows of bases $K$ of $\ker(A)$
when columns of $A$ are
scaled by positive integers.
In the lemma below, $\diag(c_1,\dots,c_d)$ stands for the $d \times d$ diagonal matrix having $c_1,\dots,c_d$ in the main diagonal.





\begin{lemma}
  \label{lemma:gcd-kernel-lemma}
  Consider a matrix $A \in \ZZ^{n \times d}$ of rank
  $r$, integers $c_1,\dots,c_d > 0$, and let $K,K' \in
  \ZZ^{d \times (d-r)}$ be bases of the integer kernels of $A$
  and $A' \coloneqq A \cdot \diag(c_1,\dots,c_d)$, respectively. For
  every $i \in [1,d]$,
  \begin{enumerate}
    \item\label{gcd-kernel-lemma:itemA} if $\gcd(K[i])
    = 0$ then $\gcd(K'[i]) = 0$, and
    \item\label{gcd-kernel-lemma:itemB} if $\gcd(K[i])
    > 0$ then $\gcd(K'[i]) \neq 0$ and $\gcd(K'[i])$ divides
    $\lcm(c_1,\dots,c_d) \cdot \gcd(K[i])$.
  \end{enumerate}
\end{lemma}

\begin{proof}
  Note that $A'$ is the matrix
  obtained from $A$ by scaling the $j$-th column of $A$
  by $c_j$ ($j \in [1,d]$). 
  Let $i \in [1,d]$ and $(M,J) \in
  \{(A,K),(A',K')\}$. By definition of kernel, 
  ${\{ J \cdot
    \vec \lambda : \vec \lambda \in \ZZ^m \}}
    = \{ \vec x \in \ZZ^d : M \cdot \vec x = \vec 0 \}$. 
  This fact has three direct consequences:

  \begin{enumerate}[(A)]
    \item\label{gcd-kernel-lemma:item2} if $\gcd(J[i])
    = 0$, then no vector $\vec x = (x_1,\dots,x_d) \in
    \ZZ^d$ satisfies both $x_i \neq 0$ and $M \cdot
    \vec x = 0$, 
    \item\label{gcd-kernel-lemma:item3} if $\gcd(J[i])
    > 0$, then there is $\vec x = (x_1,\dots,x_d) \in
    \ZZ^d$ such that $x_i = \gcd(J[i])$ and $M \cdot
    \vec x = 0$, 
    \item\label{gcd-kernel-lemma:item4} if $\gcd(J[i])
    > 0$, then for every $\vec x = (x_1,\dots,x_d) \in
    \ZZ^d$ satisfying $M \cdot \vec x = 0$ we have
    $\gcd(J[i]) \div x_i$.
  \end{enumerate}
  \Cref{gcd-kernel-lemma:itemA,gcd-kernel-lemma:itemB} in the statement of the lemma are derived from these three properties.\\[6pt]
  \inproof{\eqref{gcd-kernel-lemma:itemA}} By
  contrapositive, assume that $\gcd(K'[i]) \neq 0$.
  Hence, $\gcd(K'[i]) > 0$ and by
  \Cref{gcd-kernel-lemma:item3} there is $\vec x
  = (x_1,\dots,x_d) \in \ZZ^d$ such that $x_i =
  \gcd(K'[i])$ and $A' \cdot \vec x = \vec 0$. Let
  $\vec y \coloneqq (c_1 \cdot x_1, \dots, c_d \cdot
  x_d)$. We have $A \cdot \vec y = A \cdot
  (\diag(c_1,\dots,c_d) \cdot \vec x) = (A \cdot
  \diag(c_1,\dots,c_d)) \cdot \vec x = A' \cdot \vec x
  = \vec 0$. Since $c_i > 0$ we have $c_i \cdot x_i
  \neq 0$, which together with $A \cdot \vec y = \vec
  0$ implies $\gcd(K[i]) \neq 0$ by
  \Cref{gcd-kernel-lemma:item2}.\\[6pt]
  \inproof{\eqref{gcd-kernel-lemma:itemB}} Suppose
  $\gcd(K[i]) > 0$. By
  \Cref{gcd-kernel-lemma:item3}, there is $\vec x
  = (x_1,\dots,x_d) \in \ZZ^d$ with $A \cdot \vec x = \vec 0$ and $x_i =
  \gcd(K[i])$. Define $C
  \coloneqq \lcm(c_1,\dots,c_d)$ and $\vec y \coloneqq
  (\frac{C}{c_1} \cdot x_1, \dots, \frac{C}{c_d} \cdot
  x_d)$. Note that $\vec y \in \ZZ^d$ is well-defined,
  since $c_1,\dots,c_d > 0$. Moreover,
  $\frac{C}{c_i} \cdot x_i = \frac{C}{c_i} \cdot
  \gcd(K[i]) > 0$. We have, 
  \[
    \begin{aligned}
      A' \cdot \vec y &\textstyle= A' \cdot
      (\diag(\frac{C}{c_1}, \dots, \frac{C}{c_d}) \cdot
      \vec x) \ = \ (A \cdot \diag(c_1,\dots,c_d)) \cdot
      (\diag(\frac{C}{c_1}, \dots, \frac{C}{c_d}) \cdot
      \vec x)\\ 
      &\textstyle = A \cdot (\diag(c_1,\dots,c_d) \cdot
      \diag(\frac{C}{c_1}, \dots, \frac{C}{c_d})) \cdot
      \vec x \ = \ C \cdot A \cdot \vec x \ = \ \vec 0.
    \end{aligned}
  \]

  \noindent
  Then, by
  \Cref{gcd-kernel-lemma:item2}, $\gcd(K'[i]) >
  0$, which in turn implies that $\gcd(K'[i]) \div \frac{C}{c_i} \cdot x_i$, directly from
  \Cref{gcd-kernel-lemma:item4}. Therefore,
  $\gcd(K'[i])$ divides $\lcm(c_1,\dots,c_d) \cdot
  \gcd(K[i])$.
\end{proof}

We are now ready to discuss the runtime
of~\Cref{procedure:span-divmod}.

\begin{proposition}
  \label{theorem:runtime-span-divmod}
  \Cref{procedure:span-divmod} runs in \ptime, and on an
  input $(\Phi,f)$ such that $\Phi = \bigwedge_{i=1}^m f_i \div g_i$ 
  it returns a vector $\vec v$ satisfying
  $\norminf{\vec v} \leq ((m+3) \cdot (\norminf{\Phi}+2))^{(m+3)^3}$.
\end{proposition}

\begin{proof}
  As done in the proof of~\Cref{prop:correctness-span-divmod},
  let $\vec u_\ell \in \ZZ^m$ be the value that the
  tuple $\vec u$ declared in
  line~\ref{span-divmod:CopyT} takes during the
  $(\ell+1)$-th iteration of the \textbf{while} of
  line~\ref{span-divmod:While}, with $\ell \in \NN$
  and assuming that the \textbf{while} loop is iterated
  at least $\ell+1$ times. Similarly, given $j \in
  [1,m]$, let $F_{\ell,j}$ and $K_{\ell,j}$ be
  the set of polynomial and matrix declared in
  lines~\ref{span-divmod:Psi}
  and~\ref{span-divmod:Kernel}, respectively, during
  the $(\ell+1)$-th iteration of the \textbf{while}
  loop and at the end of the iteration of the
  \textbf{for} loop of line~\ref{span-divmod:Psi}
  where the index variable $i$ takes value $j$. Lastly,
  following the code in line~\ref{span-divmod:GCD}, we
  define $v_{\ell,j} \coloneqq \gcd(\text{row of
  $K_{\ell,j}$ corresponding to $-f_j$})$. A few preliminary remarks that follow
  directly form the definitions above:

  For the runtime of the algorithm, first consider the
  case where $\module_f(\Phi) \cap \ZZ f_j = \{0\}$ for
  every $j \in [1,m]$, which implies $\module_f(\Phi) = \ZZ
  f$, by definition of divisibility module. Focus on
  the first execution of the body of the~\textbf{while}
  loop. Since $\vec u_0 = (0,\dots,0)$, for every $j
  \in [1,m]$, $F_{0,j} = \{-f_j,\,f\}$. Since $\module_f(\Phi) \cap \ZZ
  f_j = \{0\}$, the row of $K_{0,j}$
  corresponding to $-f_j$ contains only zeros. This implies $\vec v = (0,\dots,0) = \vec u_0$ in line~\ref{span-divmod:If}, and~\Cref{procedure:span-divmod}
  returns $(0,\dots,0)$ after a single
  iteration of the~\textbf{while}.
  
  Consider now the case where $\module_f(\Phi) \cap \ZZ f_j
  \neq \emptyset$ for some $j \in [1,m]$. Note that
  this implies $f_j = a \cdot f$ for some $a \in \ZZ
  \setminus \{0\}$ and $j \in [1,m]$, hence $\bitlength{f}
  \leq \poly{\bitlength{\Phi}}$. This allows us to bound the
  size of the output of~\Cref{procedure:span-divmod}
  in terms of $\Phi$, hiding factors that depend
  on~$f$ (as done in the statement of the
  proposition). A few auxiliary definitions are handy ($\ell \in \NN$ and $j \in [1,m]$):
  \begin{itemize}
    \item We associate to $\vec u_{\ell}$ the vector
    $\widehat{\vec u}_\ell \in \{0,1\}^{m}$ given by
    $\widehat{\vec u}_\ell[i] = 1$ iff $\vec u_\ell[i]
    \neq 0$, for every $i \in [1,m]$.
    \item We associate to $F_{\ell,j}$ the set $\widehat{F}_{\ell,j} \coloneqq \{-f_j,\,f,\,\widehat{\vec u}_{\ell}[1] \cdot g_1,\,\dots,\,\widehat{\vec u}_{\ell}[m] \cdot g_m \}$. 
    \item We associate to $K_{\ell,j}$ a basis $\widehat{K}_{\ell,j}$ for the integer kernel of the matrix representing $\widehat{F}_{\ell,j}$.
    \item We associate to $v_{\ell,j}$ the integer $\widehat{v}_{\ell,j} \coloneqq \gcd(\text{row of $\widehat{K}_{\ell,j}$ corresponding to $-f_j$})$.
  \end{itemize}
  In a nutshell, $\widehat{\vec u}_\ell$ ``forgets'' the magnitude of the integers stored in $\vec u_{\ell}$, keeping only whether their value was $0$ or not. The other objects defined above reflect this change at the level of matrices, kernels and GCDs.
  Up to permutation of columns, the matrix representing $F_{\ell,j}$ can be obtained by multiplying the matrix of $\widehat{F}_{\ell,j}$ by a diagonal matrix having in the main diagonal (a permutation of) $(1,1,\vec u_{\ell}[1],\dots, \vec u_{\ell}[m])$. 
  From the definition of $\widehat{K}_{\ell,j}$ and
  by~\Cref{lemma:gcd-kernel-lemma}, we conclude that 
  \begin{equation}
    \tag{$\dagger$}
    \label{divmod-runtime-pB} 
    \text{if $\widehat{v}_{\ell,j} = 0$
      then $v_{\ell,j} = 0$,
      \ 
      and \ if $\widehat{v}_{\ell,j} \neq 0$ then $v_{\ell,j} \neq 0$ and
      $v_{\ell,j}$ divides 
      $\lcm(\vec u_{\ell}) \cdot \widehat{v}_{\ell,j}$.} 
  \end{equation}

  Recall that the matrix representing $\widehat{F}_{\ell,j}$ has $d+1$ rows and $m+2$ columns.
  Since $\norminf{\widehat{F}_{\ell,j}} \leq \norminf{\Phi}$ for every $\ell \in \NN$ and $j \in [1,m]$, by~\Cref{corollary:kernel-gcd} 
  there an integer
  $N \in [2,\,{((m+3) \cdot (\norminf{\Phi}+2))^{(m+3)}}]$
  such that $N$ is greater 
  than~$\widehat{v}_{\ell,j}$, for every $\ell \in \NN$ and $j
  \in [1,m]$. We use~\eqref{divmod-runtime-pB} above to bound the
  number of iterations and magnitude of the entries
  of~$\vec u_{\ell}$ during the procedure. We show that
  \begin{enumerate}
    \item\label{span-divmod:totcor1} $\max_{\ell \in
    \NN}(\lcm(\vec u_\ell)) = \max_{\ell =
    0}^m(\lcm(\vec u_\ell)) \leq N^{m^3}$ and for every
    $j \in [1,m]$, $\vec u_m[j] \leq N^{m^2}$, and
    \item\label{span-divmod:totcor2} the
    \textbf{while} loop of
    line~\ref{span-divmod:While} is iterated at most
    $m^3 \cdot \log_2(N) + m$ many times.
  \end{enumerate}
  In Item~\eqref{span-divmod:totcor1} above we
  slightly abused our notation, as $\vec u_\ell$ is
  undefined for $\ell \in \NN$ greater or equal than
  the number of iterations of the \textbf{while} loop
  performed by the algorithm. In these cases, we
  postulate $\lcm(\vec u_\ell) = 0$ in order to make
  the equivalence in Item~\eqref{span-divmod:totcor1}
  well-defined. From the bound $N \leq
  ((m+3) \cdot (\norminf{\Phi}+2))^{(m+3)}$,
  Items~\eqref{span-divmod:totcor1}
  and~\eqref{span-divmod:totcor2} imply
  that~\Cref{procedure:span-divmod} runs in \ptime
  and outputs a vector~$\vec v$ with $\norminf{\vec v} \leq ((m+3) \cdot (\norminf{\Phi}+2))^{(m+3)^3}$;
  proving the proposition.\\[6pt]
  \inproof{\eqref{span-divmod:totcor1}} 
  Informally, Item~\eqref{span-divmod:totcor1} states
  that $\lcm(\vec u)$ is always bounded by $N^{m^2}$,
  and that $\lcm(\vec u)$ achieves its maximum at most
  after the first $m$ iterations of the \textbf{while}
  loop. We start by proving that~${\max_{\ell =
  0}^m(\lcm(\vec u_\ell)) \leq N^{m^3}}$ and that for
  every $j \in [1,m]$, $\vec u_m[j] \leq N^{m^2}$ This
  is done by induction on $\ell \in [1,m]$, by showing
  that (whenever defined) $\vec u_\ell$ is such that, for
  every $j \in [1,m]$, if ${\vec u_\ell[j] \neq 0}$ then
  ${\widehat{v}_{\ell-1,j}\neq 0}$ and $\vec
  u_\ell[j]$ divides $\big(\widehat{v}_{\ell-1,j}
  \cdot \prod_{i=0}^{\ell-2}
  \lcm(\widehat{v}_{i,1},\dots,\widehat{v}_{i,m})\big)$.
  Note that then $\vec u_\ell[j] \leq N^{m(\ell-1)+1}$,
  since $N$ is an upper bound on every
  $\widehat{v}_{\ell,j}$, and thus for $\ell = m$ we
  get $\vec u_m[j] \leq N^{m^2}$ and $\lcm(\vec u_m)
  \leq N^{m^3}$, as required. Below, let $\vec u_{\ell}
  = (c_1,\dots,c_m)$. Note that, from line~\ref{span-divmod:GCD} of the algorithm,
  if $\ell \geq
  1$, then $c_j = v_{\ell-1,j}$ for every $j \in [1,m]$.
  \begin{description}
    \item[base case $\ell = 1$.]
    From~$\vec u_0 =
    (0,\dots,0)$ we have $F_{0,j} = \widehat{F}_{0,j} = \{-f_j,f\}$ for every $j \in
    [1,m]$. This implies $\widehat{v}_{0,j} = v_{0,j}$.
    From $c_j = v_{0,j}$, we conclude that $c_j = \widehat{v}_{0,j}$, completing the base case.
    \item[induction step $\ell \geq 2$.]  
      Let  $j \in [1,m]$ such that $c_j \neq 0$. 
      From~\eqref{divmod-runtime-pB} and $c_j = v_{\ell-1,j}$, we get
      $\widehat{v}_{\ell-1,j} \neq 0$ and
      $c_j \div (\lcm(\vec u_{\ell-1}) \cdot
      \widehat{v}_{\ell-1,j})$. Let $\vec u_{\ell-1}
      = (c_1^*,\dots,c_m^*)$. From the induction
      hypothesis, for every~${k \in [1,m]}$, if $c_k^*
      \neq 0$ then $\widehat{v}_{\ell-2,k} \neq 0$
      and $c_k^* \div \big(\widehat{v}_{\ell-2,k} \cdot
      \prod_{i=0}^{\ell-3}
      \lcm(\widehat{v}_{i,1},\dots,\widehat{v}_{i,m})\big)$.
      Therefore,
      \[ 
        \lcm(\vec u_{\ell-1}) \ \div \ \lcm\big(
          (\widehat{v}_{\ell-2,1} \cdot \prod_{i=0}^{\ell-3} \lcm(\widehat{v}_{i,1},\dots,\widehat{v}_{i,m})), 
          \dots, 
          (\widehat{v}_{\ell-2,m} \cdot \prod_{i=0}^{\ell-3} \lcm(\widehat{v}_{i,1},\dots,\widehat{v}_{i,m}))
        \big).
      \]
      From the equivalence $\lcm(a \cdot b, c \cdot b) =
      \lcm(a,c) \cdot b$,
      the right-hand side of the divisibility above
      equals
      $\prod_{i=0}^{\ell-2}
      \lcm(\widehat{v}_{i,1},\dots,\widehat{v}_{i,m})$.
      Then, the fact that $c_j$ divides $\big(\widehat{v}_{\ell-1,j} \cdot
      \prod_{i=0}^{\ell-2}
      \lcm(\widehat{v}_{i,1},\dots,\widehat{v}_{i,m})\big)$
      follows
      directly from $c_j \div (\lcm(\vec u_{\ell-1})
      \cdot \widehat{v}_{\ell-1,j})$. 
  \end{description}

  To complete the proof of~\eqref{span-divmod:totcor1}, 
  we now show that $\max_{\ell \in \NN}(\lcm(\vec
  u_\ell)) = \max_{\ell = 0}^m(\lcm(\vec u_\ell))$.
  Directly from~\Cref{claim:canonical-representation}
  in the proof of~\Cref{prop:correctness-span-divmod},
  we have that for every $\ell \geq 1$, the vector
  $\vec u_\ell$ is the canonical representation of
  $\module_f^\ell(\Phi)$. We have,
  \begin{enumerate}[(A)]
    \item\label{theo:runntime-span-P1-itemA} for every
    $j \in [1,m]$, if $\vec u_\ell[j] \neq 0$ then
    $\vec u_{\ell+1}[j]$ divides $\vec u_\ell[j]$
    (assuming both $\vec u_{\ell}$ and $\vec
    u_{\ell+1}$ defined). 

    This follows directly from~\Cref{lemma:Mellf2}.

    \item\label{theo:runntime-span-P1-itemB} If $\vec
    u_\ell$, $\vec u_{\ell+1}$ and $\vec u_{\ell+2}$
    are defined, and $\vec u_\ell$ and $\vec
    u_{\ell+1}$ have the same zero entries, then 
    also~$\vec u_\ell$ and $\vec u_{\ell+2}$ have the same
    zero entries. 
    
    Indeed, in this case $\widehat{\vec u}_{\ell} = \widehat{\vec u}_{\ell+1}$ which implies $\widehat{v}_{\ell,j} = \widehat{v}_{\ell+1,j}$ for every $j \in [1,m]$. Now, if $\vec u_{\ell+2}[j] \neq 0$ then $v_{\ell+1,j} \neq 0$ and so $\widehat{v}_{\ell+1,j} \neq 0$ by~\eqref{divmod-runtime-pB}. Then $\widehat{v}_{\ell,j} \neq 0$, and again by~\eqref{divmod-runtime-pB} we get $v_{\ell,j} \neq 0$. 
    If instead $\vec u_{\ell+2}[j] = 0$, then $\vec u_{\ell}[j] = 0$ follows 
    from~\Cref{lemma:Mellf2}.
  \end{enumerate}
  Since $\vec u$ is a tuple with $m$ entries,
  Item~\eqref{theo:runntime-span-P1-itemB} above
  ensures that every $\vec u_\ell$ and $\vec u_{r}$
  with $\ell,r \geq m$ share the same zero entries.
  Item~\eqref{theo:runntime-span-P1-itemA} states
  instead that every non-zero entry of $\vec u_\ell$
  upper bounds the corresponding entry of $\vec u_{\ell+r}$, 
  for every $r \in
  \NN$, and that this latter entry is always non-zero.
  Together,
  Items~\eqref{theo:runntime-span-P1-itemA}
  and~\eqref{theo:runntime-span-P1-itemB}  
  imply that $\max_{\ell \in \NN}(\lcm(\vec u_\ell)) =
  \max_{\ell = 0}^m(\lcm(\vec u_\ell))$.\\[6pt]
  \inproof{\eqref{span-divmod:totcor2}}
  Assume that the \textbf{while} loop iterates at least $m+1$
  times (otherwise~\eqref{span-divmod:totcor2}
  is trivially satisfied). 
  From~\eqref{span-divmod:totcor2}, 
  the vector $\vec
  u_m$ such that $\vec u_m[j] \leq N^{m^2}$ for every
  $j \in [1,m]$. As we have just discussed above,
  by~Item~\eqref{theo:runntime-span-P1-itemB}, every
  subsequent $\vec u_{m+r}$ with $r \in \NN$ has the
  same zero entries as~$\vec u_m$. Whenever $\vec
  u_{m+r}$ and $\vec u_{m+r+1}$ are both defined
  (meaning in particular that $\vec u_{m+r} \neq \vec
  u_{m+r+1}$), there must be $j \in [1,m]$ such that
  $\vec u_{m+r}[j] \neq \vec u_{m+r+1}[j]$, and
  moreover by
  Item~\eqref{theo:runntime-span-P1-itemA}, $\vec
  u_{m+r+1}[i]$ divides $\vec u_{m+r}[i]$ for every $i
  \in [1,m]$, which in particular implies that $\vec
  u_{m+r+1}[j] \leq \frac{\vec u_{m+r}[j]}{2}$. 
  Therefore, the product of all non-zero entries of $\vec u$ (at least) halves at each
  iteration of the \textbf{while} loop after the $m$-th
  one. By~\eqref{span-divmod:totcor1}, for every $j \in [1,m]$ we have $\vec
  u_m[j] \leq N^{m^2}$, so the product of all
  non-zero entries in $\vec u_m$ is bounded by $N^{m^3}$. We
  conclude that the number of iterations of the
  \textbf{while} loop after the $m$-th one is bounded
  by $\log_2(N^{m^3}) = m^3 \cdot \log_2(N)$; i.e., $m^3
  \cdot \log_2(N) + m$ many iterations overall.
\end{proof}

\subsection{Closing a system of divisibility constraints under the elimination property}

\LemmaAddEliminationProperty*

\begin{proof}
  The algorithm is simple to state: 
  \begin{algorithmic}[1]
    \State\label{elim-algo:lineF} $F \coloneqq \{ f \text{ primitive } : \text{$a \cdot f$ is in the left-hand side of a divisibility of $\Phi$, for some $a \in \ZZ \setminus \{0\}$} \}$
    \For{$f \in F$}
      \State\label{elim-algo:lineSpan} $\vec v \coloneqq (c_1,\dots,c_m) \in \ZZ^m$ s.t.~$\{f,c_1 \cdot g_1,\dots, c_m \cdot g_m\}$ spans $\module_f(\Phi)$ 
      \Comment{\Cref{lemma:module-span}}
      \State\label{elim-algo:lineH} $H$ $\coloneqq$ HNF~of the matrix representing $\{f,c_1 \cdot g_1,\dots, c_m \cdot g_m\}$
      \Comment{\Cref{prop:HNF}}
      \State\label{elim-algo:linePurge} $\Phi$ $\gets$ $\Phi$ purged of all divisibilities of the form $f \div g$ for some polynomial $g$
      \For{$(a_d,\dots,a_1,a_0)$ non-zero column of $H$}
        \label{elim-algo:lineUpFor}  
        \State\label{elim-algo:lineUp} $\Phi$ $\gets$ $\Phi \land (f \div a_d \cdot x_d + \dots + a_1 \cdot x_1 + a_0)$
      \EndFor
    \EndFor
    \State \textbf{return} $\Phi$
  \end{algorithmic}

  Below, let $\Psi$ be the system returned by the algorithm on input $\Phi$.

  The fact that $\Psi$ has the elimination property follows from properties of
  the Hermite normal form. Consider~$F$ defined as in
  line~\ref{elim-algo:lineF}, and $f \in F$. Starting from the matrix~$A \in
  \ZZ^{(d+1) \times (m+1)}$ representing the spanning set~$S \coloneqq \{f,c_1 \cdot g_1,\dots, c_m \cdot g_m\}$ computed in~line~\ref{elim-algo:lineSpan}, by
  \Cref{propositoin:HNF-lattice} we conclude that $H$ in
  line~\ref{elim-algo:lineH} spans~$\module_f(\Phi)$. 
  Moreover, by properties of the HNF, all non-zero columns of $H$ are linearly independent, hence the \textbf{for} loop in line~\ref{elim-algo:lineUpFor} is adding divisibilities~$f \div h_1,\dots, f \div h_\ell$ where $h_1,\dots,h_\ell$ is a basis of~$\module_f(\Phi)$; and~$\ell \leq m+1$. Note that line~\ref{elim-algo:linePurge} has previously removed all divisibilities of the form $f \div g$. Hence, in $\Psi$ only the divisibilities $f \div h_1,\dots, f \div h_\ell$ have $f$ as a left-hand side. Recall now that each column $(a_d,\dots,a_1,c)$ of the matrix~$A$ contains the coefficients and the constant of a distinct element~$h \in S$, 
  with $a_i$ being the coefficient of~$x_i$ for $i \in [1,d]$, and $c$ being the constant of~$h$. Again since $H$ is in HNF, it is lower triangular, and the pivot of each non-zero column is strictly below the pivot of the column before it. Following the order $x_1 \incord \dots \incord x_d$, this allows us to conclude that, for every $k \in [0,d]$, the family $\{ g_1,\dots,g_{j} \} \coloneqq \{ g : \lv(g) \incordeq x_k \text{ and } f \div g \text{ appears in } \Psi \}$ is such that $g_1,\dots,g_j$ are linearly independent polynomials forming a basis for $\module_f(\Phi) \cap \ZZ[x_1,\dots,x_k]$; i.e., $\Psi$ has the elimination property. We also note that, by virtue of the updates done in~\ref{elim-algo:lineUp}, \Cref{lemma:add-elim-property:item-1,lemma:add-elim-property:item-2} in the statement of~\Cref{lemma:add-elimination-property} directly follow.

  The fact that $\Psi$ and $\Phi$ are equivalent both over~$\ZZ$ and for solutions modulo a prime follows from~\Cref{lemma:add-elim-property:item-1,lemma:add-elim-property:item-2} together with the following property of divisibility modules:
  given a system of divisibility constraints $\Phi'$ and a primitive term $f$,
  \begin{itemize}
    \item for every $\vec a$ integer solution of $\Phi'$ and every $g \in \module_f(\Phi')$, $f(\vec a)$ divides $g(\vec a)$, 
    \item for every $p \in \PP$, $\vec b$ solution of $\Phi'$ modulo $p$ and every $g \in \module_f(\Phi')$, $v_p(f(\vec b)) \leq v_p(g(\vec b))$. 
    
    Here, note that given polynomials $g_1$ and $g_2$ with $v_p(f(\vec b)) \leq v_p(g_1(\vec b))$ and ${v_p(f(\vec b)) \leq v_p(g_2(\vec b))}$ we have $v_p(f(\vec b)) \leq v_p(a_1 \cdot g_1(\vec b)+ a_2 \cdot g_2(\vec b))$ for every $a_1,a_2 \in \ZZ$, as the $p$-adic evaluation satisfies $v_p(x \cdot y) = v_p(x) + v_p(y)$ and $\min(v_p(x),v_p(y)) \leq v_p(x+y)$, for all $x,y \in \ZZ$.
  \end{itemize}

  Let us now move to the bounds on $\Psi$ stated in~\Cref{lemma:add-elim-property:item-3}.
  Directly from $\card{F} \leq m$ and the fact that $H$ is lower triangular we conclude that at most $m \cdot (d+1)$ divisibilities are added, and so $\Psi$ has at most $m \cdot (d+2)$ divisibilities.
  We analyze the norm of $\Psi$. 
  It suffices to consider a single $f \in F$. By definition, $\norminf{f} \leq \norminf{\Phi}$, and
  from~\Cref{lemma:module-span}, the infinity norm of the matrix $A$ representing $\{f,c_1 \cdot g_1,\dots, c_m \cdot g_m\}$ is bounded by 
  $((m+3) \cdot (\norminf{\Phi}+2))^{(m+3)^3} \cdot \norminf{\Phi}$.
  Note that $A$ has $d+1$ many rows.
  By~\Cref{prop:HNF}, the matrix $H$ in line~\ref{elim-algo:lineH} is such that 
  \begin{align*}
    \norminf{H} 
    &\leq ((d+1) \cdot \norminf{A} + 1)^{\bigO{d}}\\
    &\leq \Big((d+1) \cdot \big(((m+3) \cdot (\norminf{\Phi}+2))^{(m+3)^3} \cdot \norminf{\Phi}\big) + 1\Big)^{\bigO{d}}\\
    &\leq (d+1)^{\bigO{d}}(m + \norminf{\Phi} + 2)^{\bigO{m^3  d}}.
  \end{align*}
  From the updates done in line~\ref{elim-algo:lineUp}, we conclude that $\norminf{\Psi} \leq (d+1)^{\bigO{d}}(m + \norminf{\Phi} + 2)^{\bigO{m^3 d}}$.
\end{proof}

\LemmaSubstitAndElim*

\begin{proof}
  Let $f$ be a primitive polynomial. By definition of
  divisibility module, the lemma is true as soon as we prove
  \begin{enumerate*}[(i)]
    \item\label{substit-and-elim-prop:1} $f \in
    \module_f(\Psi(\vec \nu(\vec x),\vec y))$,
    \item\label{substit-and-elim-prop:2}
    $\module_f(\Psi(\vec \nu(\vec x),\vec y))$ is a
    $\ZZ$-module, and
    \item\label{substit-and-elim-prop:3} for every
    divisibility $g' \div h'$ (with $g'$ non-zero) appearing
    in $\Phi(\vec \nu(\vec x),\vec y)$, if $b \cdot g' \in
    \module_f(\Psi(\vec \nu(\vec x),\vec y))$ for some $b
    \in \ZZ$, then $b \cdot h' \in \module_f(\Psi(\vec
    \nu(\vec x),\vec y))$. 
  \end{enumerate*}
    Indeed, by definition $\module_f(\Phi(\vec \nu(\vec
    x),\vec y))$ is the smallest set fulfilling these three
    properties, and therefore it must then be included in
    $\module_f(\Psi(\vec \nu(\vec x),\vec y))$.

    The first two properties trivially follow by definition
    of $\module_f(\Psi(\vec \nu(\vec x),\vec y))$, hence let
    us focus on Property~\eqref{substit-and-elim-prop:3}.
    Consider a divisibility $g' \div h'$ appearing in
    $\Phi(\vec \nu(\vec x), \vec y)$ and such that $b \cdot
    g' \in \module_f(\Psi(\vec \nu(\vec x),\vec y))$. By
    definition of $\Phi(\vec \nu(\vec x), \vec y)$, there is
    a divisibility $g \div h$ appearing in $\Phi$ such that
    ${(g \div h)\substitute{\vec \nu (\vec x)}{\vec x} = (g'
    \div h')}$. We split the proof depending on whether $g$
    is a primitive polynomial. 
    \begin{description}
      \item[$g$ is not a primitive polynomial.]
      By~\Cref{lemma:add-elim-property:item-1}
      in~\Cref{lemma:add-elimination-property} the
      divisibility $g \div h$ occurs in $\Psi$. So, $g' \div
      h'$ is in $\Psi(\vec \nu(\vec x), \vec y)$ and
      directly by definition of divisibility module, $b
      \cdot h' \in \module_f(\Psi(\vec \nu(\vec x), \vec
      y))$.

      \item[$g$ is a primitive polynomial.] Let
      $\widetilde{g}$ and $c' \in \ZZ \setminus \{0\}$ be
      such that $g' = c' \cdot \widetilde{g}$.
      By~\Cref{lemma:add-elim-property:item-2}
      in~\Cref{lemma:add-elimination-property}, since $g
      \div h$ appears in $\Phi$,  
        $h \in \module_g(\Psi)$. By the elimination property
        of $\Psi$, there are divisibilities $g \div h_1,
        \dots, g \div h_k$ such that $h = \lambda_1 \cdot
        h_1 + \dots + \lambda_k \cdot h_k$ for some
        $\lambda_1,\dots,\lambda_k \in \ZZ \setminus \{0\}$.
        Every divisibility $(g \div h_i)\substitute{\vec \nu
        (\vec x)}{\vec x}$ with $i \in [1,k]$ appears in
        $\Psi(\vec \nu(\vec x), \vec y)$. Since $g' = g(\vec
        \nu(\vec x), \vec y)$ and $b \cdot g' \in
        \module_f(\Psi(\vec \nu(\vec x), \vec y))$ we have
        $b \cdot h_i(\vec \nu(\vec x), \vec y) \in
        \module_{f}(\Psi(\vec \nu(\vec x), \vec y))$ for
        every $i \in [1,k]$. Note that $h' = h(\vec \nu(\vec
        x), \vec y) = \lambda_1 \cdot h_1(\vec \nu(\vec x),
        \vec y) + \dots + \lambda_k \cdot h_k(\vec \nu(\vec
        x), \vec y)$, and therefore since the divisibility
        module is a $\ZZ$-module, $b \cdot h' \in
        \module_f(\Psi(\vec \nu(\vec x), \vec y))$.
        \qedhere
    \end{description}
\end{proof}

%% file: appendix-bound-pzero.tex
\section{Bounding the number of difficult primes}
\label{appendix:bound-pzero}

In this appendix, we establish~\Cref{lemma:simple-primes,lemma:bound-sterms,lemma:bound-on-pzero}.

\LemmaSimplePrimes*

\begin{proof}
  We remark that $p$ not dividing any coefficients nor
  constants appearing in the left-hand sides of~$\Phi$
  implies that all the left-hand sides are non-zero. We show
  that the system of non-congruences defined by $f_i
  \not\equiv 0 \pmod p$  for every $i \in [1,m]$, admits a
  solution~$\vec b$. This solution can clearly be taken with
  entries in $[0,p-1]$. Furthermore, $v_p(f_i(\vec b)) = 0$
  and $f_i(\vec b) \neq 0$ for every $i \in [1,m]$, and
  therefore~$\vec b$ is a solution for $\Phi$ modulo $p$ no
  matter the values of $v_p(g_i(\vec b))$ ($i \in [1,m]$).

  Consider an arbitrary ordering $x_1 \incord \dots \incord
    x_d$ on the variables in $\vec x$. We construct $\vec b$
    by induction on $k \in [0,d]$. At the $k$-th step of the
    induction we deal with the linear terms $h$ having
    ${\lv(h) = x_k}$. Below, we write $F_0$ for the set of
    the left-hand sides in $\Phi$ that are constant
    polynomials, and $F_k$ with $k \in [1,d]$ for the set of
    the left-hand sides $f$ in $\Phi$ such that $\lv(f)
    \incordeq x_k$.

    \begin{description}
        \item[base case: $k = 0$.] 
        Every $f \in F_0$ is a non-zero integer. Then, $f
        \not\equiv 0 \pmod p$ directly follows from the
        hypothesis that $p$ does not divide any constant
        appearing in the left-hand sides of $\Phi$.
        \item[induction step: $k \geq 1$.] 
        From the induction hypothesis, there is $\vec
        b_{k-1} = (b_1,\dots,b_{k-1}) \in \ZZ^{k-1}$ such
        that for every $f \in F_{k-1}$, $f(\vec b_{k-1})
        \not\equiv 0 \pmod p$. We find a value $b_k$ for
        $x_k$ so that the following system of
        non-congruences is satisfied 
        \begin{align*}
          f(\vec b_{k-1},x_k)
          \not\equiv 0 \pmod p
          && f \in F_{k} \setminus F_{k-1}.
        \end{align*}
        Linear polynomials $f$ in $F_k \setminus F_{k-1}$
        are of the form $f(\vec x) = f'(x_1,\dots,x_{k-1}) +
        c_f \cdot x_k$. Since by hypothesis $p \nmid c_f$,
        we consider the multiplicative inverse $c_f^{-1}$ of
        $c_f$ modulo $p$, and rewrite the above system as
        $x_k \not\equiv -c_f^{-1} \cdot f'$ for every $f \in
        F_{k} \setminus F_{k-1}$. This system as a solution
        directly from the fact that $p > m \geq \card(F_k
        \setminus F_{k-1})$.
        \qedhere
    \end{description}
\end{proof}

Before proving~\Cref{lemma:bound-sterms,lemma:bound-on-pzero},
we need the following result on system of divisibility constraints with the elimination property, that will later be used also in the proof of~\Cref{claim:still-increasing}.

\begin{lemma}\label{lemma:elim-prop-lv} Let
  $\Phi(x_1,\dots,x_d)$ be a system of divisibility with
  the elimination property for the order $x_1 \incord
  \dots \incord x_d$. For every primitive term $f$ and $j
  \in [1,d]$, the set $F \coloneqq \{g : {(f
  \div g)} \text{ appears in } \Phi \}$ has at most one
  element with leading variable $x_j$.
\end{lemma}

\begin{proof}
  If $f$ does not appear in the left-hand side of a divisibility of $\Phi$,
  then $F = \emptyset$ and the lemma holds. Suppose $f$ in a left-hand side.
  For simplicity, let us define $x_0 \coloneqq \bot$. By
  definition, for every ${k \in [0,d]}$, the elimination
  property forces $\{ g_1,\dots,g_\ell \}
  \coloneqq \{ g :{ \lv(g) \incordeq x_k} \text{ and } {f
  \div g} \text{ appears in } \Phi \}$ to be such that
  $g_1,\dots,g_\ell$ are linearly independent
  polynomials forming a basis for $\module_f(\Phi) \cap
  \ZZ[x_1,\dots,x_k]$. Given $k \in [0,d]$, let us write
  $F_k \coloneqq \{g : \lv(g) \incordeq x_k
  \text{ and }(f \mid g) \text{ appear in }\Phi\}$. For $j
  \in [1,d]$, by the elimination property, $F_{j-1}$ and
  $F_j$ are sets of linearly independent vectors, that
  respectively generates $\module_f(\Phi) \cap
  \ZZ[x_1,\dots,x_{j-1}]$ and $\module_f(\Phi) \cap
  \ZZ[x_1,\dots,x_{j}]$. 
  To conclude the proof, we show by induction on $j$ that the set 
  $F_j$ has at most one element with leading variable $x_j$.
  \begin{description}
      \item[base case $j = 0$.] 
          In this case $F_0$ only contains constant polynomials (and might be empty, in that case it generates the subspace $\{0\}$). By elimination property, $F$ is a set of linearly independent vectors, hence $F_0$ contains at most one element.
      \item[induction step $j \geq 1$.]
          \emph{Ad absurdum}, suppose there are two distinct $g_1,g_2 \in F_j \setminus F_{j-1}$ such that $\lv(g_1) = \lv(g_2) = x_j$. By definition of $S$-polynomial, $S(g_1,g_2) \in \module_f(\Phi) \cap \ZZ[x_1,\dots,x_{j-1}]$. Since $F_{j-1}$ generates $\module_f(\Phi) \cap \ZZ[x_1,\dots,x_{j-1}]$, there is a sequence of integers $(\lambda_h)_{h \in F_{j-1}}$ such that 
          $\sum_{h \in F_{j-1}} \lambda_h \cdot h = S(g_1,g_2)$. 
          However, $F_{j-1} \cup \{g_1,g_2\} \subseteq F_j$ (by definition) and $F_j$ is a set of linearly independent vectors. Therefore, every $\lambda_h$ above must be $0$, and we obtain $S(g_1,g_2) = 0$, i.e., $g_1$ and $g_2$ are linearly dependent, in contradiction with $g_1,g_2 \in F_j$. 
          \qedhere
  \end{description}
\end{proof}

\LemmaBoundSterms*

\begin{proof}
  Consider a primitive term $f$. If $f$ is not a primitive part of any $f_i$, with $i \in [1,m]$,
  then $\sfterms{f}{\Phi} = \terms(\Phi)$ and so $\sfterms{f}{\Phi}$ is included in any $\sfterms{f'}{\Phi}$ where $f'$ is a primitive part of a left-hand side of $\Phi$. Hence, we can upper bound $\card{\sterms(\Phi)}$ and $\bitlength{\norminf{\Phi}}$ by only looking at these primitive parts.\\[6pt]
  \inproof{\ref{lemma:bound-sterms-one}}
    For $f$ primitive part of some polynomials in a left-hand side of $\Phi$, the elements of
    $\sfterms{f}{\Phi}$ have the form
        $S\big(g_k, S (g_{k-1}, \dots S(g_1, h))\big)$
    where $h \in \terms(\Phi)$ and $f
    \mid g_i$ is a divisibility in $\Phi$, for all $i \in [1,k]$. Moreover, each
    $g_i$ has the same leading variable as
    $h_i \coloneqq S(g_{i-1},S(g_{i-2}, \dots, S(g_1,h)))$.
    Since $\Phi$ has the elimination property, by \Cref{lemma:elim-prop-lv},  
    given $h_i$ there is at most one $g$ 
    such that $f \div g$ and $\lv(g) = \lv(h_i)$; that is~$g_i$.
    Therefore, each element of $\sfterms{f}{\Phi}$ 
    can be characterized 
    by a pair $(k,h)$ where $h \in \terms(\Phi)$ and $k \in [0,d+1]$, i.e., $\card{\sfterms{f}{\Phi}} \leq \card{\terms(\Phi)} \cdot (d+2) \leq 2 \cdot m \cdot (d+2)$, since $\card{\terms(\Phi)} \leq 2 \cdot m$.
    The number of $f$ to be considered is bounded by $m$, i.e., the number of left-hand sides, which means 
    $\card{\sterms(\Phi)} \leq 2 \cdot m^2 (d+2)$.\\[6pt]
    \inproof{\ref{lemma:bound-sterms-two}}
    Recall that $\maxbl{f}$ is the maximum bit length of a coefficient or  constant of a polynomial~$f$, and that $\maxbl{R} = \max_{f \in R}\maxbl{f}$ for a finite set $R$ of polynomials. By examinating the definition of $S$-polynomial, we get that for every $f$ and $g$,
    ${\maxbl{S(f,g)}} \leq {\maxbl{f}} + {\maxbl{g}} + 1$. Let $f$ be a primitive polynomial. As discussed in the proof of~\eqref{lemma:bound-sterms-one}, an element of $\sfterms{f}{\Phi}$ is of the form $S\big(g_k, S (g_{k-1}, \dots S(g_1, h))\big)$, where $h \in \terms(\Phi)$, $f \mid g_i$ is a divisibility in $\Phi$, for all $i \in [1,k]$, and $ k \leq d+1$. 
    Then, $\maxbl{S\big(g_k, S (g_{k-1}, \dots S(g_1, h))\big)} \leq \maxbl{h} + \big(\sum_{i=1}^{k} \maxbl{g_i}\big) + k$.
    We conclude that $\maxbl{\sterms(\Phi)} \leq (d+2) \cdot (\maxbl{\Phi}+1)$.
  \end{proof}

\LemmaBoundOnPzero*

\begin{proof}
  We first analyse $\log_2(\Pi \pdiff(\Phi))$.
  Recall that $\pdiff(\Phi)$ is the set of those primes $p$ such that 
  either (i) $p \leq m$ or (ii) $p$ divide a coefficient
  or a constant of a left-hand side of $\Phi$. 
  The product of the primes satisfying (i) is bounded by $m! \leq m^m$. 
  The product of the primes satisfying (ii) is bounded by the product of
  the coefficients or the constants in the left-hand sides of
  $\Phi$, which is at most
  $\norminf{\Phi}^{m \cdot (d+1)}$.
  From these two bounds,
  we obtain the bound on $\log_2(\Pi \pdiff(\Phi))$ stated in the lemma.

  Let us analyse~$\log_2(\Pi \pzero(\Phi))$.
  Without loss of generality, assume that the order
  $\incord$ is such that ${x_1 \incord \dots  \incord x_d}$.
  We consider the three conditions defining $\pzero(\Phi)$
  separately, and establish upper bounds for each of them.
  Recall that the number of primes dividing $n \in \ZZ$ is
  bounded by $\log_2(n)$, and that \Cref{lemma:bound-sterms}
  implies $\card{S(\sterms(\Phi))} \leq 8 \cdot m^4 
  (d+2)^2$ and $\maxbl{S(\sterms(\Phi))} \leq 2 \cdot (d+2)
  \cdot (\maxbl{\Phi}+1)+1$.

  \begin{description}
    \item[\ref{pzero:1}:] 
    Directly from the bounds above, the primes satisfying~\ref{pzero:1} 
    are at most ${8 \cdot m^4 (d+2)^2}$, and thus the $\log_2$ of their product is at most ${8 \cdot m^4 (d+2)^2 \log_2(8 \cdot m^4 (d+2)^2)}$, which is bounded by ${64 \cdot m^5 (d+2)^3}$.

    \item[\ref{pzero:2}:] The product of the primes dividing a coefficient or constant of a polynomial $f$ in $S(\sterms(\Phi))$ is bounded by the product of these coefficients and constants. There are at most $(d+1) \cdot \card{S(\sterms(\Phi))}$ such coefficients and constants. Therefore, the $\log_2$ of this product is bounded by $(d+1) \cdot \card{S(\sterms(\Phi))} \cdot \maxbl{S(\sterms(\Phi))}$, which is bounded by ${16 \cdot m^4 (d+2)^4 (\maxbl{\Phi}+2)}$.

    \item[\ref{pzero:3}:] 
    If $f$ is a primitive term such that $a \cdot f$ does not occur in the left-hand sides of $\Phi$, for any $a \in \ZZ \setminus \{0\}$, then $\sfterms{f}{\Phi} = \terms(\Phi)$ and $\module_f(\Phi) = \ZZ f$, and therefore $\lambda$, if it exists, equals to~$1$.
    Consider $f$ primitive such that $a \cdot f \in \terms(\Phi)$ appears on the left-hand side of a divisibility in $\Phi$, for some $a \in \ZZ \setminus \{0\}$, and consider 
    $g \in \sfterms{f}{\Phi}$. We first compute a bound on the minimal positive $\lambda$ such that $\lambda \cdot g \in \module_f(\Phi)$, if such a $\lambda$ exists. 
    Let $x_j \coloneqq \lv(g)$, with $j \in [0,d]$ and $x_0 \coloneqq \bot$.
    Consider the set $\{h_1,\dots,h_\ell\} \coloneqq \{h : \lv(h) \leq \lv(g) \text{ and } f \div h \text{ is in } \Phi \}$; where $\ell \leq m$. 
    From the elimination property, this set is a basis for $\module_f(\Phi) \cap \ZZ[x_1,\dots,x_j]$, and therefore $\lambda$ exists if and only if ${\ZZ g \cap \ZZ h_1 + \dots + \ZZ h_\ell \neq \{0\}}$.
    Then let $K$ be a basis for the kernel of the matrix representing the set $\{-g,h_1,\dots,h_{\ell}\}$. As observed in the context of~\Cref{procedure:span-divmod}, if $\lambda$ exists then it is the GCD~of the row of $K$ corresponding to $-g$. 
    From~\Cref{corollary:kernel-gcd},
    $\lambda \leq \big (m+3)^{m+3}\max(2,\norminf{\Phi})^{m+2}$.
    In the proof of~\Cref{lemma:bound-sterms} we have shown $\card{\sfterms{f}{\Phi}} \leq 2 \cdot m \cdot (d+2)$, hence the number of pairs $(f,g)$ to consider is bounded by~$2 \cdot m^2 \cdot (d+2)$.
    Similarly to \ref{pzero:2}, the product of the primes dividing all $\lambda$s is bounded by the product of these $\lambda$s, which is at most $\big( (m+3)^{m+3}\max(2,\norminf{\Phi})^{m+2}\big)^{2 \cdot m^2 \cdot (d+2)}$.
    Therefore, 
    the $\log_2$ of the product of the primes satisfying~\ref{pzero:3} is at most 
    $32 \cdot m^4 (d+2) \cdot (\maxbl{\Phi} + 1)$.
  \end{description}
  Summing up the bounds we have just obtained yield 
  the bound stated in the lemma.
\end{proof}

%% file: appendix-local-to-global.tex
\section{\Cref{theorem:local-to-global}: proofs
of~\Cref{claim:still-increasing}
and~\Cref{claim:new-primes-are-ok}}
\label{appendix:lemmas-local-to-global}

In this section, we prove~\Cref{claim:still-increasing}
and~\Cref{claim:new-primes-are-ok}, which are required to
establish~\Cref{theorem:local-to-global}. In the context of
this theorem, recall that $\Psi(\vec x, \vec y)$ is a
formula that is increasing for $(X_1
\incord \dots \incord X_r)$ and has the elimination property
for an order~$(\incord) \in (X_1
\incord \dots \incord X_r)$. Here, $\vec x =
(x_1,\dots,x_d)$ are the variables appearing in $X_1$,
ordered as $x_1 \incord \dots \incord x_d$, and $\vec y$ are
the variables appearing in $\bigcup_{j=2}^{r} X_j$. We also
have solutions~$\vec b_p$ for $\Psi$ modulo $p$, for
every~$p \in \pzero(\Psi)$, and we have inductively computed
a map $\vec \nu \colon X_1 \to \ZZ$ the following three
properties:
\begin{description}
    \item[\textlabel{IH1}{apx:th:l-t-g:internal:IH1}:] 
    For every $p \in \pzero(\Psi)$ and $x \in X_1$, \ $\vec
    \nu(x) \equiv b_{p,x} \pmod {p^{\mu_p + 1}}$, where
    $b_{p,x}$ is the entry of $\vec b_p$ corresponding to
    $x$, and $\mu_p \coloneqq \max \{ v_p(f(\vec b_p)) \in
    \NN : f \text{ is in the left-hand side of a
    divisibility of } \Psi \}$.
    \item[\textlabel{IH2}{apx:th:l-t-g:internal:IH2}:] 
    For every prime $p \notin \pzero(\Psi)$ and for every
    $h,h' \in \sterms(\Psi)$ with leading variable in $X_1$,
    if $S(h, h')$ is not identically zero, then $p$ does not
    divide both $h(\vec \nu (\vec x))$ and $h'(\vec \nu
    (\vec x))$. 
    \item[\textlabel{IH3}{apx:th:l-t-g:internal:IH3}:]
    $h(\vec \nu (\vec x)) \neq 0$ for every $h \in
    \sterms(\Psi)$ that is non-zero and with $\lv(h) \in
    X_1$.
\end{description}

\noindent
The formula $\Psi'(\vec y)$ considered
in~\Cref{claim:still-increasing}
and~\Cref{claim:new-primes-are-ok} is defined as $\Psi'
\coloneqq \Psi\substitute{\vec \nu(x)}{x : x \in X_1}$.

\ClaimStillIncreasing*

At first glance,~\Cref{claim:still-increasing} might appear
intuitively true: since the notion of $r$-increasing form is
mainly a property on sets $X_1 \incord \dots \incord X_r$ of
orders of variables, and during the proof
of~\Cref{theorem:local-to-global} we are inductively
handling the smallest set $X_1$, it might seem trivial that
instantiating the variables in $X_1$ preserve increasingness
for $X_2 \incord \dots \incord X_r$. However, in general,
this is not the case. 
To see this, we repropose the example given in~\Cref{sec:intro-local-global}.
Consider the system of
divisibility constraints~$\Psi$ in increasing form for the order $u
\incord v \incord x \incord y \incord z$ and with the
elimination property for that order: 
\begin{align*}
    v &\div u + x + y\\
    v &\div x\\ 
    y+2 & \div z+1\\
    v &\div z\,.
\end{align*}
From the first two divisibilities, we have $(u+y) \in
\module_v(\Psi)$; i.e., $(u-2)+(y+2) \in \module_v(\Psi)$.
Therefore, if $u$ were to be instantiated as $2$, the
resulting formula $\Psi'$ would satisfy $(y+2) \in
\module_v(\Psi')$ and hence $(z+1) \in \module_v(\Psi')$,
from the third divisibility. Then, $1 \in \module_v(\Psi')$
would follow from the last divisibility, violating the
constraints of the increasing form. Fortunately, due to the
definition of $\sfterms{f}{\Psi}$, $u = 2$ contradicts the
property~\eqref{apx:th:l-t-g:internal:IH3} kept during the
proof of~\Cref{theorem:local-to-global}, meaning that the
above issue does not occur in our setting. Indeed, note that
$S(y+2,u+x+y) = 2-u-x$ is in $\sfterms{v}{\Psi}$, and so is
$S(2-u-x,x) = 2-u$. Then,~\eqref{apx:th:l-t-g:internal:IH3}
forces $2-u \neq 0$, excluding $u = 2$ as a possible
solution. This observation is the key to
establish~\Cref{claim:still-increasing}.

Given a set $A$ of polynomials, an integer $a \in \ZZ$ and a
variable $x$ occurring in those polynomials, we define
$A\substitute{a}{x} \coloneqq \{ f(a,\vec y) : f(x,\vec y)
\in A\}$, that is the set obtained by partially evaluating
$x$ as $a$ in all polynomials in $A$. This notion is
extended to sequences of value-variable pairs as
$A\substitute{a_i}{x_i : i \in I}$.

\begin{proof}[Proof of~\Cref{claim:still-increasing}]
    To show the statement, we consider an order $\incord'$
    in $(X_1 \incord \dots \incord X_r)$. Note that any
    order $(X_2 \incord \dots \incord X_r)$ can be
    constructed from elements in $(X_1 \incord \dots
    \incord X_r)$ by simply forgetting $X_1$. Let $\vec y =
    (y_1,\dots, y_j)$, with $y_1 \incord' \dots \incord'
    y_j$, be the variables in $\bigcup_{i=2}^r X_i$. To
    simplify the presentation, we denote by $a',b',\dots$
    and $f',g',\dots$ integers and polynomials related to
    $\Psi'$, and by $a,b,\dots$ and $f,g,\dots$ integers and
    polynomials related to $\Psi$. By definition of
    increasing form, we need to establish that for every $k
    \in [1,j]$ and primitive polynomial $f'(\vec y)$ such
    that $a' \cdot f'$ appears in the left-hand side of a
    divisibility in $\Psi'$, for some $a' \in \ZZ \setminus
    \{0\}$, and $\lv(f') = y_k$, we have
    $\module_{f'}(\Psi') \cap \ZZ[y_1,\dots,y_k] = \ZZ f'$.
    By definition of $\Psi'$ and since $a' \cdot f'$ appears
    in a left-hand side, there is a primitive polynomial
    $f(\vec x, \vec y)$ and a scalar $a \in \ZZ \setminus
    \{0\}$ such that $a \cdot f$ is in the left-hand side of
    some divisibility in $\Psi$, and $a' \cdot f'(\vec y) =
    a \cdot f(\vec \nu(\vec x),\vec y)$. Note that this
    implies $a \div a'$ and $\lv(f) \not\in X_1$. We prove
    that $\frac{a'}{a} \cdot \module_{f'}(\Psi') \subseteq
    \module_{f}(\Psi)\substitute{\vec \nu(x)}{x : x \in
    X_1}$. Note that this inclusion implies $\Psi'$ in
    increasing form. To see this, take $g' \in
    \module_{f'}(\Psi') \cap \ZZ[y_1,\dots,y_k]$. We have
    $\frac{a'}{a} \cdot g' \in
    \module_{f}(\Psi)\substitute{\vec \nu(x)}{x : x \in
    X_1}$, and thus there is $g(\vec x, \vec y) \in
    \module_f(\Psi)$ such that $\frac{a'}{a} \cdot g' =
    g(\vec \nu(\vec x), \vec y)$. Since $\lv(g') \incord'
    y_k$, we have $\lv(g) \incord' y_k$. Since $\Psi$ is
    increasing for $\incord'$, we conclude that $g \in \ZZ
    f$. Note that $(\ZZ f)\substitute{\vec \nu(x)}{x \in
    X_1} \subseteq \ZZ f'$. Then $\frac{a'}{a} \cdot g' \in
    \ZZ f'$. Since $f'$ is primitive, we get $g' \in \ZZ
    f'$. This shows $\module_{f'}(\Psi') \cap
    \ZZ[y_1,\dots,y_k] \subseteq \ZZ f'$, and the other
    inclusion directly follows by definition of
    $\module_{f'}(\Psi')$. We conclude that $\Psi'$ is
    increasing. 

    To conclude the proof of~\Cref{claim:still-increasing},
    let us show that $\frac{a'}{a} \cdot \module_{f'}(\Psi')
    \subseteq \module_{f}(\Psi)\substitute{\vec \nu(x)}{x :
    x \in X_1}$. By definition of $\module_{f'}(\Psi')$,
    this follows as soon as we prove the following three
    properties:
    \begin{enumerate}[(A)]
        \item\label{still-incr-A} $\frac{a'}{a} \cdot f'$
        belongs to $\module_{f}(\Psi)\substitute{\vec
        \nu(x)}{x:x\in X_1}$,
        \item\label{still-incr-B}
        $\module_{f}(\Psi)\substitute{\vec \nu(x)}{x:x\in
        X_1}$ is a $\ZZ$-module, and 
        \item\label{still-incr-C} If $g' \div h'$ is a
        divisibility in $\Psi'$ and $b' \cdot g' \in
        \module_{f}(\Psi)\substitute{\vec \nu(x)}{x:x\in
        X_1}$ for some $b' \in \ZZ\setminus\{0\}$, then $b'
        \cdot h' \in  \module_{f}(\Psi)\substitute{\vec
        \nu(x)}{x:x\in X_1}$.
    \end{enumerate}
    By definition of divisibility module, $\frac{a'}{a}
    \cdot \module_{f'}(\Psi')$ is the smallest set that
    satisfies the three properties above, and therefore it
    must be included in $\module_{f}(\Psi)\substitute{\vec
    \nu(x)}{x:x\in X_1}$.\\[6pt]
    \inproof{\ref{still-incr-A}} By definition of $f$,\,
    $a' \cdot f' = a \cdot f(\vec \nu(\vec x), \vec y)$ and
    $a \div a'$, hence $\frac{a'}{a} \cdot f' = f(\vec
    \nu(\vec x), \vec y)$, and by definition of divisibility
    module $f(\vec \nu(\vec x), \vec y) \in
    \module_{f}(\Psi)\substitute{\vec \nu(x)}{x:x\in
    X_1}$.\\[6pt]
    \inproof{\ref{still-incr-B}} This follows directly
    from the definition of divisibility module being a
    $\ZZ$-module. Indeed, substitutions preserve the notion
    of $\ZZ$-module.\\[6pt]
    \inproof{\ref{still-incr-C}} This property follows
    from our definition of $\sfterms{f}{\Psi}$ together with
    the property~\eqref{apx:th:l-t-g:internal:IH3} and the
    fact that $\Psi$ has the elimination property for the
    order~$\incord$ (not to be confused with the
    order~$\incord'$, which does not guarantee the
    elimination property). Consider a divisibility $g'(\vec
    y) \div h'(\vec y)$ occurring in $\Psi'$ and $b' \in \ZZ
    \setminus \{0\}$ such that $b' \cdot g' \in
    \module_{f}(\Psi)\substitute{\vec \nu(x)}{ x : x \in
    X_1}$. By definition of $\Psi'$, there is a divisibility
    $g(\vec x, \vec y) \div h(\vec x, \vec y)$ in $\Psi$
    such that $g' = g(\vec \nu(\vec x), \vec y)$ and $h' =
    h(\vec \nu(\vec x), \vec y)$. Also, by definition of
    $\module_{f}(\Psi)\substitute{\vec \nu(x)}{ x : x \in
    X_1}$, there is a polynomial $\widehat{g}(\vec x, \vec
    y) \in \module_{f}(\Psi)$ such that $b' \cdot g' =
    \widehat{g}(\vec \nu(\vec x), \vec y)$.

    To conclude the proof, it suffices to show that $b'
    \cdot g = \widehat{g}$. Indeed, since $g \div h$ appears
    in $\Psi$ and $\widehat{g} \in \module_f(\Psi)$, we then
    get $b' \cdot h \in \module_f(\Psi)$ by the definition
    of divisibility module, which implies $b' \cdot h' \in
    \module_{f}(\Psi)\substitute{\vec \nu(x)}{x:x\in X_1}$
    by definition of $h$; concluding the proof.

    Since $\widehat{g} \in \module_f(\Psi)$ and $\Psi$ has
    the elimination property for~$\incord$, there are
    linearly independent polynomials $h_1,\dots,h_\ell$ such
    that the divisibilities $f \div h_i$ appear in $\Psi$
    and there are ${\lambda_1, \dots,\lambda_{\ell} \in \ZZ
    \setminus \{0\}}$ such that $\widehat{g} =
    \sum_{i=1}^\ell \lambda_i \cdot h_i$. Thanks
    to~\Cref{lemma:elim-prop-lv}, we can arrange these
    polynomials so that $\lv(h_1) \incord \dots \incord
    \lv(h_\ell)$. We write $c_i$ for the coefficient
    corresponding to the leading variable of~$h_i$. Since
    $\lv(f) \not\in X_1$ (stated earlier) and $\Psi$ is
    increasing, ${\lv(h_i) \in \bigcup_{k=2}^r X_k}$ holds
    for every $i \in [1,\ell]$. From $g' = g(\vec \nu(\vec
    x), \vec y)$ and $b' \cdot g' = \widehat{g}(\vec
    \nu(\vec x), \vec y)$ we directly get $b' \cdot g(\vec
    \nu(\vec x), \vec y) = \widehat{g}(\vec \nu(\vec x),
    \vec y)$. Therefore, $(b' \cdot g - \widehat{g})(\vec
    \nu(\vec x), \vec y) = 0$, implying that $b' \cdot g -
    \widehat{g}$ is either constant or has its leading
    variable in $X_1$. This implies that $b' \cdot g -
    \sum_{i = 1}^\ell \lambda_i \cdot h_i$ is either
    constant or has its leading variable in $X_1$. Since the
    $\lambda_i$ are non-zero, and moreover $\lv(h_i)$ is not
    in $X_1$ and $\lv(h_1) \incord \dots \incord
    \lv(h_\ell)$, we have $\lv(b' \cdot g - \sum_{i =
    k+1}^\ell \lambda_i \cdot h_i) = \lv(h_k)$ for every $k
    \in [1,\ell]$, and the coefficient corresponding to the
    leading variable of $b' \cdot g - \sum_{i
    =k+1}^{\ell}\lambda_i\cdot h_i$ is exactly $\lambda_k
    \cdot c_k$. 
    
    We show by induction on $k \in [1,\ell+1]$, with base
    case $k = \ell+1$, that $\alpha_{k} \cdot (b' \cdot g -
    \sum_{i=k}^\ell \lambda_i \cdot h_i) = b' \cdot
    S(g,h_\ell,\dots,h_k)$, where $\alpha_k \coloneqq
    \prod_{i=k}^\ell c_i$, and $S(f_1,\dots,f_n)$ is short
    for $S( \dots (S(f_1,f_2),\dots),f_n)$; e.g.,
    $S(f_1,f_2,f_3) = S(S(f_1,f_2),f_3)$.

    \begin{description}
        \item[base case $k = \ell+1$:] For the base case,
        $\alpha_{\ell+1} = 1$ and the equivalence becomes
        $b' \cdot g = b' \cdot g$. 
        \item[induction step $k \leq \ell$:]
        we have ${\alpha_{k+1} (b' \cdot g -
        \sum_{i=k+1}^\ell \lambda_i \cdot h_i) = b' \cdot
        S(g,h_\ell,\dots,h_{k+1})}$ by induction hypothesis.
        Note that, from the left-hand side of this equation,
        the coefficient corresponding to the leading
        variable of $b' \cdot S(g,h_\ell,\dots,h_{k+1})$ is
        $c_k \cdot \alpha_{k+1} \cdot \lambda_k$. Then,
        {\allowdisplaybreaks
        \begin{flalign*}
            & \alpha_k \cdot (b' \cdot g - \sum_{i=k}^\ell
            \lambda_i \cdot h_i)\\ 
            =\ & c_k \cdot \alpha_{k+1} (b' \cdot g -
            \sum_{i=k}^\ell \lambda_i \cdot h_i)
            &\text{definition of}~\alpha_{k}\\
            =\ & c_k \cdot \alpha_{k+1} (b' \cdot g -
            \sum_{i=k+1}^\ell \lambda_i \cdot h_i) - c_k
            \cdot \alpha_{k+1} \cdot \lambda_{k} \cdot
            h_{k}\\
            =\,& c_k \cdot (b' \cdot
            S(g,h_\ell,\dots,h_{k+1})) - (c_k \cdot
            \alpha_{k+1} \cdot \lambda_{k}) \cdot h_{k}
            &\text{induction hypothesis}\\
            =\,& S(b' \cdot S(g,h_\ell,\dots,h_{k+1}), h_k)
            &\text{coeff.~leading var.}~h_k~\text{is}~c_k\\
            &&\hspace{-3cm}\text{coeff.~leading var.}~(b'
            \cdot
            S(g,h_\ell,\dots,h_{k+1}))~\text{is}~ c_k \cdot
            \alpha_{k+1} \cdot \lambda_k\\
            =\,& b' \cdot S(g,h_{\ell},\dots,h_k)
            &\hspace{-3cm}S(b' \cdot f_1,f_2) = b' \cdot
            S(f_1,f_2), \text{by definition of }
            S\text{-polynomial}.
        \end{flalign*}
        }
    \end{description}
    Thanks to the equality $\alpha_{k} \cdot (b' \cdot g -
    \sum_{i=k}^\ell \lambda_i \cdot h_i) = b' \cdot
    S(g,h_\ell,\dots,h_k)$ we just established, we conclude
    that $\alpha_{1} \cdot (b' \cdot g - \widehat{g}) = b'
    \cdot S(g,h_\ell,\dots,h_1)$. Moreover, from $\lv(b'
    \cdot g - \sum_{i = k+1}^\ell \lambda_i \cdot h_i) =
    \lv(h_k)$ we conclude that
    $\lv(S(g,h_\ell,\dots,h_{k+1})) = \lv(h_{k})$, for every
    $k \in [1,\ell]$. Then, since ${g \in \terms(\Psi)}$ and
    the divisibilities ${f \div h_1},\, \dots,\, {f \div
    h_\ell}$ appear in $\Psi$, by definition of
    $\sfterms{f}{\Psi}$, we conclude that
    ${S(g,h_\ell,\dots,h_1) \in \sfterms{f}{\Psi}}$. Recall
    that $b' \cdot g - \widehat{g}$ is either constant or
    has its leading variable in $X_1$. The same is true for
    $S(g,h_\ell,\dots,h_1)$, and we have $(\alpha_{1} \cdot
    (b' \cdot g - \widehat{g}))(\vec \nu(\vec x)) = b' \cdot
    S(g,h_\ell,\dots,h_1)(\vec \nu(\vec x))$. From $(b'
    \cdot g - \widehat{g})(\vec \nu(\vec x)) = (b' \cdot g -
    \widehat{g})(\vec \nu(\vec x),\vec y) = 0$ and $b' \neq
    0$ we get ${S(g,h_\ell,\dots,h_1)(\vec \nu(\vec x)) =
    0}$. From the
    property~\eqref{apx:th:l-t-g:internal:IH3}, this can
    only occur when $S(g,h_\ell,\dots,h_1) = 0$, and so
    $\alpha_1 \cdot (b' \cdot g - \widehat{g}) = 0$. By
    definition $\alpha_1 \neq 0$, and therefore $b' \cdot g
    = \widehat{g}$, concluding the proof
    of~\eqref{still-incr-C}.
  \end{proof}
  
\ClaimNewPrimesAreOk* 

\begin{proof}
The first statement of the claim follows
from~\eqref{apx:th:l-t-g:internal:IH1} and the definition of
$\mu_p$ (the reasoning is analogous to the one in the base
case $r=1$ of the induction
of~\Cref{theorem:local-to-global}). For the second
statement, consider a prime $p$ not belonging
to~$\pzero(\Psi)$. We provide a solution~$\vec b_p$ for
$\Psi'(\vec y)$ modulo~$p$. Let $\vec y = (y_1,\dots,y_j)$
with $y_1 \incord \dots \incord y_j$. To compute $\vec b_p =
(b_{p,1},\dots,b_{p,j})$, where $b_{p,k}$ is the value
assigned to~$y_k$, we consider two cases that depend on
whether $p$ divides some~$\alpha_i$ appearing in the first
block of divisibilities
of~\Cref{eq:l-t-g:outer-induction-psi} (i.e., these are
the~$\alpha_i$~with~$i \in [\ell+1,n]$).

\begin{description}
  \item[case {$p \nmid \alpha_i$} for all {$i \in
    [\ell+1,n]$}.] This case is relatively simple. Starting
    from $y_1$ and proceeding in increasing order of
    variables, we compute $b_{p,k+1}$ ($k \in \NN$) by
    solving the system
    \begin{align}
      \label{eq:case-p-alphai}
      h(b_{p,1},\dots,b_{p,k},\,y_{k+1}) &\not\equiv 0 && \pmod p &
      h \in \terms(\Psi') \text{ s.t.}~\lv(h) = y_{k+1}.
    \end{align}
    With respect to the $h$ above, let us write
    $h(b_{p,1},\dots,b_{p,k}\,y_{k+1}) = c_h + a_h \cdot
    y_{k+1}$ where $c_h$ is the constant term obtained by
    partially evaluating $h$ with respect to
    $(b_{p,1},\dots,b_{p,k})$ and $a_h$ is the coefficient
    of $y_{k+1}$ in $h$. By definition of $\Psi'$, the term
    $h$ is obtained by substituting $\vec x$ for $\vec
    \nu(\vec x)$ in a polynomial of $\Psi$, and in that
    polynomial~$y_{k+1}$ has coefficient $a_h$. Since $p
    \not\in \pzero(\Psi)$, from Condition~\ref{pzero:2} we
    conclude that $p \nmid a_h$, and so $a_h$ has an inverse
    $a_h^{-1}$ modulo $p$. The system of non-congruences
    above is equivalent to the system $\mcS_{k+1}$ given by
    \begin{align*}
      y_{k+1} &\not\equiv -a_h^{-1} \cdot c_h && \pmod p &
      h \in \terms(\Psi') \text{ s.t.}~\lv(h) = y_{k+1}.
    \end{align*}
    From Condition~\ref{pzero:1} we have $p >
    \card{\terms(\Psi)} \geq \card{\terms(\Psi')}$, and so
    it suffices to take~$b_{p,k+1}$ to be an element in
    $[0,p-1]$ that differs from every $-a_h^{-1} \cdot c_h$
    appearing in the rows of~$\mcS_{k+1}$.

    The solution $\vec b_p$ resulting from the systems of
    non-congruences~$\mcS_{1},\dots,\mcS_j$ is such that,
    for every $h \in \terms(\Psi')$, $v_p(h(\vec b_p)) = 0$.
    Therefore, $\vec b_p$ is a solution for $\Psi'$ modulo
    $p$.

  \item[case {$p \div \alpha_i$} for some {$i \in
  [\ell+1,n]$}.] This case is  involved. Since $p$ divides
  some $\alpha_i = f_i(\vec \nu(\vec x))$, and $p \not \in
  \pzero(\Psi)$, by Condition~\ref{pzero:2} we have $p
  \div f(\vec \nu(\vec x))$, where $f$ is the primitive
  polynomial obtained by dividing every coefficient and
  constant of $f_i$ by $\gcd(f_i)$. 
  Recall that $\vec x = (x_1,\dots,x_d)$ with $x_1 \incord \dots \incord x_d \incord y_1 \incord \dots \incord y_j$, 
  and note that $\lv(f)
  \incordeq x_d$. Below, let us define $u \coloneqq v_p(f(\vec
  \nu(\vec x)))$. The idea is to use $f$ to iteratively
  construct the solution $\vec b_p$ for $\vec y =
  (y_1,\dots,y_j)$. We rely on the following induction
  hypotheses ($k \in [0,j]$):
  \begin{description}
    \item[\textlabel{IH1$^\prime$}{th:l-t-g:internal:IH1p}:] 
      \begin{minipage}[t]{0.95\linewidth}
        for every non-zero polynomial $g(\vec
        x,y_1,\dots,y_t) \in \terms(\Psi)$ such that $t \leq
        k$,\\
        \hphantom{.}\quad if $\ZZ g \cap \module_f(\Psi)
        \neq \{0\}$ then $v_p(g(\vec \nu(\vec
        x),b_{p,1},\dots,b_{p,t})) = u$, and 
      \end{minipage}

      \item[\textlabel{IH2$^\prime$}{th:l-t-g:internal:IH2p}:] 
      \begin{minipage}[t]{0.95\linewidth}
        for every non-zero polynomial $h(\vec
        x,y_1,\dots,y_t) \in \sfterms{f}{\Psi}$ such that $t
        \leq k$,\\ 
        \hphantom{.}\quad if $\ZZ h \cap \module_f(\Psi) =
        \{0\}$ then $v_p(h(\vec \nu(\vec
        x),b_{p,1},\dots,b_{p,t})) = 0$.
      \end{minipage}
  \end{description}

  Let us first show that by constructing $\vec b_p$ so that
  it satisfies the hypotheses above for $k = j$ implies that
  $\vec b_p$ is a solution for $\Psi'$ modulo $p$. Consider
  a divisibility $\alpha_i + f_i'(\vec  y) \mid \beta_i  +
  g_i'(\vec  y)$ in~$\Psi'$, with $i \in [\ell+1,m]$ and
  $f_i' = 0$ if $i \leq n$. Recall that $\alpha_i = f_i(\vec
  \nu(\vec x))$ and $\beta_i = g_i(\vec \nu(\vec x))$, and
  given $h \coloneqq f_i + f_i'$ and $h' \coloneqq g_i +
  g_i'$, the divisibility $h \div h'$ occurs~in $\Psi$. We
  have two cases:
  \begin{itemize}
    \item $\ZZ h \cap \module_f(\Psi) \neq \{0\}$. In this
    case, by definition of $\module_f(\Psi)$ we have $\ZZ h'
    \cap \module_f(\Psi) \neq \{0\}$. According
    to~\eqref{th:l-t-g:internal:IH1p}, $v_p(h(\vec \nu(\vec
    x), \vec b_p)) = v_p(h'(\vec \nu(\vec x), \vec b_p)) =
    u$. By definition of $h$ and $h'$, we get $v_p(\alpha_i
    + f_i'(\vec b_p)) = v_p(\beta_i  + g_i'(\vec b_p)) = u$.
    Note that $f(\vec \nu(\vec x))$ is non-zero
    by~\eqref{apx:th:l-t-g:internal:IH3}, hence its $p$-adic
    evaluation $u$ belongs to $\NN$, which forces $\alpha_i
    + f_i'(\vec b_p)$ to be non-zero.
    \item $\ZZ h \cap \module_f(\Psi) = \{0\}$. Recall that
    $\terms(\Psi) \subseteq \sfterms{f}{\Psi}$, by
    definition. Hence, directly from
    \eqref{th:l-t-g:internal:IH2p} we get $v_p(h(\vec
    \nu(\vec x), \vec b_p)) = v_p(\alpha_i + f_i'(\vec b_p))
    = 0$. This implies $\alpha_i + f_i'(\vec b_p)$ non-zero,
    and moreover $v_p(\alpha_i + f_i'(\vec  b_p)) \leq
    v_p(\beta_i  + g_i'(\vec  b_p))$ no matter what is the
    value of $v_p(\beta_i  + g_i'(\vec  b_p))$.
  \end{itemize}
  Note moreover that~\eqref{th:l-t-g:internal:IH1p}
  and~\eqref{th:l-t-g:internal:IH2p} directly imply~$\max \{
  v_p(g(\vec b_p)) \in \NN : g \in \terms(\Psi') \} \leq u$. 

  To conclude the proof, we show how to construct $\vec b_p$
  satisfying~\eqref{th:l-t-g:internal:IH1p}
  and~\eqref{th:l-t-g:internal:IH2p}. 

  \begin{description}
    \item[base case $k = 0$.] 
      We establish~\eqref{th:l-t-g:internal:IH1p}
      and~\eqref{th:l-t-g:internal:IH2p} for polynomials
      with variables in~$\vec x$, by showing the three
      properties below, for every non-zero polynomial $h \in
      \sterms(\Psi)$ with~${\lv(h) \incordeq x_{d}}$.
      \begin{enumerate}[(A)]
        \item\label{p-adic-zero} Either~${\ZZ f \cap \ZZ h
        \neq \{0\}}$ or~$p \nmid h(\vec \nu(\vec x))$.
        \item\label{p-adic-A} If $\ZZ f \cap \ZZ h \neq
        \{0\}$, then $v_p(h(\vec \nu(\vec x))) = v_p(f(\vec
        \nu(\vec x)))$.
       
        \item\label{p-adic-B} If $p \nmid h(\vec \nu(\vec
        x))$ then $v_p(h(\vec \nu(\vec x))) = 0$ and $\ZZ h
        \cap \module_f(\Psi) = \{0\}$.
      \end{enumerate}

  These three items imply~\eqref{th:l-t-g:internal:IH1p}
  and~\eqref{th:l-t-g:internal:IH2p}. To
  establish~\eqref{th:l-t-g:internal:IH1p}, take $g(\vec x)
  \in \terms(\Psi)$ such that $\ZZ g \cap \module_f(\Psi)
  \neq \{0\}$. From~\eqref{p-adic-B} we must have $p \div
  g(\vec \nu(\vec x))$. Hence, ${\ZZ f \cap \ZZ h \neq
  \{0\}}$ by~\eqref{p-adic-zero}, and from~\eqref{p-adic-A}
  we get $v_p(h(\vec \nu(\vec x))) = v_p(f(\vec \nu(\vec
  x)))$. For~\eqref{th:l-t-g:internal:IH2p}, take $h(\vec x)
  \in \sfterms{f}{\Psi}$ such that $\ZZ h \cap
  \module_f(\Psi) = \{0\}$. By definition of
  $\module_f(\Psi)$, $\ZZ h \cap \ZZ f = \{0\}$ and so $p
  \nmid h(\vec \nu(\vec x))$ by~\eqref{p-adic-zero}.
  From~\eqref{p-adic-B}, $v_p(h(\vec \nu(\vec x))) = 0$. We
  conclude the base case by
  establishing~\eqref{p-adic-zero}--\eqref{p-adic-B}.\\[6pt]
  \inproof{\eqref{p-adic-zero}} Since~$\Psi$ has the
      elimination property, $f \in \terms(\Psi)$.
      Then,~\eqref{p-adic-zero} follows directly
      from~\eqref{apx:th:l-t-g:internal:IH2}; remark that
      $S(f,h)=0$ is equivalent to $\ZZ f \cap \ZZ h \neq
      \{0\}$.\\[6pt]
  \inproof{\eqref{p-adic-A}} By $\ZZ f \cap \ZZ h \neq
  \{0\}$ there are $\lambda_1,\lambda_2 \in \ZZ \setminus
  \{0\}$ such that $\lambda_1 \cdot f = \lambda_2 \cdot h$.
  Without loss of generality, $\gcd(\lambda_1,\lambda_2) =
  1$, and thus $\gcd(\lambda_2,\gcd(f)) = \lambda_2$. The
  polynomial $f$ is primitive, hence $\lambda_2 = 1$ and we
  get $h = \lambda_1 \cdot f$. Since $p \not \in
  \pzero(\Psi)$, from Condition~\ref{pzero:2} and
  $\lambda_1 \div \gcd(h)$ we derive $p \nmid \lambda_1$.
  Therefore, $v_p(h(\vec \nu(\vec x))) = v_p(\lambda_1 \cdot
  f(\vec \nu(\vec x))) = v_p(f(\vec \nu(\vec x)))$.\\[6pt]
  \inproof{\eqref{p-adic-B}} Trivially, $p \nmid h(\vec
  \nu(\vec x))$ equals $v_p(h(\vec \nu(\vec x))) = 0$. To
  show ${\ZZ h \cap \module_f(\Psi) = \{0\}}$, first note
  that $\ZZ h \cap \ZZ f = \{0\}$, directly from~$p \div
  f(\vec \nu(\vec x))$ and~\eqref{p-adic-A}. \emph{Ad
  absurdum}, assume $\ZZ h \cap \module_f(\Psi) \not=
  \{0\}$. Since $\Psi$ is increasing for $\vec \chi
  \coloneqq (X_1 \incord \dots \incord X_r)$, and $\lv(h)$
  and $\lv(f)$ are both in $X_1$, $\Psi$ is increasing no
  matter the order of the variables imposed on $X_1$. Take
  an~order $(\incord^\prime) \in \vec \chi$ for which
  $\lv_{\incord^\prime}(h) \, \incordeq^\prime \,
  \lv_{\incord^\prime}(f)$, and let ${x_1^\prime
  \incord^\prime \dots \incord^\prime x_d^\prime}$ be the
  order for the variables $x_1,\dots,x_d$. Since $\Psi$ is
  increasing for~$\incord^\prime$, $\module_f(\Psi) \cap
  \ZZ[x_1^\prime,\dots,x_{\lv_{\incord^\prime}(f)}^\prime] =
  \ZZ f$. However, $\ZZ h \subseteq
  \ZZ[x_1^\prime,\dots,x_{\lv_{\incord^\prime}(f)}^\prime]$
  by definition of $\incord^\prime$, hence from $\ZZ h \cap
  \module_f(\Psi) \not= \{0\}$ we obtain $\ZZ h \cap \ZZ f
  \not= \{0\}$, a contradiction. This
  proves~\eqref{p-adic-B}. 

  \item[induction step.] Let us assume that
  $b_{p,1},\dots,b_{p,k}$ are defined for the variables
  $y_1,\dots,y_{k}$ with $k \in [0,j-1]$, so that the
  induction hypotheses hold. We provide the value
  $b_{p,k+1}$ for~$y_{k+1}$ while
  keeping~\eqref{th:l-t-g:internal:IH1p}
  and~\eqref{th:l-t-g:internal:IH2p} satisfied. We divide
  the proof into two cases, depending on whether there is a
  term $g \in \terms(\Psi)$ with $\lv(g) = y_{k+1}$ such
  that $\ZZ g \cap \module_f(\Psi) \neq \{0\}$.

  \begin{description}
    \item[case $g$ does not exist.] 
      In this case,~\eqref{th:l-t-g:internal:IH1p} is
      fulfilled no matter the value of $b_{p,k+1}$, so we
      focus on finding such a value
      satisfying~\eqref{th:l-t-g:internal:IH2p}. It suffices
      to consider the system  
      \begin{align*}
        h(b_{p,1},\dots,b_{p,k},\,y_{k+1}) &\not\equiv 0 && \pmod p &
        h \in \sfterms{f}{\Psi} \text{ s.t.}~\lv(h) = y_{k+1}.
      \end{align*}
      Similarly to the system in~\Cref{eq:case-p-alphai},
      writing $c_h + a_h \cdot y_{k+1}$ for
      $h(b_{p,1},\dots,b_{p,k},\,y_{k+1})$, we obtain the
      equivalent system of non-congruences 
      \begin{align*}
        y_{k+1} &\not\equiv -a_h^{-1} \cdot c_h && \pmod p &
        h \in \sfterms{f}{\Psi} \text{ s.t.}~\lv(h) = y_{k+1}.
      \end{align*}
      Since $p \not\in \pzero(\Psi)$ and
      from~\ref{pzero:1}, this system admits a solution
      $b_{p,k+1}$ in $[0,p-1]$. Note
      that~\eqref{th:l-t-g:internal:IH2p} is satisfied,
      since every polynomial in that hypothesis is
      considered in these non-congruence systems. 
    \item[case $g$ exists.] Recall that $g$ is a polynomial
    in $\terms(\Psi)$ such that $\lv(g) = y_{k+1}$ and $\ZZ
    g \cap \module_f(\Psi) \neq \{0\}$. Let $u \coloneqq
    v_p(f(\vec \nu(\vec x)))$. In order to
    satisfy~\eqref{th:l-t-g:internal:IH1p} it suffices to
    find $b_{p,k+1} \in \ZZ$ satisfying the following
    (non-empty) system of non-congruences 
    \begin{align*}
      &\forall g \in \terms(\Psi) \text{ s.t.}~\lv(g) = y_{k+1} \text{ and } \ZZ g \cap \module_f(\Psi) \neq \{0\},\\
      &\qquad\begin{aligned}
        g(b_{p,1},\dots,b_{p,k},\,y_{k+1}) &\equiv 0 && \pmod {p^u}\\
        g(b_{p,1},\dots,b_{p,k},\,y_{k+1}) &\not\equiv 0 && \pmod {p^{u+1}}.
        \end{aligned}
    \end{align*}
    Similarly to the system in~\Cref{eq:case-p-alphai},
      writing $c_g + a_g \cdot y_{k+1}$ for
      $g(b_{p,1},\dots,b_{p,k},\,y_{k+1})$, we obtain the
      equivalent system of non-congruences 
    \begin{align}
      \label{eq:l-to-g:hard-system}
      &\forall g \in \terms(\Psi) \text{ s.t.}~\lv(g) = y_{k+1} \text{ and } \ZZ g \cap \module_f(\Psi) \neq \{0\},\\
      &\notag\qquad\begin{aligned}
        y_{k+1} &\equiv -a_g^{-1} \cdot c_g && \pmod {p^u}\\
        y_{k+1} &\not\equiv -a_g^{-1} \cdot c_g && \pmod {p^{u+1}}.
        \end{aligned}
    \end{align}
    Focus on the congruences $y_{k+1} \equiv -a_g^{-1} \cdot
    c_g \pmod {p^u}$ of this system. These only have a
    solution if the right-hand side is the same modulo $p^u$
    for every $g \in \terms(\Psi)$ with $\lv(g) = y_{k+1}$
    and $\ZZ g \cap \module_f(\Psi) \neq \{0\}$. We prove
    that this is indeed the case. Consider~$g_1$ and $g_2$
    such that $g_i \in \terms(\Psi)$ with $\lv(g_i) =
    y_{k+1}$ and ${\ZZ g_i \cap \module_f(\Psi) \neq
    \{0\}}$, for~${i \in \{1,2\}}$. Let $\lambda_1$ and
    $\lambda_2$ be the smallest positive integers such that
    both $\lambda_1 \cdot g_1$ and $\lambda_2 \cdot g_2$
    belong to~$\module_f(\Psi)$. By definition of
    divisibility module and \mbox{$S$-polynomial},
    ${S(\lambda_1 \cdot g_1, \lambda_2 \cdot g_2) \in
    \module_f(\Psi) \cap \ZZ[x_1,\dots,x_d,y_1,\dots,y_k]}$.
    According to the elimination property of $\Psi$, there
    is a (finite) basis $B$ for $\module_f(\Psi) \cap
    \ZZ[x_1,\dots,x_d,y_1,\dots,y_k]$ such that for every
    ${h \in B}$, $f \div h$ is a divisibility in $\Psi$.
    Moreover, $\lv(h) \incordeq y_k$ and thus
    by~\eqref{th:l-t-g:internal:IH1p} we get $v_p(h(\vec
    \nu(\vec x),b_{p,1},\dots,b_{p,k})) = u$. Now, since
    $S(\lambda_1 \cdot g_1, \lambda_2 \cdot g_2)$ is a
    linear combination of elements in $B$, we conclude that
    $p^u \div S(\lambda_1 \cdot g_1, \lambda_2 \cdot g_2)$.
    By writing $g_i(\vec x, y_1,\dots,y_{k+1})$ as
    $g_i'(\vec x, y_1,\dots,y_k) + a_i \cdot y_{k+1}$, for
    $i \in \{1,2\}$, this divisibility can be rewritten as
    the congruence:
    \[ 
      (\lambda_2 \cdot a_{2}) \cdot (\lambda_1 \cdot g_1') \equiv (\lambda_1 \cdot a_1) \cdot (\lambda_2 \cdot g_2') \pmod {p^u}.
    \]
    From~$p \not \in \pzero(\Psi)$,~\ref{pzero:2}
    and~\ref{pzero:3}, we conclude that $p \nmid \lambda_1
    \cdot \lambda_2 \cdot a_1 \cdot a_2$. By multiplying
    both sides of the above congruence by the inverse
    $(\lambda_1 \cdot \lambda_2 \cdot a_1 \cdot a_2)^{-1}$
    of $\lambda_1 \cdot \lambda_2 \cdot a_1 \cdot a_2$
    modulo $p^u$, we conclude that $a_{1}^{-1} \cdot g_1'
    \equiv a_2^{-1} \cdot g_2' \pmod {p^u}$. This shows that
    the right-hand side is the same across all the
    congruences and non-congruences of the system
    in~\Cref{eq:l-to-g:hard-system}. Moreover, $p >
    \card{\terms(\Psi)}$ by~\ref{pzero:1}, and therefore
    this system is feasible, and more precisely has a
    solution $b_{p,k+1}$ of the form $b_{p,k+1} \coloneqq
    p^u \cdot \gamma$ for some $\gamma \in [1,p-1]$. Pick
    such a solution, which by construction
    satisfies~\eqref{th:l-t-g:internal:IH1p}.
     
    We show that $b_{p,k+1}$ also
    satisfies~\eqref{th:l-t-g:internal:IH2p}. Here is where
    the existence of the polynomial ${g \in \terms(\Psi)}$
    satisfying $\lv(g) = y_{k+1}$ and $\ZZ g \cap
    \module_f(\Psi) \neq \{0\}$ plays a role. From $\ZZ g
    \cap \module_f(\Psi) \neq \{0\}$ and since $\Psi$ has
    the elimination property, we can find a polynomial~$g_0$
    such that $f \div g_0$ is in $\Psi$, and $\lv(g_0) =
    y_{k+1}$. We prove~\eqref{th:l-t-g:internal:IH2p}
    arguing by contraposition. Let $h \in \sfterms{f}{\Psi}$
    such that $\lv(h) = y_{k+1}$ and $p \div h(\vec \nu(\vec
    x),b_{p,1},\dots,b_{p,k+1})$. If~$S(h,g_0)$ is zero,
    i.e.,~$h$ and $g_0$ are linearly dependent, then $\ZZ h
    \cap \module_f(\Psi) \neq \{0\}$ follows by definition
    of $g_0$, and~\eqref{th:l-t-g:internal:IH2p} holds for
    $h$. Suppose that~$S(h,g_0)$ is non-zero. From the
    construction of $b_{p,k+1}$ and since~$g_0$ is a
    polynomial considered in~\Cref{eq:l-to-g:hard-system},
    we have $p \div g_0(\vec \nu(\vec
    x),b_{p,1},\dots,b_{p,k+1})$. Then, by definition of
    $S$-polynomial, $p \div S(h,g_0)(\vec \nu(\vec
    x),b_{p,1},\dots,b_{p,k})$. By definition of
    $\sfterms{f}{\Psi}$, note that $h \in \sfterms{f}{\Psi}$
    and $g_0 \in \terms(\Psi)$ implies $S(h,g_0) \in
    \sfterms{f}{\Psi}$. Since $S(h,g_0)$ is non-zero, the
    induction hypothesis~\eqref{th:l-t-g:internal:IH2p}
    implies that $\ZZ S(h,g_0) \cap \module_f(\Psi) \neq
    \{0\}$. Then, $\ZZ h \cap \module_f(\Psi) \neq \{0\}$
    follows directly from the fact that $f \div g_0$ appears
    in $\Psi$ (and so $\ZZ g_0 \cap \module_f(\Psi)$). Once
    more, we conclude that~\eqref{th:l-t-g:internal:IH2p}
    holds for $h$.
  \end{description}
  \end{description}
\end{description} 
Following the case analysis above, we construct solutions
$\vec b_p$ for $\Psi'(\vec y)$ modulo $p$, for every $p \in
\pzero(\Psi')$. This concludes the proof
of~\Cref{claim:new-primes-are-ok}.
\end{proof}

\section{\Cref{theorem:local-to-global}: proof
of~\Cref{eq:gamma-inductive-bound}}
\label{appendix:l-t-g:bounds-on-gamma}

We recall that $\underline{O} \in \pZZ$ is the minimal
positive integer greater or equal than $4$ such that the map
${x \mapsto \underline{O}(x+1)}$ upper bounds the linear
functions hidden in the $\bigO{.}$ appearing
in~\Cref{lemma:add-elimination-property}. 
The integer~$\Gamma(r,\ell,w,m,d)$, with $r,\ell,w,m,d \in \pZZ$ and $r \leq d$,
is the maximum bit length of the minimal positive solution
of any system of divisibility constraints $\Phi$ such that:
\begin{itemize}
  \item $\Phi$ is $r$-increasing.
  \item The maximum bit length of a coefficient or constant
  appearing in $\Phi$, i.e., $\maxbl{\Phi}$, is at most
  $\ell$.
  \item For every $p \in \pdiff(\Phi)$, consider a solution
  $\vec b_p$ of $\Phi$ modulo $p$ minimizing $\mu_p
  \coloneqq \max\{{v_p(f(\vec b_p))} : f \text{ is in the
  left-hand side of a divisibility in } \Phi \}$. Then,
  $\log_2\left(\prod_{p \in \pdiff(\Phi)} p^{\mu_p + 1}
  \right) \leq  w$.
  \item $\Phi$ has at most $m$ divisibilities.
  \item $\Phi$ has at most $d$ variables.
\end{itemize}
Since we want to find an upper bound for $\Gamma$, assume
without loss of generality that $\Gamma(r,\ell,w,m,d)$ is
always at least $\min(\ell,w)$. Let us prove~\Cref{eq:gamma-inductive-bound}.

\EqGammaInductiveBound*

\proofparagraph{Analysis on~$\Gamma(1,\ell,w,m,d)$} 
This case corresponds to the base case of the main
induction, where the solutions are found thanks to the
system of congruences in~\Cref{eq:l-t-g:base-case}, where
for $p \in \pdiff(\Phi)$, $\mu_p \coloneqq \max \{
v_p(f(\vec b_p)) : f \text{ is in the left-hand side of a
divisibility of } \Phi \}$. From the Chinese remainder
theorem, this system of congruences has a solution where
every variable is in $[1,\prod_{p \in \pdiff(\Phi)}
p^{\mu_p+1}]$. Therefore, every variable is bounded by
$2^{w}$ by definition of $w$, and therefore its bit length
is bounded by $w + 3$, since $\bitlength{x} = 1 +
\ceil{\log_2(|x|+1)} \leq \ceil{\log_2(|x|)} + 2 \leq
\log_2(|x|) + 3$, and $w$ is positive.



\proofparagraph{Analysis on~$\Gamma(r,\ell,w,m,d)$ with $r \geq 2$} 
This case corresponds to the induction step of the main
induction, where the solutions are found thanks to the
system of (non)congruences in~\Cref{eq:l-to-g:non-cong-sys}.
At the start of the induction, we add the elimination
property to $\Phi$. According
to~\Cref{lemma:add-elimination-property}, we obtain a system
$\Psi$ with $n \leq m \cdot (d+2)$ divisibilities and
$\maxbl{\Psi} 
\leq \underline{O}(m^3d + 1) \cdot \log_2((d+1) (m+\norminf{\Phi}+2))+3$. We find solutions $\vec b_p$ for $\Psi$ modulo $p$,
for every $p \in \pzero(\Psi)$. For $p \in \pdiff(\Phi)$,
these are the solutions $\vec b_p$ for $\Phi$ modulo $p$
stated in the hypothesis of the theorem. For $p \in
\pzero(\Psi) \setminus \pdiff(\Phi)$, we compute $\vec b_p$
as a solution for $\Phi$ modulo $p$, taken such that for
every $f$ left-hand side of a divisibility in $\Phi$,
$v_p(f(\vec b_p)) = 0$. The existence of such a solution is
guaranteed by~\Cref{lemma:simple-primes}, and as discussed
when presenting the procedure the vector~$\vec b_p$ is a
solution for $\Psi$ modulo $p$ such that for every $f$
left-hand side of a divisibility in $\Psi$, $v_p(f(\vec
b_p)) = 0$. As usual, given $p \in \pzero(\Psi)$, let $\mu_p
\coloneqq \max \{ v_p(f(\vec b_p)) : f \text{ is in the
left-hand side of a divisibility of } \Psi \}$. 

Suppose that the set $X_1 = \{x_1,\dots,x_{d'}\}$ of
variables considered in this step is ordered as ${x_1
\incord \dots \incord x_{d'}}$ (with $d' \leq d$). Recall
that the values assigned to these variables are chosen
inductively, starting with $x_1$ and following the
order~$\incord$. Let $\vec \nu$ be the map computed in this
way. Given $k \in [0,d-1]$, at the $(k+1)$-th iteration we
defined the set $P_k$ as 
\[
    P_k \coloneqq \left\{ p \in \PP : p \in \pzero(\Psi) \text{ or there is } h \in S(\sterms(\Psi)) \backslash \{0\} \text{ s.t.}~\lv(h) \incordeq x_k \text{ and } p \mid h(\vec  \nu(x_1,\dots,x_k)) \right\},
\]
and added to it the smallest prime not in $\pzero(\Psi)$, if
the above definition yields $P_k = \pzero(\Psi)$.

For simplicity, below let $s \coloneqq
\card{S(\sterms(\Psi))}$, $t \coloneqq
\norminf{S(\sterms(\Psi))}$ and $\wprime \coloneqq
\log_2(\prod_{p\in\pzero(\Psi)}p^{\mu_p+1})$, which are all
at least $1$. 

Inductively on $k \in [0,d-1]$, we show that
$\log_2(\vec \nu(x_{k+1})) \leq B$ where 
\[ 
  B \coloneqq C \cdot (\log_2(C))^3 
  \quad \text{and} \quad 
  C \coloneqq 
  2^4 \cdot \wprime \cdot s^3 \cdot 
  \big(5 + \log_2\log_2(t \cdot (d+1))\big)^2.
\]
Therefore, $\bitlength{\vec \nu(x_{k+1})} \leq B+3 \leq
2^{18} \cdot s^4 \cdot \big(5 + \log_2\log_2(t \cdot
(d+1))\big)^3 \cdot \wprime \cdot (\log_2(\wprime)+2)^3$,
where this last inequality follows from a straightforward
computation together with the fact that $(\log_2(x))^3 \leq
5 \cdot x$ for every $x \geq 1$. Note that we do not
simplify $(\log_2(\wprime+2))^3$ into $5 \cdot (\wprime+2)$,
as this would yield an exponentially worse bound for
$\Gamma(r,\ell,\eta,m,d)$ later on.

\begin{description}
  \item[base case $k = 0$.] 
    In this case, $P_0 = \pzero(\Psi) \cup \{p\}$ where $p$
    is the smallest prime not in $\pzero(\Psi)$. Then,
    $\card{P_0} = \card{\pzero(\Psi)} + 1$. We bound $\vec
    \nu(x_1) \in \pZZ$ by applying~\Cref{thm:mixed-crt} to
    the system of (non)congruences
    in~\Cref{eq:l-to-g:non-cong-sys}. We get:
    \begin{align*}
      \vec \nu(x_1) 
      &\leq \Big(\prod_{p\in\pzero(\Psi)}p^{\mu_p+1}\Big) \cdot \big((s+1) \cdot \card{(P_0 \setminus \pzero(\Psi))}\big)^{4 \cdot (s+1)^2 (3 + \ln \ln (\card{(P_0 \setminus \pzero(\Psi))}+1))}\\
      &\leq \Big(\prod_{p\in\pzero(\Psi)}p^{\mu_p+1}\Big) \cdot (s+1)^{12 \cdot (s+1)^2}
    \end{align*}
    Therefore, $\log_2(\vec \nu(x_1)) \leq \wprime + 12 \cdot
    (s+1)^2 \log(s+1)$.
  \item[induction step $k \geq 1$.] 
    Let us first bound $\card{(P_k \setminus
    \pzero(\Psi))}$. By definition, 
    \[ 
      P_k \setminus \pzero(\Psi) = \{ p \in \PP \setminus \pzero(\Psi) : \lv(h) \incordeq x_k \text{ and } p \div h(\vec \nu(x_1,\dots,x_k))\text{ for some }h \in S(\sterms(\Psi)) \setminus \{0\}\}.
    \]
    By induction hypothesis, for every $h \in
    S(\sterms(\Psi))$, $\abs{h(\vec \nu(x_1,\dots,x_k))}
    \leq (k \cdot 2^B + 1) \cdot t$, and therefore
    $\card{(P_k \setminus \pzero(\Psi))} \leq s \cdot
    \log_2((k \cdot 2^B + 1) \cdot t) \leq s \cdot
    \log_2(2^B \cdot t \cdot (d+1))$. Note that $s \cdot
    \log_2(2^B \cdot t \cdot (d+1)) \geq 1$, hence this
    bound on $\card{(P_k \setminus \pzero(\Psi))}$ already
    capture the case where one prime had to be added to
    $P_k$ in order to make this set different form
    $\pzero(\Psi)$. We bound $\vec \nu(x_1) \in \pZZ$ by
    applying~\Cref{thm:mixed-crt} to the system of
    (non)congruences in~\Cref{eq:l-to-g:non-cong-sys}:
    \begin{align*}
      \vec \nu(x_{k+1}) 
        &\leq \Big(\prod_{p\in\pzero(\Psi)}p^{\mu_p+1}\Big) \cdot \big((s+1) \cdot \card{(P_k \setminus \pzero(\Psi))}\big)^{4 \cdot (s+1)^2 (3 + \ln \ln (\card{(P_k \setminus \pzero(\Psi))}+1))}\\
        &\leq \Big(\prod_{p\in\pzero(\Psi)}p^{\mu_p+1}\Big) \cdot \big((s+1)^2 \cdot\log_2(2^B t \cdot (d+1))\big)^{4 \cdot (s+1)^2 (3 + \ln \ln (1+ s \cdot \log_2(2^B t \cdot (d+1))))}.
    \end{align*}
    Then, a simple analysis using properties of logarithms
    shows that  
    $\log_2(\vec \nu (x_{k+1}))$ is at most  
    \begin{align*} 
      &2^4 \cdot \wprime \cdot s^3 \cdot 
      \big(5 + \log_2\log_2(t \cdot (d+1))\big)^2 \cdot (\log_2(B))^2\\
      =\,& C \cdot (\log_2(B))^2
      &\text{definition~of~$C$.}\\ 
      \leq\,& B,
    \end{align*}


    where the latter inequality holds from the fact that,
    whenever $C \geq 45$, every element~$x_i$ of the
    recurrence relation ${\big(x_0 = C,\ x_{i+1} = C \cdot
    (\log_2(x_{i}))^2\big)}$ is bounded by $C \cdot
    (\log_2(C))^3$, i.e., $B$.

\end{description}

We have established that the bit length of the solutions for
the variables in $X_1$ can be bounded with $B+3$. Next, we
want to bound $B+3$ using the arguments of $\Gamma$. To do
so, we first derive upper bounds for $s$, $t$ and $\wprime$.
For $s$ and $t$, from~\Cref{lemma:bound-sterms} we obtain $s
\leq 8 \cdot m^4 \cdot (d+2)^6$ and $\log_2(t) \leq 2 \cdot
(d+2) \cdot (\maxbl{\Phi}+1)+1$. For $\wprime$, we have 
{\allowdisplaybreaks
\begin{flalign*}
  \wprime 
  & \leq \log_2\Bigl( \prod_{p \in \pzero(\Psi)}p^{\mu_p + 1} \Bigr)\\
  & \leq \log_2\Bigl( \prod_{p \in \pzero(\Psi) \setminus \pdiff(\Phi)}p^{\mu_p + 1} \cdot \prod_{p \in \pdiff(\Phi)}p^{\mu_p + 1} \Bigr)\\
  & \leq \log_2\Bigl( \prod_{p \in \pzero(\Psi) \setminus \pdiff(\Phi)}p^{\mu_p + 1} \Bigr) + w\\
  & \leq \log_2\Bigl( \prod_{p \in \pzero(\Psi) \setminus \pdiff(\Phi)}p \Bigr) + w 
  & \hspace{-2cm}\mu_p = 0 \text{ for all } p \not\in \pdiff(\Phi)\\
  & \leq \log_2\Bigl( \prod_{p \in \pzero(\Psi)}p \Bigr) + w\\
  & \leq 64 \cdot n^5 (d+2)^4 (\maxbl{\Psi}+2) + w 
  & \text{by~\Cref{lemma:bound-on-pzero}}\\
  & \leq 64 \cdot (m \cdot (d+2))^5 (d+2)^4 (\underline{O}(m^3d + 1) \cdot \log_2((d+1) (m+\norminf{\Phi}+2))+5) + w\\
  & \leq 128 \cdot \underline{O} \cdot m^9 (d+2)^{11} \cdot (\ell + w).
\end{flalign*}
Then, $B+3$ is bounded as follows:
\begin{flalign*}
    B+3 & \leq 2^{18} \cdot s^4 \cdot \big(5 + \log_2\log_2(t \cdot (d+1))\big)^3 \cdot \wprime \cdot (\log_2(\wprime)+2)^3\\
    & \leq 2^{30} \cdot m^{16} (d+2)^{24}\big(5 + \log_2\log_2(t \cdot (d+1))\big)^3 \cdot \wprime \cdot (\log_2(\wprime)+2)^3
    & \text{bound on } s 
    \\
    & \leq 2^{38} \cdot m^{16} (d+2)^{25}  (1+\log_2(\maxbl{\Psi}+1))^3 \cdot \wprime \cdot (\log_2(\wprime)+2)^3
    & \text{bound on $\log_2(t)$}
    \\ 
    & \leq 2^{54} \cdot m^{17} (d+2)^{26} \log_2( \underline{O})^3 \cdot (2 + \log_2(\ell))^3 \cdot \wprime \cdot (\log_2(\wprime)+2)^3 & \text{bound on $\maxbl{\Psi}$}
    \\ & \leq 2^{104} \cdot m^{27} (d+2)^{38} \underline{O} \cdot \log_2( \underline{O})^6  \cdot (\ell + w) \cdot (\log_2(\ell + w))^6 & \text{bound on $\wprime$}.
\end{flalign*}
}

The procedure continues by recursively computing a positive
integer solution for the formula $\Phi'(\vec y) \coloneqq
\Phi\substitute{\vec \nu(\vec x)}{x : x \in X_1}$, which is
$s$-increasing for some $s \leq r-1$. In the recursion, the
procedure uses solutions $\vec b_p$ for $\Phi'$ modulo $p$
for every $p \in \pdiff(\Phi')$, computed according
to~\Cref{claim:phi-prime:new-primes-are-ok}. Hence, to
conclude the analysis on $\Gamma$, it suffices to find
positive integers $\ell', w', m', d'$ such that $\Phi'$ is
one of the formulae considered for $\Gamma(r-1, \ell', w',
m', d')$. Let us bound these integers:

\begin{itemize}
  \item $\Phi'$ has fewer variables and divisibilities than
  $\Phi$, therefore we can choose $m' = m$ and $d' = d$.
  \item The coefficients of the variables in the polynomials
  of $\Phi'$ are all from $\Phi$, therefore their bit-length
  is bounded by $\ell$. Let us bound the constants of the
  polynomials in $\Phi'$. These constants have the form
  $f(\vec \nu (\vec x))$ with $f$ being a polynomial with
  coefficients and constant bounded from~$\Phi$. So,
  $\maxbl{f(\vec \nu (\vec x))} \leq \bitlength{2^B \cdot
  \norminf{\Phi} \cdot d + \norminf{\Phi}}$, and from the
  bounds on $B+3$ we can set
  \begin{align*}
    \ell' = 2^{105} \cdot m^{27} (d+2)^{38} \underline{O} \cdot \log_2( \underline{O})^6 \cdot (\ell + w) \cdot (\log_2(\ell + w))^6.
  \end{align*}


  \item Let $\mu_p \coloneqq \max\{{v_p(f(\vec b_p))} : f
  \text{ is in the left-hand side of a divisibility in }
  \Phi' \}$. Thanks
  to~\Cref{claim:phi-prime:new-primes-are-ok}, if $p \in
  \pzero(\Psi)$, then $\mu_p =  \max\{{v_p(f(\vec b_p))} : f
  \text{ is in the left-hand side of a divisibility in }
  \Psi \}$, and otherwise if $p \not\in \pzero(\Psi)$, then
  $\mu_p$ is the $p$-adic valuation of a constant left-hand
  side of $\Phi'$. We derive the following bound
  on~$\log_2\Bigl( \prod_{p \in \pdiff(\Phi')}p^{\mu_p + 1}
  \Bigr)$, which yields a value for~$w'$:
  \begin{align*}
    & \log_2\Bigl( \prod_{p \in \pdiff(\Phi')}p^{\mu_p + 1} \Bigr)\\
    =\,& \log_2\Bigl( \prod_{p \in \pdiff(\Phi') \setminus \pzero(\Psi)}p^{\mu_p+1} \Bigr) 
    + \log_2\Bigl( \prod_{p \in \pdiff(\Phi') \cap \pzero(\Psi)}p^{\mu_p + 1} \Bigr) \\ 
    \leq \,& \log_2\Bigl( \prod_{p \in \pdiff(\Phi') \setminus \pzero(\Psi)}p^{\mu_p} \Bigr) 
    + \log_2\Bigl( \prod_{p \in \pdiff(\Phi') \setminus \pzero(\Psi)}p \Bigr) 
    + \log_2\Bigl( \prod_{p \in \pzero(\Psi)}p^{\mu_p + 1} \Bigr) \\ 
    \leq \,& \log_2\Bigl( \prod_{\substack{\alpha \text{ constant and}\\ 
    \text{left-hand side~in $\Phi'$}}} \alpha \Bigr) 
    + \log_2\Bigl( \prod_{p \in \pdiff(\Phi')}p \Bigr) 
    + \wprime 
    &\text{from~\Cref{claim:phi-prime:new-primes-are-ok}}\\
    \leq \,& m \cdot \maxbl{\Phi'}
    + \log_2\Bigl( \prod_{p \in \pdiff(\Phi')}p \Bigr) 
    + \wprime\\
    \leq \,&m \cdot \maxbl{\Phi'}
    + m^2(d+2)(\maxbl{\Phi'}+2)
    + \wprime
    &\text{from~\Cref{lemma:bound-on-pzero}}\\
    \leq \,& 2^{109} \cdot m^{29} (d+2)^{39} \underline{O} \cdot \log_2( \underline{O})^6 \cdot (\ell + w) \cdot (\log_2(\ell + w))^6 = w'.
  \end{align*}
\end{itemize}

Note that since the bound we obtained for $\ell'$ is greater
than $B+3$, the value 
\[\Gamma(r-1, \ 2^{104} \cdot m^{27} (d+2)^{38}
 \underline{O} \cdot \log_2( \underline{O})^6 \cdot (\ell +
 w) \cdot (\log_2(\ell + w))^6, \ w', \ m, \ d)\] bounds not
 only the bit length of the minimal positive solution of
 $\Phi'$, but also of the solutions assigned to variables
 in~$X_1$. This concludes the proof of~\Cref{eq:gamma-inductive-bound}.